\documentclass[]{rsos}

\newtheorem{theorem}{\bf Theorem}[section]

\newtheorem{definition}{\bf Definition}[section]
\newtheorem{assumption}{\bf Assumption}[section]

\usepackage[numbers]{natbib}
\usepackage{physics}
\usepackage{bbm}
\usepackage[ruled,vlined,linesnumbered]{algorithm2e}
\newcommand{\bs}[1]{\boldsymbol{#1}}
\newcommand\independent{\protect\mathpalette{\protect\independenT}{\perp}}
\def\independenT#1#2{\mathrel{\rlap{$#1#2$}\mkern2mu{#1#2}}}

\DeclareMathOperator*{\argmin}{arg\,min}
\SetKwComment{Comment}{// }{} 
\SetKw{KwDownTo}{downto}
\usepackage{multirow}
\usepackage{booktabs}
\usepackage{subcaption}
\usepackage[final,protrusion=true,expansion=true]{microtype} 

\begin{document}

\title{Data-driven sparse identification of governing PDEs via knockoff filters and multi-criteria trade-offs}

\author{
Pongpisit Thanasutives$^{1}$, Naichang Ke$^{2}$ and Yoshinobu Kawahara$^{1, 2}$}

\address{$^{1}$RIKEN Center for Advanced Intelligence Project (AIP)\\
$^{2}$The University of Osaka}

\subject{Computer science, Mathematics}


\corres{Pongpisit Thanasutives\\
\email{pongpisit.thanasutives@riken.jp}}

\begin{abstract}
	We propose KO-PDE-IDENT, a data-driven framework for identifying parsimonious partial differential equations (PDEs) with false discovery rate (FDR) control. PDE discovery from noisy observations is often hindered by extreme multicollinearity among candidate terms, which causes typical sparse-regression methods to select spurious terms. To address this problem, KO-PDE-IDENT initially mines a support set of potential candidate terms via model-X knockoff filters with finite-sample FDR control, then refines and ranks the surviving PDE alternatives. The framework integrates three components. First, knockoff feature statistics are constructed by coupling $\ell_{0}$-constrained adaptive best-subset selection with SHapley Additive exPlanations (SHAP), yielding an effective and computationally efficient difference statistic. Second, a recursive feature elimination (RFE) procedure removes terms whose marginal contributions are dispensable and assesses statistical necessity through knockoff-perturbed hypothesis testing. Third, the final model selection is formulated as a multi-criteria decision-making (MCDM) problem, where the optimal governing equation is the alternative that best balances a wide range of criteria such as predictive accuracy, model complexity and coefficient uncertainty. We evaluate KO-PDE-IDENT on five canonical PDEs under severe noise corruption. Empirical results show that our framework can exactly recover the true PDE structure, eliminating false discoveries while retaining all true underlying terms, with low coefficient estimation error.
\end{abstract}

\begin{fmtext}
\end{fmtext}


\maketitle

\section{Introduction}

Sparse identification of nonlinear dynamics (SINDy) has emerged as a data-driven paradigm for discovering governing equations---commonly in the form of partial differential equations (PDEs) identified directly from observed data~\cite{brunton2016sindy,rudy2017data,schaeffer2017learning}, which is of interest in this paper. As an alternative to deriving physics solely from first principles, this approach harnesses machine learning alongside rich, high-fidelity observational data to uncover hidden physical laws in complex dynamical systems where traditional derivations are intractable. The paradigm has since been extended through coordinate-discovering autoencoders~\cite{champion2019datadriven}, parallel implicit variants~\cite{kaheman2020sindypi}, weak formulations that mitigate noise amplification~\cite{reinbold2020using,messenger2021weak}, and bootstrap-based ensemble methods~\cite{fasel2022ensemble}.

Despite remarkable advances, discovering parsimonious and physically accurate models from noisy, high-dimensional data remains challenging. Typical PDE identification frameworks construct large candidate libraries comprising polynomial interactions and higher-order spatial derivatives. These libraries inherently exhibit extreme multicollinearity, rendering sparse regression algorithms susceptible to observational noise and prone to selecting spurious terms. Heuristic thresholding methods such as sequential thresholded least-squares~\cite{rudy2017data} and relaxed-regularised regression~\cite{zheng2019sr3,champion2020unified}, together with cross-validation and information criteria~\cite{mangan2017model,chen2008extended}, are widely employed to promote sparsity, yet they lack guiding statistical guarantees against false discoveries. Bayesian formulations with sparsifying priors~\cite{hirsh2022sparsifying,north2022bayesian,galioto2020bayesian} provide posterior uncertainty over coefficients, and deep symbolic-numeric architectures~\cite{long2018pdenet,long2019pdenet2} learn flexible structural representations end-to-end; however, neither delivers finite-sample control of the false discovery rate (FDR). The recent application of stability selection to PDE learning~\cite{maddu2022stability} bounds the expected number of false positives but, as in classical stability selection~\cite{meinshausen2010stability}, does not control the FDR itself. Consequently, practitioners have no mathematically rigorous assurance that the discovered governing equation captures genuine physical mechanisms rather than artefacts of noise. Controlling FDR, on its own, cannot guarantee recovery of the true governing PDE, as this also depends on the candidate library's expressivity (whether we have an overcomplete library) and the signal-to-noise ratio (i.e., sufficient observational data quality). It would, however, provide the first explicit, finite-sample bound on the proportion of false discoveries, transforming PDE identification from a heuristic search/optimisation into a principled statistical inference procedure.

The fundamental limitation of existing approaches is that they attempt to identify the governing PDE mainly via optimisation lens: minimise some loss function over the full candidate library and hope that the selected terms are the correct ones. When the library is highly multicollinear---as it inevitably is for overcomplete polynomial-derivative bases---this strategy conflates weak true signals with correlated spurious terms, and no regularisation tuning can provide a formal guarantee on the rate of false discoveries. Inspired by this research gap, we argue that a more systematic approach is needed: one that first extracts an FDR-controlled set of potential candidate terms---analogous to mining gold from ore---before applying model selection to the refined set. This systematic approach would prevent the search from being lost in the vast space of possible candidate terms and make the subsequent model selection more efficient as well as more reliable.

To this end, we propose \textsc{KO-PDE-IDENT} (\underline{K}nock\underline{O}ff-based \underline{PDE} \underline{IDENT}ification), a PDE identification framework built on statistical inference. At its core is the application of model-X knockoff filters~\cite{candes2018modelx,barber2015controlling} to PDE discovery. The intuition is straightforward: for each candidate term in the library, we generate a synthetic ``knockoff'' copy that is statistically indistinguishable from the original in its correlation with every other term, yet conditionally independent of the response given the originals. A knockoff therefore serves as a negative control---a feature that, by construction, has no genuine relationship with the governing dynamics. By comparing the importance score of each original term against that of its knockoff, the filter ranks all candidates and chooses a data-driven threshold that bounds the expected proportion of false discoveries at any specified level, providing finite-sample FDR guarantees~\cite{benjamini1995controlling}. To overcome the randomness inherent in a single knockoff realisation, we aggregate multiple realisations via the e-value-based Benjamini--Hochberg (e-BH) procedure of Ren and Barber~\cite{ren2024derandomised,wang2022evalues}, which derandomises the filter while preserving FDR control. In this paper, we introduce the SHAP-based difference statistic (SHAP-DS), a new feature statistic that couples an $\ell_{0}$-constrained adaptive best-subset selection (\texttt{abess}) solver~\cite{zhu2020polynomial,zhu2022abess} with SHapley Additive exPlanations~\cite{lundberg2017unified,lundberg2020trees}. For constructing knockoffs, two covariance estimators---Graphical Lasso with cross-validated regularisation and the Ledoit--Wolf estimator---are run simultaneously, and the estimator yielding the lower information-criterion score is preferred.

In the presence of extreme multicollinearity, FDR control alone may permit a small fraction of highly correlated spurious terms to persist within the controlled set. To further eliminate these false discoveries, we introduce a recursive feature elimination (RFE) procedure operating in the feature-attribution space. Rather than relying on regression coefficients, our SHAP-based selection removes terms whose marginal contributions are dispensable, then validates the statistical necessity of each retained term via a knockoff-perturbed hypothesis test that compares its penalised least-squares residual against that obtained when its conditionally independent knockoff copy is substituted in. This hypothesis test exploits the conditional-independence properties of model-X knockoffs.

Even with a small, refined FDR-controlled support set, selecting the single ``best'' governing equation remains non-trivial. Different model selection methods~\cite{mangan2017model,bozdogan2000akaike,thanasutives2024ubic} impose different trade-offs between model complexity, predictive accuracy, and coefficient uncertainty, so relying on a single information criterion (e.g., AIC~\cite{akaike1974new,mangan2017model} or BIC~\cite{schwarz1978bic}) is insufficient. We therefore formulate the final selection as a multi-criteria decision-making (MCDM) problem, enabling an optimal balance across a wider range of criteria---an approach previously unexplored in the data-driven PDE discovery literature. More specifically, best-subset selection enumerates all PDE alternatives over the reduced support, and each is assessed against five distinct types of criteria, comprising six criteria altogether (since model complexity is captured by two complementary measures): distributional predictive accuracy via conformal prediction~\cite{vovk2005algorithmic,lei2018distribution,barber2021predictive,kim2020predictive}, model complexity (a structural-complexity score and the robust information complexity criterion RICOMP-MM~\cite{bozdogan2000akaike,guney2021robust}), PDE coefficient uncertainty via recentred influence function (RIF) regression~\cite{firpo2009unconditional}, miscalibration assessed through the Murphy--Ehm decomposition~\cite{murphy1973vector,ehm2016quantiles,gneiting2007strictly}, and a coverage-width criterion~\cite{khosravi2010cwc} that combines interval sharpness with empirical coverage. An ensemble of five MCDM methods---TOPSIS~\cite{hwang1981multiple}, VIKOR~\cite{opricovic2004compromise}, COMET~\cite{salabun2015characteristic}, PROMETHEE-II~\cite{brans1985preference} and CoCoSo~\cite{yazdani2019cocoso}---is aggregated, e.g., via mean normalised preferences, to identify the optimal governing PDE.

To the best of our knowledge, KO-PDE-IDENT is the first framework to apply knockoff filtering and MCDM to data-driven PDE discovery. We validate \textsc{KO-PDE-IDENT} through empirical experiments on canonical PDE systems of varying dimensionality, structural complexity and noise intensity: the 1D Burgers, Korteweg--de Vries (KdV), and Kuramoto--Sivashinsky (KS) equations under 50\% additive noise, a 2D Reaction--Diffusion (RD) system under 10\% noise, and a 3D Gray--Scott (GS) system under 0.1\% noise. Across all five benchmarks, the framework attains an empirical FDR (eFDR) of zero and an empirical power (ePOWER) of one, recovering the true PDE structure with mean coefficient errors below 5\% in 1D and below 1\% in the multi-component (multiple target variables) multi-dimensional systems, consistently outperforming the notable baselines such as STLSQ~\cite{brunton2016sindy}, SR3~\cite{zheng2019sr3} and SSR~\cite{boninsegna2018sparse}. The finite-sample FDR guarantee is provided by the knockoff/e-BH screening stage under the stated model-X assumptions; the subsequent RFE and MCDM serve as empirical refinement and multi-criteria model-selection stages, whose performance is verified through these benchmarks rather than by additional formal guarantees. Through the combination of a rigorous screening stage with the principled empirical refinement that follows, KO-PDE-IDENT advances data-driven PDE discovery from an exercise of loss optimisation and hyperparameter tuning to one in which the candidate set entering model selection carries an explicit, theoretically guaranteed bound on its false discovery proportion---bringing the rigour of high-dimensional inference to scientific discovery, a research domain where it has been conspicuously absent.

\paragraph{Contributions.} Our main contributions are summarised as follows.
\begin{itemize}
	\item \textbf{FDR-controlled PDE discovery via knockoff filters.} We provide, to our knowledge, the first application of the model-X knockoff filter to data-driven PDE identification, equipping the discovery process with a finite-sample FDR guarantee through e-BH-derandomised aggregation of multiple knockoff realisations.
	\item \textbf{Robust feature statistics via abess and SHAP.} We introduce the SHAP-DS feature statistic, which combines $\ell_{0}$-constrained best-subset regression with SHAP attributions to produce flip-sign-compliant statistics at lower computational cost than SWAP-type feature statistics~\cite{gimenez2019knockoffs} while remaining effective against spurious terms (false discoveries).
	\item \textbf{RFE with SHAP selection and knockoff perturbations.} We design an RFE procedure in feature-attribution space that removes terms whose marginal contributions are dispensable, and then validates each retained term via a one-sided Wilcoxon signed-rank test against its conditionally independent knockoff copy.
	\item \textbf{Multi-criteria decision-making for model selection.} We reformulate the final equation choice as an MCDM problem over a best-subset-enumerated candidate set, scoring each candidate on (distributional) conformalised predictive accuracy, structural and RICOMP-MM complexity, RIF-based coefficient uncertainty, Murphy--Ehm miscalibration and a coverage-width criterion, and aggregating the rankings of TOPSIS, VIKOR, COMET, PROMETHEE-II and CoCoSo into an ultimate robust decision on the optimal governing PDE.
\end{itemize}

\section{Methodology}

This section describes how \textsc{KO-PDE-IDENT} identifies the governing equation from noisy observations in three stages. Section~\ref{sec:sindy} casts PDE identification as sparse regression over a weak-form candidate library. Section~\ref{sec:ko} then constructs FDR-controlled support sets via model-X knockoff filters with a SHAP-based feature statistic, and refines them through RFE. Section~\ref{sec:mcdm} performs the final equation selection through a multi-criteria ranking of PDE alternatives, each of which is represented as the best subset. A summary of the principal notation used throughout the paper is provided in Appendix~\ref{app:notation}.

\subsection{Sparse identification of governing PDEs}\label{sec:sindy}
We aim to discover governing PDEs parameterised as follows:

\begin{equation}
    u_{t} = \mathcal{L}(u, \partial_{x}u, \partial_{xx}u, \dots; \bs{\mu}) = \sum_{k=1}^{c}\mathcal{L}_{k}(u, \partial_{x}u, \partial_{xx}u, \dots)\,\mu_{k}.
\label{eq:pde-formulation}
\end{equation}

\noindent We assume the governing equation is parsimonious, in that it is constructed from a small number of important terms whose coefficients are collected in $\bs{\mu}$. The small integer $c$ denotes the true support size. $\mathcal{L}$ is an unknown nonlinear operator that depends on spatial derivatives of the state variable $u$, e.g.,\ $\partial_{x}u$, $\partial_{xx}u$, each of which can be computed, on a spatiotemporal grid, through numerical differentiation.

We formulate the system of linear equations corresponding to~\eqref{eq:pde-formulation}, building an overcomplete candidate library based on the state-variable data $\bs{\mathrm{U}}$, its derivatives and their (nonlinear) interactions as follows:

\begin{equation}
    \bs{\mathrm{U}}_{t} = \bs{\Phi}\bs{\xi} = 
    \begin{pmatrix}
    \vline & \vline & \vline & \vline & \vline\\
    \bs{\phi_{1}} & \cdots & \bs{\phi_{j}} & \cdots & \bs{\phi_{p}}\\
    \vline & \vline & \vline & \vline & \vline\\
    \end{pmatrix}\bs{\xi}
    = \sum_{j=1}^{p} \xi_{j}\bs{\phi_{j}}.
\label{eq:linear-system}
\end{equation}

\noindent Here $\bs{\xi} \in \mathbb{R}^{p}$ is the sparse coefficient vector that extends $\bs{\mu}$ to the $p$-dimensional library: $\xi_{j} = \mu_{k}$ if $j$ corresponds to the $k$-th active term in~\eqref{eq:pde-formulation}, and $\xi_{j} = 0$ otherwise, so $\norm{\bs{\xi}}_{0} = c \ll p$. Recovering exactly the active support of $\bs{\xi}$ together with its nonzero entries thus constitutes successful identification. Throughout the knockoff-based selection that follows, we work under the linear-model assumption $\bs{\mathrm{U}}_{t} = \bs{\Phi}\bs{\xi} + \bs{\eta}$, where $\bs{\eta}$ is Gaussian noise. A common candidate library is, for example, given by

\begin{equation*}
    \bs{\Phi}(\bs{\mathrm{U}}) = 
    \begin{pmatrix}
    \vline & \vline & \vline & \vline & \vline & \vline & \vline\\
    \bs{1} & \bs{\mathrm{U}} & \bs{\mathrm{U}}^{2} & \cdots & \bs{\mathrm{U}}_{x} & \bs{\mathrm{U}}\bs{\mathrm{U}}_{x} & \cdots\\
    \vline & \vline & \vline & \vline & \vline & \vline & \vline\\
    \end{pmatrix},
\end{equation*}

\noindent which consists of candidate terms based on polynomial interactions between $\bs{\mathrm{U}}$ and its spatial derivatives. However, this naive polynomial library is known to yield noise-intolerant PDE discovery~\cite{rudy2017data} and also incurs computational inefficiency at a large number of samples (i.e., the number of rows in $\bs{\Phi} \in \mathbb{R}^{n \times p}$). In this work, we instead employ a weak (integral) formulation of the PDE~\cite{reinbold2020using} to obtain noise-robust representations of the candidate terms (and of the response $\bs{\mathrm{U}}_{t}$) as follows:

\begin{equation}
    \bs{\phi^{i}_{j}} = \int_{\Omega_{i}} wq_{j} d\Omega,\, i = 1, \dots, n_{\Omega};
\label{eq:weak-form}
\end{equation}

\noindent where $q_{j}$ denotes a candidate function, for example, the nonlinear advection $u\partial_{x}u$. $\Omega_{i}$ is the $i$-th subdomain, and $w$ is a smooth weight function (typically vanishing along the boundary $\partial\Omega_{i}$). Integration by parts is then used to transfer the computation of spatial derivatives onto the smooth, noiseless $w$ instead of the noisy field $u$, thereby suppressing the noise amplification associated with direct numerical differentiation. We use the weak formulation of Reinbold et al.~\cite{reinbold2020using} as implemented in the PySINDy library~\cite{kaptanoglu2022pysindy}. Note that denoising is also applied to the observed state data $\bs{\mathrm{U}}$ prior to library construction, which further enhances noise robustness; we apply the block-matching and 3D filtering (BM3D) algorithm of Dabov et al.~\cite{dabov2007bm3d} for one-dimensional PDEs, whose space--time fields are image-like and well suited to two-dimensional image denoising, and one-dimensional Savitzky--Golay filters otherwise.

\subsection{Knockoff filters}\label{sec:ko}
Supposing the overcomplete candidate library $\bs{\Phi}$ has been constructed, we turn our attention to identifying a small subset of predictive terms with the FDR controlled and, ideally, high statistical/selection power. We seek a compact yet overcomplete support---a subset of candidate terms that contains every genuinely relevant term of the unknown governing equation---through model selection procedures with FDR control such as knockoff filters and selective inference~\cite{lee2016exact}. Writing $[p] = \{1, \dots, p\}$ for the index set of the $p$ candidate terms and $\#\{\cdot\}$ for the cardinality of a set, and denoting by $\hat{\mathcal{S}} \subseteq [p]$ the selected support set, the FDR is defined as the expected proportion of falsely selected candidates,

\begin{equation}
    \textrm{FDR} = \mathbb{E}\left[\textrm{FDP}\right] = \mathbb{E}\left[ \frac{\#\left\{ j \in \hat{\mathcal{S}}: \xi_{j} = 0 \right\}}{\#\{j \in \hat{\mathcal{S}}\}} \right];
\label{eq:fdr}
\end{equation}

\noindent where $\hat{\mathcal{S}}$ is the set of filtered candidates, i.e., the support set. $j \in \hat{\mathcal{S}}$ and $\xi_{j} = 0$ together indicate that the $j$-th candidate term is selected but it does not appear in the true governing equation. For convenience, we use $\abs{\hat{\mathcal{S}}}$ and $\#\{j \in \hat{\mathcal{S}}\}$ interchangeably to express the cardinality of the set. The expectation in~\eqref{eq:fdr} is taken over the randomness of the knockoff construction and the data. To evaluate the experimental results, we report for each benchmark instance the \emph{empirical FDR} (eFDR): the realised false discovery proportion (FDP) of the single support set $\hat{\mathcal{S}}$ returned on that instance, that is, the FDP in~\eqref{eq:fdr} evaluated without the expectation. The eFDR is therefore a realised FDP for one instance, not an estimate of the population quantity~\eqref{eq:fdr}; in particular, a realised eFDR may exceed the target level $q$ even when the procedure controls $\textrm{FDR} \leqslant q$, because the guarantee bounds the expected, not realised, FDP. The quality of $\hat{\mathcal{S}}$ can also be evaluated by its statistical power:

\begin{equation}
\textrm{POWER} = \mathbb{E}\left[ \frac{\#\left\{ j \in \hat{\mathcal{S}}: \xi_{j} \neq 0 \right\}}{\#\{j \in [p]: \xi_{j} \neq 0\}} \right],
\label{eq:power}
\end{equation}

\noindent which represents the expected proportion of the $c$ truly relevant terms recovered in $\hat{\mathcal{S}}$ (the expected true-positive rate). As with the FDR, we report for each benchmark instance the \emph{empirical POWER} (ePOWER): the realised true-positive rate of the single support set $\hat{\mathcal{S}}$, that is, the bracketed ratio in~\eqref{eq:power} evaluated without the expectation. POWER in~\eqref{eq:power} is thus the expectation of ePOWER over the randomness of the knockoff construction and the data. Both denominators are well-defined: $c \geqslant 1$ by definition (the parsimonious governing equation~\eqref{eq:pde-formulation} contains at least one term), and $\abs{\hat{\mathcal{S}}} \geqslant 1$ (otherwise we would discover nothing) throughout this paper. We call $\hat{\mathcal{S}}$ overcomplete if it contains every genuinely relevant term, i.e., $\mathcal{H}^{*} \subseteq \hat{\mathcal{S}}$, where $\mathcal{H}^{*} = \{j \in [p] : \xi_{j} \neq 0\} = [p] \setminus \mathcal{H}_{0}$ is the true support set of cardinality $c$ (the complement of the null set $\mathcal{H}_{0}$). The overcompleteness property that we prefer to hold throughout this paper (Definition~\ref{def:overcomplete}) is that the selected support is overcomplete in this sense---intuitively, that we admit no false negatives, possibly at the cost of a few correlated false positives that subsequent stages will prune. For a single instance this is equivalent to $\textrm{ePOWER} = 1$. We stress that overcompleteness is \emph{not} required for the finite-sample FDR guarantee of the knockoff/e-BH screening stage: that guarantee holds under the model-X assumptions regardless of whether all true terms are retained. Overcompleteness is instead an empirical property that we prefer to hold across our benchmarks; nonetheless, in genuine discovery problems, where the governing equation is unknown, it cannot be confirmed in general.

\subsubsection{Constructing knockoffs for FDR control}
Because the candidate terms are represented through weak-form integrals over a collection of randomly sampled subdomains $\Omega_{i}$ (see~\eqref{eq:weak-form}), each row of the resulting candidate matrix $\bs{\Phi}$ corresponds to one such subdomain, i.e., $n$ becomes $n_{\Omega}$. This randomness justifies treating the candidate terms as random variables rather than fixed-X designs. We therefore use model-X knockoffs~\cite{candes2018modelx}---which assume a known (or estimable) conditional distribution of the features---to build the computationally efficient copies $\bs{\tilde{\Phi}}$ for the FDR-controlled selection procedure, as defined below.

\begin{definition}[Model-X knockoffs]
Given the family of random variables $\bs{\phi} = (\phi_{1}, \dots, \phi_{p})$, the MX knockoffs are built with the following properties.

\begin{enumerate}
    \item Pairwise exchangeability: $(\bs{\phi}, \bs{\tilde{\phi}})_{\mathrm{swap}(S)} \stackrel{\textrm{d}}{=} (\bs{\phi}, \bs{\tilde{\phi}})$. The joint distribution of $\bs{\phi}$ and its knockoffs $\bs{\tilde{\phi}}$ remains unchanged when any subset of variables is swapped with its corresponding knockoffs. For example, $(\bs{\phi}, \bs{\tilde{\phi}}) = (\phi_{1}, \phi_{2}, \phi_{3}, \tilde{\phi}_{1}, \tilde{\phi}_{2}, \tilde{\phi}_{3})$ follows the same distribution as $(\bs{\phi}, \bs{\tilde{\phi}})_{\mathrm{swap}(\{1, 3\})} = (\tilde{\phi}_{1}, \phi_{2}, \tilde{\phi}_{3}, \phi_{1}, \tilde{\phi}_{2}, \phi_{3})$. This property indicates that the knockoffs mimic the distribution of the original variables.
    \item Conditional independence: $\bs{\tilde{\phi}} \independent \bs{y} \mid \bs{\phi}$, which means that $\bs{\tilde{\phi}}$ is independent of (or irrelevant to) the response $\bs{y}$ given $\bs{\phi}$. This holds by construction in the MX setting: $\bs{\tilde{\phi}}$ is sampled from a known conditional distribution depending only on $\bs{\phi}$, so the response $\bs{y}$ does not enter the construction of the conditional distribution $\bs{\tilde{\phi}} \mid \bs{\phi}$.
\end{enumerate}

\noindent In this paper, $\bs{y}$ denotes the response vector on the left-hand side of the system of linear equations~\eqref{eq:linear-system}: typically the time-derivative vector $\bs{\mathrm{U}}_{t}$, or, in the weak formulation~\eqref{eq:weak-form} adopted in this paper, the weak-form integral of $u_{t}$ against the test function obtained by integration by parts. A variable $\phi_{j}$ is said to be \emph{null} if and only if $\bs{y}$ is independent of $\phi_{j}$ conditionally on the other variables $\bs{\phi}_{-j} = \{\phi_{1}, \dots, \phi_{p}\} \setminus \{\phi_{j}\}$~\cite{candes2018modelx}. The subset of null variables is denoted by $\mathcal{H}_{0} \subset [p]$, and a variable $\phi_{j}$ is called \emph{non-null} or relevant if $j \notin \mathcal{H}_{0}$; under the linear-model assumption $\bs{y} = \bs{\Phi}\bs{\xi} + \bs{\eta}$, this coincides with $\mathcal{H}_{0} = \{j \in [p]: \xi_{j} = 0\}$.
\end{definition}

\noindent Knockoffs are treated as null terms for the true model because they are constructed to be irrelevant to the response. The relevant terms in the true model should therefore exhibit significantly greater importance than their knockoff counterparts.

For clarity, we collect the assumptions underlying the theoretical guarantees, together with the overcompleteness property used in the empirical claims of this paper.

\begin{assumption}[Model-X feature distribution]\label{ass:mx}
The candidate features $\bs{\phi} = (\phi_{1}, \dots, \phi_{p})$ are random with a known or estimable joint distribution. In this paper, $\bs{\phi} \sim \mathcal{N}(\bs{0}, \bs{\Sigma})$, with $\bs{\Sigma}$ estimated from the data using the Graphical Lasso or Ledoit--Wolf estimator.
\end{assumption}

\begin{assumption}[Linear response model]\label{ass:linear}
The response follows the sparse linear model $\bs{y} = \bs{\Phi}\bs{\xi} + \bs{\eta}$, so that the conditional-independence nulls defined above coincide with $\mathcal{H}_{0} = \{j \in [p] : \xi_{j} = 0\}$.
\end{assumption}

\begin{definition}[Overcompleteness]\label{def:overcomplete}
The selected support $\hat{\mathcal{S}}$ is \emph{overcomplete} if it retains every truly relevant term, $\mathcal{H}^{*} \subseteq \hat{\mathcal{S}}$ (equivalently, $\textrm{ePOWER} = 1$ on a single instance). Overcompleteness is \emph{not} required for the finite-sample FDR guarantees, which hold under Assumptions~\ref{ass:mx} and~\ref{ass:linear} alone; it is a preferred property that we monitor through ePOWER across all benchmarks.
\end{definition}

\paragraph{Gaussian knockoffs.}Suppose $\bs{\phi} \sim \mathcal{N}(\bs{0}, \bs{\Sigma})$. An exemplar joint distribution of the original variables and their knockoffs that satisfies pairwise exchangeability is $(\phi, \tilde{\phi}) \sim \mathcal{N}(\bs{0}, \bs{G})$, where

\begin{equation}
    \bs{G} = 
    \begin{pmatrix}
    \bs{\Sigma} & \bs{\Sigma} - \mathrm{diag}(\bs{d})\\
    \bs{\Sigma} - \mathrm{diag}(\bs{d}) & \bs{\Sigma}
    \end{pmatrix}.
\end{equation}

\noindent $\mathrm{diag}(\bs{d})$ is any diagonal matrix chosen such that $\bs{G}$ is positive semidefinite, $\bs{G} \succeq 0$; this condition is necessary and sufficient for the existence of $\bs{\tilde{\Phi}}$. In the fixed-X paradigm of Barber and Cand{\`e}s~\cite{barber2015controlling}, the mimicked moments are constructed from the empirical Gram matrix: after column normalisation $\norm{\bs{\phi_{j}}}_{2} = 1$, one sets $\bs{\Sigma} = \bs{\Phi}^{\intercal}\bs{\Phi} = \bs{\tilde{\Phi}}^{\intercal}\bs{\tilde{\Phi}}$ and $\bs{\Sigma} - \mathrm{diag}(\bs{d}) = \bs{\Phi}^{\intercal}\bs{\tilde{\Phi}}$, so that $\begin{pmatrix} \bs{\Phi} & \bs{\tilde{\Phi}} \end{pmatrix}^{\intercal}\begin{pmatrix} \bs{\Phi} & \bs{\tilde{\Phi}} \end{pmatrix} = \bs{G}$. In the MX paradigm adopted here, $\bs{\Sigma}$ denotes instead the population covariance (to be estimated) and the knockoffs are drawn/sampled from the following conditional multivariate Gaussian:
\begin{equation}
    \bs{\tilde{\Phi}} \mid \bs{\Phi} \stackrel{\textrm{d}}{=} \mathcal{N}\!\left(\bs{\Phi} - \bs{\Phi}\bs{\Sigma}^{-1}\mathrm{diag}(\bs{d}),\; 2\mathrm{diag}(\bs{d}) - \mathrm{diag}(\bs{d})\bs{\Sigma}^{-1}\mathrm{diag}(\bs{d})\right),
\end{equation}
hence the term ``Gaussian knockoffs.'' For notational convenience, we henceforth write $\bs{X}$ for the augmented matrix $\begin{pmatrix} \bs{\Phi} & \bs{\tilde{\Phi}} \end{pmatrix}$, so that $\bs{X}_{1:p}$ refers to $\bs{\Phi}$, and $\bs{X}_{p+1:2p}$ refers to $\bs{\tilde{\Phi}}$.

In practice, a wide range of covariance estimators can be used; the Ledoit--Wolf and graphical lasso estimators are among the most common choices. Also, several alternative ways of obtaining $\mathrm{diag}(\bs{d})$ exist, giving rise to distinct knockoff sampling algorithms; semidefinite-programming (SDP), equicorrelated, conditional-independence (CI) and minimum-variance reconstructability (MVR) constructions are widely used variants, which we compare in Section~\ref{sec:exp_gen}. For a detailed treatment of these constructions and the broader theory of MX knockoffs, we refer the reader to~\cite{candes2018modelx,barber2020robust}.

\paragraph{FDR control with feature statistics.}The importance scores of candidate terms are used to identify the relevant terms of the governing equation. For each $j$th candidate, we compute a test statistic $W_{j} = W_{j}\left(\bs{X}, \bs{y}\right)$ satisfying the flip-sign property, so that an important term is expected to yield a large positive value of $W_{j}$, providing evidence against the null hypothesis. Exploiting the pairwise exchangeability, the flip-sign property is given by

\begin{equation}
    W_{j}\left(\bs{X}_{\mathrm{swap}(\mathcal{S})}, \bs{y}\right) = 
    \begin{cases}
    -W_{j}\left(\bs{X}, \bs{y}\right); & j \in \mathcal{S},\\
    W_{j}\left(\bs{X}, \bs{y}\right); & j \notin \mathcal{S}.
    \end{cases}
\label{eq:flipsign}
\end{equation}

\noindent We design $W_{j} = f(Z_{j}, \tilde{Z}_{j})$, where $Z_{j}$ and $\tilde{Z}_{j}$ are candidate-importance scores assigned to the $j$-th original feature and its knockoff copy, respectively. The combining function $f$ is taken to be antisymmetric, $f(a, b) = -f(b, a)$, so that swapping a feature with its knockoff flips the sign of $W_{j}$. Constructing the importance scores requires an underlying (sparse) regression estimator from which feature-importance values can be extracted. The classic statistic is the lasso coefficient difference, $W_{j} = Z_{j} - \tilde{Z}_{j} = \abs{\hat{\beta}_{j}(\lambda)} - \abs{\hat{\beta}_{p + j}(\lambda)}$, where $\bs{\hat{\beta}}(\lambda) = \argmin_{\bs{\beta}}\frac{1}{2}\norm{\bs{y} - \bs{X}\bs{\beta}}^{2}_{2} + \lambda\norm{\bs{\beta}}_{1}$ collects the lasso coefficients estimated under $\ell_{1}$-regularisation controlled by $\lambda$. By design, $W_{j}$ is equally likely to take positive or negative values under the null hypothesis: for any threshold $\tau$, the cardinalities $\#\{j \in \mathcal{H}_{0}: W_{j} \geqslant \tau\}$ and $\#\{j \in \mathcal{H}_{0}: W_{j} \leqslant -\tau\}$ are equal in distribution, where $\mathcal{H}_{0} = \{j \in [p]: \xi_{j} = 0\}$ collects the null candidates. This assumption is essential for deriving an approximate upper bound on the FDP below. Choosing the support set as $\hat{\mathcal{S}} = \{j \in [p]: W_{j} \geqslant \tau\}$ with $\abs{\hat{\mathcal{S}}} > 0$, the FDP at fixed threshold $\tau$ is approximated by

\begin{equation*}
    \textrm{FDP}(\tau) = \frac{\#\{j \in \mathcal{H}_{0}: W_{j} \geqslant \tau\}}{\#\{j \in [p]: W_{j} \geqslant \tau\}} \approx \frac{\#\{j \in \mathcal{H}_{0}: W_{j} \leqslant -\tau\}}{\#\{j \in [p]: W_{j} \geqslant \tau\}} \leqslant \frac{\#\{j \in [p]: W_{j} \leqslant -\tau\}}{\#\{j \in [p]: W_{j} \geqslant \tau\}} = \widehat{\textrm{FDP}}(\tau).
\end{equation*}

Let $q$ denote the target false discovery rate. The knockoff filter finds the smallest data-driven threshold $\tau$ that satisfies $\widehat{\textrm{FDP}}(\tau) \leqslant q$. Throughout this section, we use $\hat{\mathcal{S}}$ to denote the support set returned by the knockoff-filtering procedure of interest---whether from a single realisation or aggregated across multiple knockoffs (introduced later). The subsequent stages produce refined outputs, denoted $\hat{\mathcal{S}}_{\textrm{RFE}}$ for the recursive feature elimination stage and $\hat{\mathcal{S}}_{\textrm{MCDM}} = \mathcal{S}^{*}$ for the final equation chosen by the MCDM stage.

\begin{theorem}[False discovery rate control using knockoff filters]
Set the data-driven threshold as
\begin{equation}
    \hat{\tau} = \min\left\{\tau > 0: \frac{\textrm{offset} + \#\{j \in [p]: W_{j} \leqslant -\tau\}}{\#\{j \in [p]: W_{j} \geqslant \tau\}} \leqslant q\right\},\quad \textrm{offset} \in \{0, 1\}.
\label{eq:ko_threshold}
\end{equation}

\noindent Assume the support set $\hat{\mathcal{S}} = \{j \in [p]: W_{j} \geqslant \hat{\tau}\}$ contains at least one candidate. Under Assumptions~\ref{ass:mx} and~\ref{ass:linear}, the false discovery rates of $\hat{\mathcal{S}}$ are controlled as follows:

\begin{equation}
    q \geqslant \begin{cases}
    \textrm{mFDR} = \mathbb{E}\left[\frac{\abs{\hat{\mathcal{S}} \cap \mathcal{H}_{0}}}{\abs{\hat{\mathcal{S}}} + \frac{1}{q}}\right], & \textrm{offset} = 0\quad (\textrm{knockoff}),\\
    \textrm{FDR} = \mathbb{E}\left[\frac{\abs{\hat{\mathcal{S}} \cap \mathcal{H}_{0}}}{\abs{\hat{\mathcal{S}}}}\right], & \textrm{offset} = 1\quad (\textrm{knockoff+}).
    \end{cases}
\end{equation}

\noindent By adding $1$ to the inferred number of negatives, the knockoff+ filter controls the usual FDR rather than the modified FDR (mFDR) and is slightly more conservative than the vanilla knockoff filter. We adopt the knockoff+ filter as the default. The FDR guarantee for knockoff+ is proved in Appendix~\ref{app:ko_proof}.
\end{theorem}

\paragraph{Sparse regression with SHAP feature importance.}In our framework, we adopt SHAP as the feature importance for the following reasons: it is applicable to any sparse-regression estimator without distributional or model-class assumptions, it provides a meaningful, game-theoretic measure of each feature's marginal contribution to the prediction (with the unique additivity, efficiency and symmetry guarantees of Shapley values), and it admits a closed-form expression for linear models that makes it especially fast to compute in our setting. As the underlying sparse-regression estimator we use \texttt{abess}~\cite{zhu2020polynomial,zhu2022abess}, an $\ell_{0}$-constrained best-subset selection method, which is well suited to our framework because it explicitly controls the maximum support size $s_{\max}$, ensuring that the candidate set entering the knockoff filter remains compact even when the candidate library is large. Given $s_{\max}$, \texttt{abess} performs a golden-section search over a plausible range of support sizes and selects the output support set under $3$-fold cross-validated error. Coupling \texttt{abess} with SHAP-based feature-importance scores within the knockoff filter therefore yields a small subset that controls the FDR while retaining high statistical power.

The objective function of the $\ell_{0}$-constrained sparse-regression problem solved by \texttt{abess} is
\begin{equation}
    \min_{\bs{\beta}}\frac{1}{2n}\norm{\bs{y} - \bs{X}\bs{\beta}}^{2}_{2} + \lambda_{\textrm{ridge}}\norm{\bs{\beta}}^{2}_{2}, \textrm{ subject to } 0 \leqslant \norm{\bs{\beta}}_{0} \leqslant s_{\max}.
\end{equation}

\noindent The default value of $\lambda_{\textrm{ridge}}$ is $10^{-5}$, unless stated otherwise. The maximum support size $s_{\max}$ is determined through an initial Elastic Net screening procedure. Specifically, we fit an elastic net (with an $\ell_{1}$ ratio of $0.5$) over a grid of scaling factors $\kappa \in \{0.5, 0.75, 1\}$. For each value of $\kappa$, the regularisation parameter is chosen using a plug-in estimator based on a theoretical bound proportional to $\sqrt{\log(p)/n}$ and further scaled by the empirical noise variance estimated from an ordinary least-squares (OLS) fit. We perform $3$-fold cross-validation to find the best scaling factor then define $s_{\max}$ (for the \texttt{abess} estimator) as the number of nonzero coefficients of the corresponding fitted elastic net. The multitask elastic net is employed instead when discovering a system of PDEs with more than one state variable of interest.

Rooted in coalitional game theory, SHAP~\cite{lundberg2017unified,lundberg2020trees} provides a unified measure of feature importance by distributing a model's prediction among its constituent features based on their marginal contributions. For each input $\bs{x}^{(i)}$ ($i = 1, \dots, n$) and each feature $j$ in the augmented matrix $\bs{X} = \begin{pmatrix} \bs{\Phi} & \bs{\tilde{\Phi}} \end{pmatrix}$, the (interventional) SHAP value $\varphi_{j}(i)$ is the weighted average of the $j$-th feature's marginal contributions to the prediction of $\bs{x}^{(i)}$ across all feature subsets $\mathcal{S}$ that exclude it,

$$
\varphi_{j}(i) = \sum_{\mathcal{S} \subseteq P \setminus \{j\}} \frac{\abs{\mathcal{S}}!\,\bigl(\abs{P} - \abs{\mathcal{S}} - 1\bigr)!}{\abs{P}!} \left[ v\!\left(\mathcal{S} \cup \{j\}; \bs{x}^{(i)}\right) - v\!\left(\mathcal{S}; \bs{x}^{(i)}\right) \right],
$$

\noindent where $P = [2p]$ indexes the $2p$ columns of the augmented matrix $\bs{X}$ (so that $\abs{P} = 2p$), and $v(\mathcal{S}; \bs{x}^{(i)})$ is the model's expected prediction at $\bs{x}^{(i)}$ when only the features in $\mathcal{S}$ are observed and the rest are marginalised under the (assumed independent) feature distribution. We aggregate the per-sample SHAP values into a single per-feature importance score $Z_{j} = \frac{1}{n}\sum_{i=1}^{n} \abs{\varphi_{j}(i)}$, the mean absolute value across the $n$ samples.

Because the \texttt{abess} fit is a linear model, with prediction function $g(\bs{x}) = \bs{x}^{\intercal}\hat{\bs{\beta}}$, the interventional SHAP value admits the closed form $\varphi_{j}(i) = \hat{\beta}_{j}\bigl(X_{ij} - \bar{X}_{j}\bigr)$, where $\bar{X}_{j}$ is the column mean of $\bs{X}_{:,j}$. In particular, $\varphi_{j}(i) = 0$ for every $j \notin \hat{\mathcal{S}}_{\texttt{abess}}$ (the abess-selected support, on which $\hat{\beta}_{j} \neq 0$), so only the abess-selected features carry nonzero SHAP contributions. We exploit this closed form for an efficient computation of $Z_{j}$. The computational complexity (per one knockoff realisation) of SHAP-DS is therefore $O(np)$, cheaper than SWAP-type feature statistics (see Appendix~\ref{app:complexity}).

This formulation naturally yields a valid knockoff feature statistic. Recall from the general construction $W_{j} = f(Z_{j}, \tilde{Z}_{j})$ that an antisymmetric combining function $f(a,b) = -f(b,a)$ is sufficient to guarantee the flip-sign property of $W_{j}$. SHAP-DS adopts the simplest such combiner---the difference, $f(a,b) = a - b$, which is antisymmetric by inspection---and applies it to non-negative SHAP-based importance scores $Z_{j}, \tilde{Z}_{j} \geqslant 0$. Concretely, because the augmented matrix $\bs{X}$ is constructed to satisfy pairwise exchangeability, the SHAP scores $Z_{j}$ (for the original candidate) and $\tilde{Z}_{j}$ (for its knockoff) exchange roles under $\mathrm{swap}(\mathcal{S})$ whenever $j \in \mathcal{S}$. Composing this swap with the antisymmetric difference yields $W_{j} = Z_{j} - \tilde{Z}_{j}$ satisfying the flip-sign property~\eqref{eq:flipsign}, $W_{j}(\bs{X}_{\mathrm{swap}(\mathcal{S})}, \bs{y}) = -W_{j}(\bs{X}, \bs{y})$ whenever $j \in \mathcal{S}$, so SHAP-DS is admissible as a feature statistic for FDR control within the knockoff framework. This argument implicitly requires the underlying fitting procedure to be equivariant under column relabelling and deterministic given the data, so that swapping a feature with its knockoff swaps the corresponding importance scores; both properties hold for \texttt{abess} with fixed initialisation.

\subsubsection{Aggregating multiple knockoff realisations with FDR control}
Because the support set $\hat{\mathcal{S}}$ may be sensitive to randomness in constructing one particular knockoff realisation $\bs{\tilde{\Phi}}$, we aggregate multiple knockoff statistics to produce a single, less randomised support set while preserving FDR control without unduly sacrificing statistical power.

\paragraph{Aggregation via e-values and the e-BH procedure.}Under the null hypothesis, a non-negative random variable $e_{j}$ is called an e-value if $\mathbb{E}\left[e_{j}\right] \leqslant 1$. Ren and Barber~\cite{ren2024derandomised} show that, under a relaxed definition of an e-value, multiple-hypothesis testing via the e-BH procedure of Wang and Ramdas~\cite{wang2022evalues} is equivalent to MX knockoff filtering, in the sense that both procedures yield the same support set. Building on the weaker but sufficient condition $\sum_{j \in \mathcal{H}_{0}}\mathbb{E}\left[e_{j}\right] \leqslant p$, we devise an e-BH procedure that combines averaged e-values, one per candidate term, to yield a single support set $\hat{\mathcal{S}}_{\textrm{ebh}}$ with FDR control. For the $k$-th realisation of the knockoff filter, the e-value is

\begin{equation}
    e^{k}_{j} = p \cdot \frac{\mathbbm{1}\{W^{k}_{j} \geqslant \hat{\tau}^{k}\}}{1 + \#\{j^{\prime} \in [p]: W^{k}_{j^{\prime}} \leqslant -\hat{\tau}^{k}\}},
\end{equation}

\noindent where $\mathbbm{1}\{\textrm{condition}\}$ denotes the indicator function, equal to $1$ when the condition holds and $0$ otherwise. In the numerator the condition is evaluated for the candidate $j$ that indexes $e^{k}_{j}$, whereas the denominator counts over a separate dummy index $j^{\prime}$; this follows the e-value construction of Ren and Barber~\cite{ren2024derandomised}. This definition is in fact inspired by the central assumption underlying the construction of the knockoff statistics: each null $j \in \mathcal{H}_{0}$ is equally likely to have $W_{j} \geqslant \hat{\tau}$ as to have $W_{j} \leqslant -\hat{\tau}$, which leads to the condition $\sum_{j \in \mathcal{H}_{0}}\mathbb{E}\left[e_{j}\right] \leqslant p$, because we have

\begin{align*}
    \sum_{j \in \mathcal{H}_{0}}\mathbb{E}\left[e_{j}\right] = p\sum_{j \in \mathcal{H}_{0}}\mathbb{E}\left[\frac{\mathbbm{1}\{W^{k}_{j} \geqslant \hat{\tau}^{k}\}}{1 + \#\{j^{\prime} \in [p]: W^{k}_{j^{\prime}} \leqslant -\hat{\tau}^{k}\}}\right]
    &\leqslant p\sum_{j \in \mathcal{H}_{0}}\mathbb{E}\left[\frac{\mathbbm{1}\{W^{k}_{j} \geqslant \hat{\tau}^{k}\}}{1 + \#\{j^{\prime} \in \mathcal{H}_{0}: W^{k}_{j^{\prime}} \leqslant -\hat{\tau}^{k}\}}\right]\\
    &= p \cdot \mathbb{E}\left[\frac{\#\{j \in \mathcal{H}_{0}: W^{k}_{j} \geqslant \hat{\tau}^{k}\}}{1 + \#\{j \in \mathcal{H}_{0}: W^{k}_{j} \leqslant -\hat{\tau}^{k}\}}\right] \leqslant p,
\end{align*}

\noindent where the final inequality follows from the Barber--Cand{\`e}s supermartingale lemma~\cite{barber2015controlling,candes2018modelx}: at any stopping time $\hat{\tau}$, we have
\begin{equation*}
    \mathbb{E}\!\left[\frac{\#\{j \in \mathcal{H}_{0}: W_{j} \geqslant \hat{\tau}\}}{1 + \#\{j \in \mathcal{H}_{0}: W_{j} \leqslant -\hat{\tau}\}}\right] \leqslant 1.
\end{equation*}

\noindent The e-BH procedure for derandomising knockoffs operates on the averaged e-value $e^{\textrm{avg}}_{j} = \frac{1}{K}\sum^{K}_{k=1}e^{k}_{j}$, selecting the set of discoveries as 

\begin{equation}
    \hat{\mathcal{S}}_{\textrm{ebh}} = \left\{j \in [p]: e^{\textrm{avg}}_{j} \geqslant \frac{p}{q_{\textrm{ebh}}\hat{\jmath}}\right\},\quad \textrm{where } \hat{\jmath} = \max\left\{j \in [p]: e^{\textrm{avg}}_{(j)} \geqslant \frac{p}{q_{\textrm{ebh}}j}\right\},
\end{equation}

\noindent and $e^{\textrm{avg}}_{(1)} \geqslant \dots \geqslant e^{\textrm{avg}}_{(p)}$ are the order statistics of $e^{\textrm{avg}}_{j}$. We refer the reader to~\cite{ren2024derandomised,wang2022evalues} for a detailed proof that the e-BH procedure applied to the averaged e-values satisfies $\textrm{FDR} \leqslant q_{\textrm{ebh}}$. As recommended by Ren and Barber, we initialise $q_{\textrm{ebh}} = q$ and use Algorithm~\ref{alg:adaptive_fdr} to tune $q_{\textrm{ebh}}$ adaptively, increasing it from $q_{0}$ in steps of $\Delta q$ (up to a maximum target $q_{\max}$), satisfying $s_{\min} \leqslant \abs{\hat{\mathcal{S}}} \leqslant s_{\max}$. For any fixed $q_{\textrm{ebh}}$, imposing the upper bound $s_{\max}$ on the support size within the e-BH rejection rule preserves the finite-sample FDR guarantee at level $q_{\textrm{ebh}}$ (Theorem~\ref{thm:smax_unified}, Appendix~\ref{app:smax_proof}). When $q_{\textrm{ebh}}$ is tuned adaptively during Algorithm~\ref{alg:adaptive_fdr}, the realised level $\hat{q}_{\textrm{ebh}} \leqslant q_{\max}$ (the data-dependent value of $q_{\textrm{ebh}}$ returned by the adaptive search) is itself stochastic, so the clean finite-sample bound is stated at the pre-specified ceiling $q_{\max}$ rather than at $\hat{q}_{\textrm{ebh}}$: because the e-values are fixed (decoupled) and the returned set is e-BH self-consistent at $\hat{q}_{\textrm{ebh}} \leqslant q_{\max}$, it is also self-consistent at the fixed level $q_{\max}$, and self-consistency at a fixed level is what the e-BH theorem requires. Decoupling the e-value construction from the adaptive FDR search also ensures monotonicity of the support sets across target levels (Appendix~\ref{app:smax_proof}, Theorem~\ref{thm:decoupled_ebh}). Starting from a generous $q_{0}$ avoids losing any active term, which is essential for preserving the overcompleteness property (equivalently, maximum power). We use $q_{0} = 0.5$, $q_{\max} = 1$ and $\Delta q = 0.01$ as defaults, which empirically suffice across all benchmarks, and wish to retain at least two active terms, i.e., $s_{\min} = 2$. The formal FDR guarantee applies to the knockoff/e-BH screening rule for each fixed covariance estimator satisfying the model-X assumptions (Assumptions~\ref{ass:mx} and~\ref{ass:linear}); the subsequent choice between the two covariance estimators in $\mathcal{C}$ is an empirical, information-criterion-based selection step (Algorithm~\ref{alg:adaptive_fdr}) and is not itself covered by the model-X/e-BH FDR theorem. When the feature distribution is itself estimated rather than known---as here, through $\bs{\Sigma}$---the robustness analysis of Barber, Cand{\`e}s and Samworth~\cite{barber2020robust} shows that the guarantee degrades gracefully: the FDR of knockoffs sampled from an estimated distribution exceeds the nominal level by at most a quantity controlled by the empirical KL divergence between the estimated and true conditional feature distributions, so accurate covariance estimation translates directly into near-nominal FDR control.

\begin{algorithm}[t]
\caption{$\ell_{0}$-constrained multiple knockoff filters with adaptive FDR control} \label{alg:adaptive_fdr}
\KwIn{
  Initial target FDR $q_{0}$, max target FDR $q_{\max}$, FDR step $\Delta q = 0.01$, candidate library $\bs{\Phi}$, response $\bs{y}$, min support size $s_{\min} = 2$, max support size $s_{\max}$, number of knockoff copies $K$, and collection of covariance estimators $\mathcal{C} = \{\textrm{GraphicalLasso}, \textrm{Ledoit-Wolf}\}$
}
\KwOut{Support set $\hat{\mathcal{S}}$}
\BlankLine
\Comment{Initialise a target FDR for each estimator $cov \in \mathcal{C}$}
$q_{cov} \gets q_{0},\;\forall\,cov \in \mathcal{C}$;\quad $l_{\min} \gets \infty$;\quad $\hat{\mathcal{S}} \gets \emptyset$\\
\Comment{Precompute feature statistics per estimator}
\For{each $cov \in \mathcal{C}$}{
  Generate $K$ MX knockoff copies using $cov$\\
  Compute SHAP-based importance statistics and e-values $\bs{e^{k}_{cov}}$ for each knockoff copy\\
  Aggregate e-values $\bs{e^{\textrm{avg}}_{cov}} \gets \frac{1}{K}\sum_{k=1}^{K}\bs{e^{k}_{cov}}$\\
}
Mark every estimator as \textit{active}\\
\BlankLine
\While{$\exists\,$ active $cov \in \mathcal{C}$}{
  \For{each active $cov \in \mathcal{C}$}{
	\Comment{Aggregating multiple knockoffs with FDR control}
    Apply the size-constrained e-BH procedure on $\bs{e^{\textrm{avg}}_{cov}}$ to select a subset $\mathcal{S}$ controlling the FDR at $q_{cov}$, subject to $\abs{\mathcal{S}} \leqslant s_{\max}$; take the union of all $\mathcal{S}$ for multiple responses\\
    \uIf{$\abs{\mathcal{S}} < s_{\min}$}{
      $q_{cov} \gets q_{cov} + \Delta q$ \Comment{Relax FDR to discover more terms}
      \If{$q_{cov} > q_{\max}$}{Mark $cov$ as \textit{inactive}\\}
    }
    \Else{
      Mark $cov$ as \textit{inactive}\\
      \Comment{May occur when there are multiple response variables}
      \If{$\abs{\mathcal{S}} > s_{\max}$}{
        Truncate $\mathcal{S}$ to at most $s_{\max}$ terms via mixed-integer-optimised sparse regression (MIOSR)\label{line:miosr}\\
      }
      \Comment{i.e., $l$ is an information-criterion-based score}
      Compute cross-validated score $l$ for the linear model fitted to $\left(\bs{\Phi}_{[:,\,\mathcal{S}]},\,\bs{y}\right)$\\
      \If{$l < l_{\min}$}{
        $l_{\min} \gets l$;\quad $\hat{\mathcal{S}} \gets \mathcal{S}$;\quad $q \gets q_{cov}$ \Comment{Track the adaptive FDR}
      }
    }
  }
}
\BlankLine
\Return{$\hat{\mathcal{S}},\, q$} \Comment{Here, $q$ is regarded as the tuned target FDR.}
\end{algorithm}

\noindent The finite-sample FDR guarantee of Theorem~\ref{thm:smax_unified} applies when the size constraint $\abs{\hat{\mathcal{S}}} \leqslant s_{\max}$ is imposed inside the e-BH rejection count, as analysed in Appendix~\ref{app:smax_proof}. The additional MIOSR truncation invoked at line~\ref{line:miosr} in Algorithm~\ref{alg:adaptive_fdr} is an empirical refinement step used only to reconcile the union of per-response support sets when multiple response variables are present; it is not part of the per-response e-BH rejection rule and is therefore not itself covered by the per-response FDR proof.

\subsubsection{Recursive feature elimination via SHAP selection and knockoff-perturbed hypothesis testing}
While the knockoff filter provides finite-sample FDR control, multicollinearity in the overcomplete candidate library can still allow a small fraction of dispensable candidate terms to persist. We therefore further refine the selected support set $\hat{\mathcal{S}}$ via the RFE procedure detailed in Algorithm~\ref{alg:shap_rfe}. Unlike traditional RFE methods that iteratively drop features based on regression coefficients---often unstable for collinear variables---our approach assesses the statistical necessity of each term using SHAP-based selection followed by knockoff-perturbed hypothesis testing.

The RFE procedure operates in two subprocedures. The first subprocedure adopts the SHAP-based subset-selection mechanism of Marc{\'\i}lio and Eler~\cite{marcilio2020shap}. The candidate terms in the support set obtained by aggregating multiple knockoffs (with FDR control) are first ranked in descending order of their mean absolute SHAP values, generating a sequence of nested subsets analogous to those produced by forward regression. To determine the truncation point, we follow the $1$-standard-deviation rule: we evaluate the cross-validated $R^{2}$ of each subset under $3$-fold CV and select the most parsimonious subset within one standard deviation of the peak performance.

Given the SHAP ranking of the remaining terms, the second subprocedure assesses the statistical necessity of each lower-ranked candidate while protecting the most dominant ones (indicated by left elbow detection; see \url{https://github.com/vlavorini/kneefinder} for the implementation) from elimination, thereby mitigating type I error inflation from multiple testing. For each lower-ranked term (processed from least to most important), we apply knockoff-perturbed hypothesis testing as detailed in Algorithm~\ref{alg:shap_rfe}. For a target vulnerable term $j$, we fit a covariance estimator on the current support set and generate $K$ knockoff realisations. For each realisation $k$, we replace the original $j$-th term with its knockoff counterpart and compute the penalised least-squares (Ridge-style) loss of the perturbed system. The distribution of perturbed scores, $\mathcal{E}_{\textrm{swap}}$, is then compared against the baseline score of the unperturbed system, $\mathcal{E}$, to determine whether the term should be removed. Formally, we test the per-feature hypotheses

\begin{equation*}
\begin{aligned}
H_0^{(j)}: & \quad \mathcal{E}_{\mathrm{swap}} \geqslant \mathcal{E}, 
& \text{the original feature performs no worse than its knockoff copy,} \\
H_1^{(j)}: & \quad \mathcal{E}_{\mathrm{swap}} < \mathcal{E}, 
& \text{the knockoff copy outperforms the original feature.}
\end{aligned}
\end{equation*}

\noindent We apply a one-sided Wilcoxon signed-rank test to the differences $\mathcal{E}_{\mathrm{swap}} - \mathcal{E}$. If the resulting $p$-value falls below the prescribed significance level $\alpha$ (set to $0.10$ by default), we reject $H_0^{(j)}$ based on the evidence that the knockoff-substituted system attains a significantly lower penalised residual than the original: substituting the conditionally independent knockoff for feature~$j$ does not impair, and may even improve, the fit. The original term therefore offers no greater predictive power than its synthetic counterpart and should be deleted from $\hat{\mathcal{S}}$, after which the RFE procedure restarts to account for the reduced support set.

\begin{algorithm}[t]
\caption{Recursive feature elimination via SHAP selection and knockoff perturbations} \label{alg:shap_rfe}
\KwIn{FDR-controlled library $\bs{\Phi}_{[:,\, \hat{\mathcal{S}}]}$ obtained by aggregating knockoffs, response $\bs{y}$, significance level $\alpha$, number of knockoff copies $K$, and covariance estimator $cov$}
\KwOut{Support set $\hat{\mathcal{S}}_{\mathrm{RFE}}$}
\BlankLine

\Comment{Subprocedure 1: SHAP selection}
Rank the terms in $\bs{\Phi}_{[:,\, \hat{\mathcal{S}}]}$ by their mean absolute SHAP values in descending order\\
Generate nested subsets $\mathcal{S}_1 \subset \mathcal{S}_2 \subset \dots \subset \hat{\mathcal{S}}$ based on the SHAP ranking\\
Evaluate the cross-validated $R^2$ scores for all subsets\\
Update $\hat{\mathcal{S}}$ to be the smallest subset within 1 standard deviation of the peak $R^2$ score\\
Sort the terms in $\hat{\mathcal{S}}$ in descending order of their SHAP importance\\
\BlankLine

\Comment{Subprocedure 2: Knockoff-perturbed hypothesis testing}
\While{True}{
    Fit models with nested subsets from $\hat{\mathcal{S}}$ then compute the corresponding extended Bayesian information criterion (EBIC)~\cite{chen2008extended} scores\\
    $j_{\mathrm{knee}} \gets \textrm{FindKnee(EBIC)}$;\quad $eliminated \gets \textrm{False}$ \Comment{Knee/Elbow detection}
    
    \Comment{Protect dominant features}
    \For{$j \gets \abs{\hat{\mathcal{S}}}$ \KwDownTo $j_{\mathrm{knee}} + 1$}{
        $\mathcal{E} \gets \textrm{RidgeLoss}(\bs{\Phi}_{[:,\, \hat{\mathcal{S}}]}, \bs{y})$;\quad $\mathcal{E}_{\mathrm{swap}} \gets \emptyset$\\
        
        Fit $cov$ on $\bs{\Phi}_{[:,\, \hat{\mathcal{S}}]}$ then generate $K$ knockoff realisations $\bs{\tilde{\Phi}^{1}}, \dots, \bs{\tilde{\Phi}^{K}}$\\
        
        \For{$k \gets 1$ \KwTo $K$}{
            $\bs{\Phi}_{\mathrm{swap}} \gets \bs{\Phi}_{[:,\, \hat{\mathcal{S}}]}$\\
            Replace the $j$th term in $\bs{\Phi}_{\mathrm{swap}}$ with its knockoff (the $j$-th term from $\bs{\tilde{\Phi}^{k}}$)\\
            Append $\textrm{RidgeLoss}(\bs{\Phi}_{\mathrm{swap}}, \bs{y})$ to $\mathcal{E}_{\mathrm{swap}}$\\
        }
        
        Compute the \textit{p-value} via a one-sided Wilcoxon signed-rank test: $H_0^{(j)}: \mathcal{E}_{\mathrm{swap}} - \mathcal{E} \geqslant \bs{0}$\\

        \Comment{In practice, the decision to remove the $j$th feature whenever $\textit{p-value} < \alpha$ is left to the user; in our implementation, we eliminate the $j$th feature if $\textit{p-value} < \alpha$ holds for at least one $cov$ in $\mathcal{C} = \{\textrm{GraphicalLasso}, \textrm{Ledoit-Wolf}\}$.}
        \If{$\textit{p-value} < \alpha$}{
            Remove the $j$th term from $\hat{\mathcal{S}}$;\quad $eliminated \gets \textrm{True}$ \Comment{Reject $H_0^{(j)}$}
            \textbf{break} \Comment{Continue the vulnerability check with $\hat{\mathcal{S}}$}
        }
    }
    \If{$\neg eliminated$}{
        \textbf{break} \Comment{No vulnerable candidate worth being eiliminated found}
    }
}
\Return{$\hat{\mathcal{S}}_{\mathrm{RFE}} \gets \hat{\mathcal{S}}$}
\end{algorithm}

\subsection{Multi-criteria trade-offs}\label{sec:mcdm}

\begin{algorithm}[t]
	\caption{Recursive multi-criteria ranking of PDE alternatives} \label{alg:mcdm_pde}
	\KwIn{
		Decision matrix $\bs{\mathcal{F}}$ comprising $M$ evaluation criteria for PDE alternatives,\\
		optimisation directions $\bs{o} \in \{-1, +1\}^M$ ($-1$ for $\min$ and $+1$ for $\max$)
	}
	\KwOut{Optimal governing equation support set $\mathcal{S}^{*}$}
	\BlankLine
	\Comment{Preprocessing PDE alternatives}
	Rank the PDE alternatives using an information criterion (EBIC)\\
	Initialise the active set $\mathcal{A}$ by retaining the top $N_{\mathcal{A}}$ alternatives ($N_{\mathcal{A}}=5$ by default)\\
	Construct $\bs{\tilde{\mathcal{F}}}$ by normalising each criterion in $\bs{\mathcal{F}}$ via relative minimum-shifts and/or logarithmic scaling; extract the sub-matrix $\bs{\tilde{\mathcal{F}}}_{\mathcal{A}}$ corresponding to active candidates in $\mathcal{A}$\\
	Compute criteria weights $\bs{w} \gets \text{VarianceWeights}(\bs{\tilde{\mathcal{F}}}_{\mathcal{A}}, \bs{o})$\\
	\BlankLine
	\Comment{Recursive MCDM and rank aggregation}
	Initialise the ranking $\mathcal{W} \gets \emptyset$\\
	\While{$|\mathcal{A}| > 0$}{
		\If{$|\mathcal{A}| = 1$}{
			Append the remaining alternative to $\mathcal{W}$ and \textbf{break}\\
		}
		\Comment{Multi-criteria trade-offs}
		Apply MCDM algorithms (TOPSIS, VIKOR, COMET, PROMETHEE-II, CoCoSo) to $\bs{\tilde{\mathcal{F}}}_{\mathcal{A}}$ using $\bs{w}$; extract the matrix of individual algorithmic rankings $\bs{\mathcal{R}}$\\
		
		\tcc{A simpler alternative to rank aggregation is mean normalised preference aggregation: normalise each method's preferences to $[0,1]$, flip dispreference methods, and select the alternative with the highest mean. Both approaches yield identical final rankings on every benchmark considered.}
		\Comment{Stabilise selection via consensus methods}
		Apply ranking aggregations (Borda count, Schulze, Kemeny--Young, Plurality, and Pairwise Greedy methods) to $\bs{\mathcal{R}}$\\
		Apply another Borda count over these five aggregated rankings, summing positional scores to obtain a composite score for each alternative in $\mathcal{A}$\\
		\Comment{Select winner with deterministic tiebreaker}
		Identify the alternative $a^{*}$ by lexicographic ordering: (i) most rank-1 wins across the five aggregation methods, (ii) highest composite Borda score, (iii) largest support size\\
		Append $a^{*}$ to $\mathcal{W}$ and remove it from $\mathcal{A}$; update $\bs{\tilde{\mathcal{F}}}_{\mathcal{A}}$ to reflect the reduced set\\
	}
	\BlankLine
	\Return{$\mathcal{W}, \mathcal{S}^{*} \gets \text{the first support set in } \mathcal{W}$}
\end{algorithm}

\subsubsection{Constructing PDE alternatives as best-subsets}
After the overcomplete candidate library $\bs{\Phi}$ is pruned by our knockoff filter and the RFE procedure, a reduced set of high-potential terms remains. The final task is then to identify the governing equation among the implied PDE alternatives. Because the reduced support set is sufficiently small, we generate alternatives via best-subset selection (i.e., a brute-force combinatorial search becomes tractable) that enumerates all subsets across support sizes ranging from one to the size of the reduced set. Each subset constitutes a competing PDE alternative, whose attributes, such as predictive accuracy, complexity and uncertainty, are then jointly assessed. We remark that, unlike in forward regression, allowing one more term into the best subset---that is, increasing the size of the best subset by one---does not necessarily correspond to adding one more term to that subset.

\subsubsection{Multi-criteria ranking of PDE alternatives}\label{sec:mcdm_rank}
Identifying the optimal governing equation among the generated alternatives requires navigating intrinsic trade-offs---primarily between model complexity and predictive power, as emphasised by conventional information criteria. To incorporate a broader set of attributes characterising the PDE alternatives, we frame this final model selection as an MCDM problem.

For each alternative, prediction intervals are constructed via conformal prediction~\cite{vovk2005algorithmic,lei2018distribution} using the jackknife+-after-bootstrap technique~\cite{barber2021predictive,kim2020predictive}, with a ridge estimator and an absolute residual conformity score, providing distribution-free coverage guarantees on unseen data. Conformalised quantile regression~\cite{romano2019conformalized} is a viable alternative that likewise retains distribution-free coverage. Standard quantile regression may also be used in place of conformal prediction, albeit without the coverage guarantee; the guarantee is not strictly required, but is convenient for setting the parameters of the coverage-width criterion below. The intervals and associated point estimates are evaluated under repeated $3$-fold cross-validation ($10$ repetitions with different randomisations), yielding the decision matrix whose criteria are defined as follows.

\begin{itemize}
    \item \textit{Distributional predictive accuracy.} Assessed via the relative pinball loss averaged across the $5$th, $50$th and $95$th percentiles, corresponding to $90\%$ nominal coverage. For each quantile, the pinball loss of the $\ell_{2}$-penalised quantile-regression forecast is divided by that of the corresponding null forecast (the unconditional sample quantile), yielding a ratio (similar to $1-\textrm{pseudo-}R^{2}$) in which lower values indicate greater predictive accuracy. Averaging over the three quantiles provides a joint measure of the fidelity of both point and interval predictions.
    \item \textit{Model complexity.} Assessed by two complementary measures: the \emph{structural complexity}
    $$
    \mathrm{SC}(\mathcal{S}) = \sum_{j \in \mathcal{S}}\bigl(\textrm{polynomial degree of term } j + \textrm{derivative order of term } j + 1\bigr),
    $$
    which penalises models not only for having more active terms but also for including high-order derivatives and high-degree polynomials; and the robust information complexity criterion (RICOMP-MM)~\cite{bozdogan2000akaike,guney2021robust}, which penalises ill-conditioned models through a complexity measure derived from the Fisher information matrix. For example, the structural complexity of $\left\{uu_{x}, u_{xx}\right\}$ would be $(2 + 1 + 1) + (1 + 2 + 1) = 8$.
    \item \textit{PDE uncertainty.} Following the uncertainty-penalised model selection of Thanasutives et al.~\cite{thanasutives2024ubic,thanasutives2025ubic}, we use the ratio $\norm{\bs{\hat{\sigma}}^{\textrm{RIF}}_{\textrm{train}}}_{1} / \norm{\bs{\hat{\beta}}^{\textrm{RIF}}_{\textrm{train}}}_{1}$ as a scale-invariant measure of coefficient uncertainty on split training data. This favours alternatives whose coefficients, obtained by heteroscedasticity-robust RIF regression~\cite{firpo2009unconditional}, exhibit lower dispersion relative to their magnitudes---i.e., smaller standard errors relative to the coefficient estimates.
    \item \textit{Miscalibration.} Extracted as the calibration-error component of the Murphy--Ehm decomposition of the pinball scoring function~\cite{murphy1973vector,ehm2016quantiles,gneiting2007strictly}, evaluated at the target quantile levels. This component captures the systematic bias between the predicted quantile and the empirical conditional quantile, isolated from the discrimination (resolution) component of the score. As with distributional predictive accuracy, miscalibration is averaged over the $5$th, $50$th and $95$th percentiles.
    \item \textit{Coverage-width criterion (CWC).} Adapted from Khosravi et al.~\cite{khosravi2010cwc},
    $$
    \textrm{CWC} = \frac{\textrm{NMPIL}}{g_{\sigma}\!\bigl(\eta_{\textrm{cwc}}\,(\textrm{PICP} - \mu_{\textrm{cwc}})\bigr)},
    $$
    which jointly evaluates the empirical coverage probability (PICP) and the normalised mean prediction-interval length (NMPIL) of the prediction intervals constructed from the conformal predictions. The numerator NMPIL is the mean interval width normalised by the range of the observed response, so that lower NMPIL indicates sharper intervals. The denominator activates a sigmoid penalty $g_{\sigma}$ when the empirical coverage falls below a threshold $\mu_{\textrm{cwc}} = 0.85$, slightly below the nominal coverage of $0.9$; the steepness $\eta_{\textrm{cwc}} = 200$ is large because the coverage guarantee is already supplied by conformal prediction. CWC is therefore minimised, and combines the requirements of sharpness and sustained coverage.
\end{itemize}

All six criteria are to be minimised (lower loss, lower structural complexity, lower informational complexity, lower PDE uncertainty, lower miscalibration and lower CWC are better), and the optimisation directions are encoded by the vector $\bs{o}$ in Algorithm~\ref{alg:mcdm_pde}. We normalise the decision matrix to a common scale and compute criterion weights using a variance-based weighting scheme, mitigating subjective bias towards any single metric.

\paragraph{Multi-criteria decision-making methods.}To assess the alternatives, we employ an ensemble of five established MCDM methods, all implemented through \texttt{pymcdm}~\cite{kizielewicz2023pymcdm}. The methods are complementary in how they aggregate criterion values into a single score: TOPSIS~\cite{hwang1981multiple} measures the relative distance of each alternative to an ideal best and an ideal worst point; VIKOR~\cite{opricovic2004compromise} balances overall utility against the worst individual-criterion regret to find a compromise solution; COMET~\cite{salabun2015characteristic} interpolates each alternative's score from a grid of human-interpretable characteristic objects using triangular fuzzy numbers, avoiding rank reversal under alternative changes; PROMETHEE-II~\cite{brans1985preference} performs pairwise outranking by integrating preference functions over criterion differences into a net flow; and CoCoSo~\cite{yazdani2019cocoso} combines a weighted sum, a weighted product, and a compromise mean into an aggregate strategy. Because these methods rest on distinct mathematical formulations of optimality and may occasionally yield conflicting rankings, we stabilise the final selection through rank aggregation. The individual algorithmic rankings are first fused via five consensus methods (the Borda count, Schulze, Kemeny--Young, Plurality and Pairwise Greedy methods), and a further Borda count is then applied over the five aggregated rankings, summing positional scores into a single composite score. The entire procedure is embedded in a recursive elimination loop: at each iteration, the winning alternative is identified by lexicographic ordering---most rank-$1$ wins across the five aggregation methods, then highest composite Borda score, then largest support size---and removed from the pool, after which the remaining alternatives are re-evaluated. The top-ranked alternative is ultimately selected as the optimal governing equation.

\section{Empirical experiments and results}\label{sec:exp}

We evaluate KO-PDE-IDENT through four main experiments on canonical nonlinear PDEs. Each experiment targets one component of the framework: sensitivity to the initial FDR level and the choice of knockoff feature statistic (Section~\ref{sec:exp_fdr}); the knockoff generation method and aggregation procedure (Section~\ref{sec:exp_gen}); structural recovery via RFE (Section~\ref{sec:exp_rfe}); and multi-criteria model selection together with coefficient accuracy (Section~\ref{sec:exp_mcdm}). The Python code will be available for public access at \url{https://github.com/Pongpisit-Thanasutives/KO-PDE-IDENT}.

\subsection{Experimental setup and benchmark PDE systems}\label{sec:exp_setup}

\begin{table}[t]
\centering
\caption{PDE dataset descriptions. The number of discretised spatial and temporal points is specified in the ``Spatial'' and ``Temporal'' columns, respectively.}
\label{tab:pde_datasets}
\resizebox{\textwidth}{!}{
\begin{tabular}{lccc}
\toprule
PDE & Governing equation(s) & Spatial $(x, y, z)$ & Temporal $(t)$ \\
\midrule
Burgers & $\partial_t u = 0.1\partial_x^2 u - u\partial_x u$ & 256 on $[-8, 8]$ & 101 on $[0, 10]$ \\
KdV & $\partial_t u = -\partial_x^3 u - 6u\partial_x u$ & 512 on $[-20, 20]$ & 501 on $[0, 40]$ \\
KS & $\partial_t u = -\partial_x^2 u - \partial_x^4 u - u\partial_x u$ & 1024 on $[0, 32\pi]$ & 251 on $[0, 100]$ \\
\multirow{2}{*}{RD} & $\partial_t u = u - u^3 + v^3 - uv^2 + u^2v + 0.1(\partial_x^2 u + \partial_y^2 u)$ & \multirow{2}{*}{$256 \times 256$ on $[-10, 10] \times [-10, 10]$} & \multirow{2}{*}{201 on $[0, 10]$} \\
 & $\partial_t v = v - u^3 - v^3 - uv^2 - u^2v + 0.1(\partial_x^2 v + \partial_y^2 v)$ & & \\
\multirow{2}{*}{GS} & $\partial_t u = 0.014 - 0.014u - uv^2 + 0.02(\partial_x^2 u + \partial_y^2 u + \partial_z^2 u)$ & \multirow{2}{*}{$128 \times 128 \times 128$ on $[-1.25, 1.25]^3$} & \multirow{2}{*}{100 on $[0, 10]$} \\
 & $\partial_t v = -0.067v + uv^2 + 0.01(\partial_x^2 v + \partial_y^2 v + \partial_z^2 v)$ & & \\
\bottomrule
\end{tabular}
}
\end{table}

\begin{figure}[t]
    \centering
    \begin{subfigure}[b]{0.325\textwidth}
        \centering
        \includegraphics[width=\textwidth]{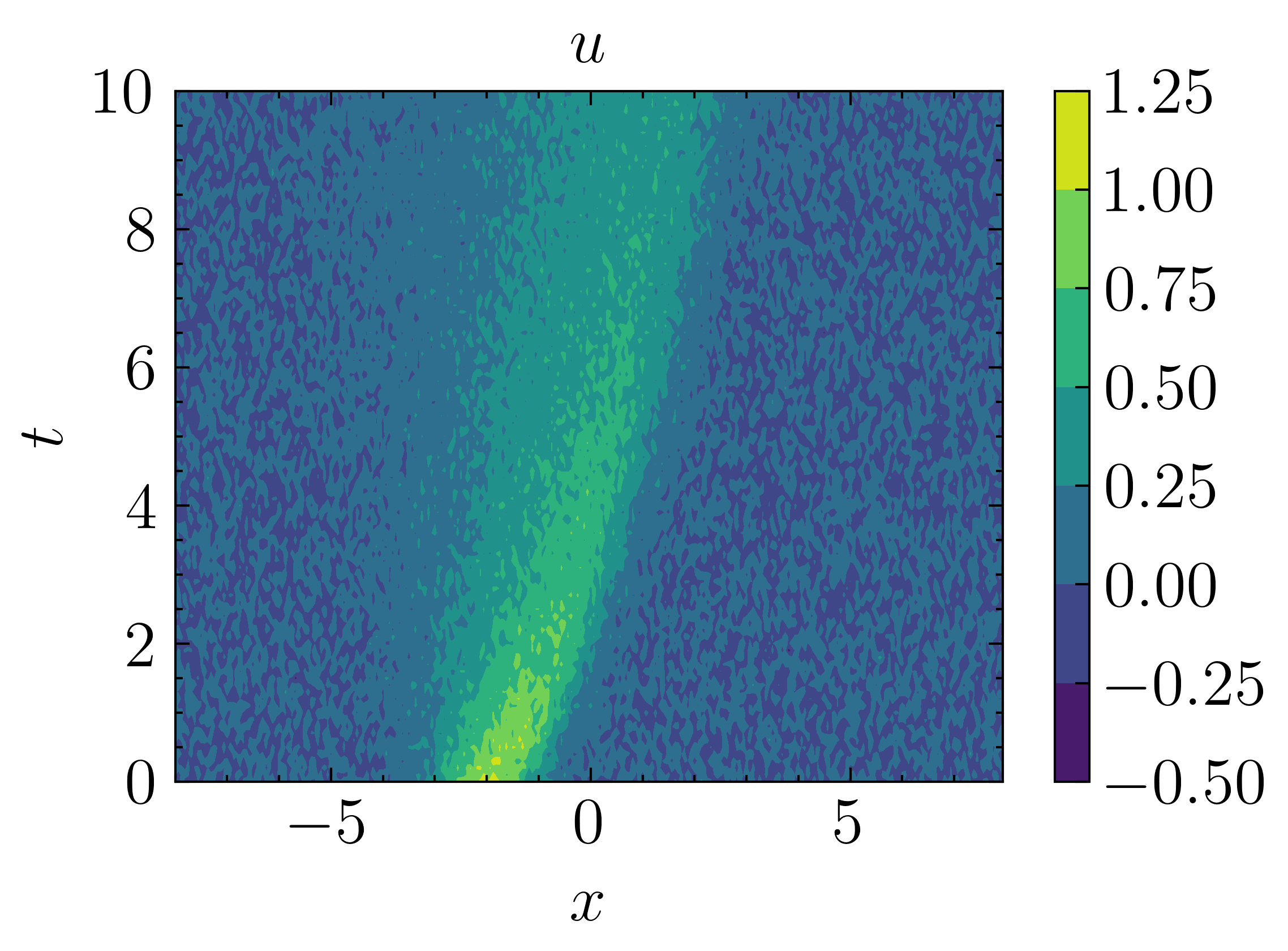}
        \caption{Burgers}
        \label{fig:burgers}
    \end{subfigure}
    \hfill
    \begin{subfigure}[b]{0.325\textwidth}
        \centering
        \includegraphics[width=\textwidth]{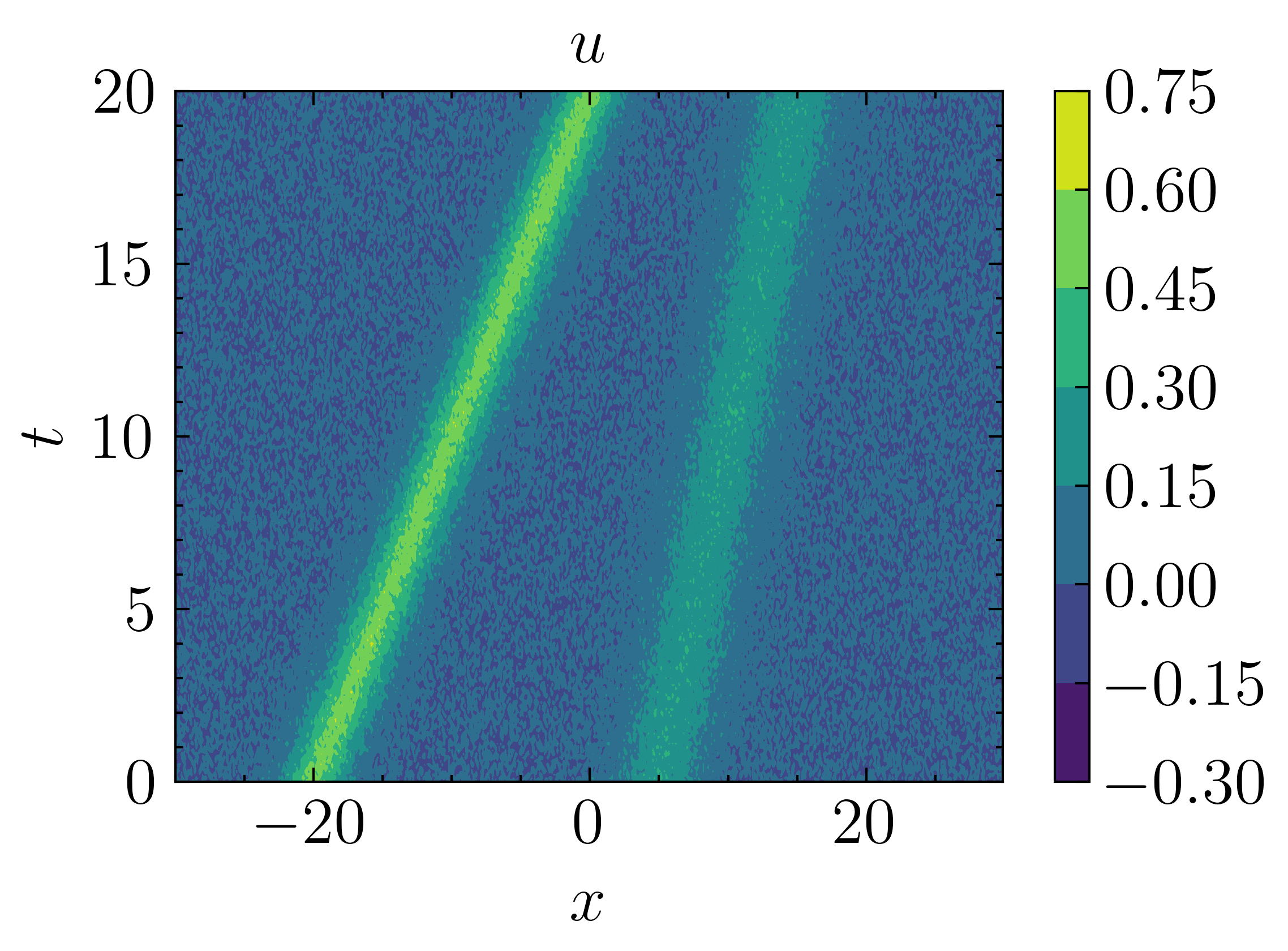}
        \caption{Korteweg--de Vries}
        \label{fig:kdv}
    \end{subfigure}
    \hfill
    \begin{subfigure}[b]{0.325\textwidth}
        \centering
        \includegraphics[width=\textwidth]{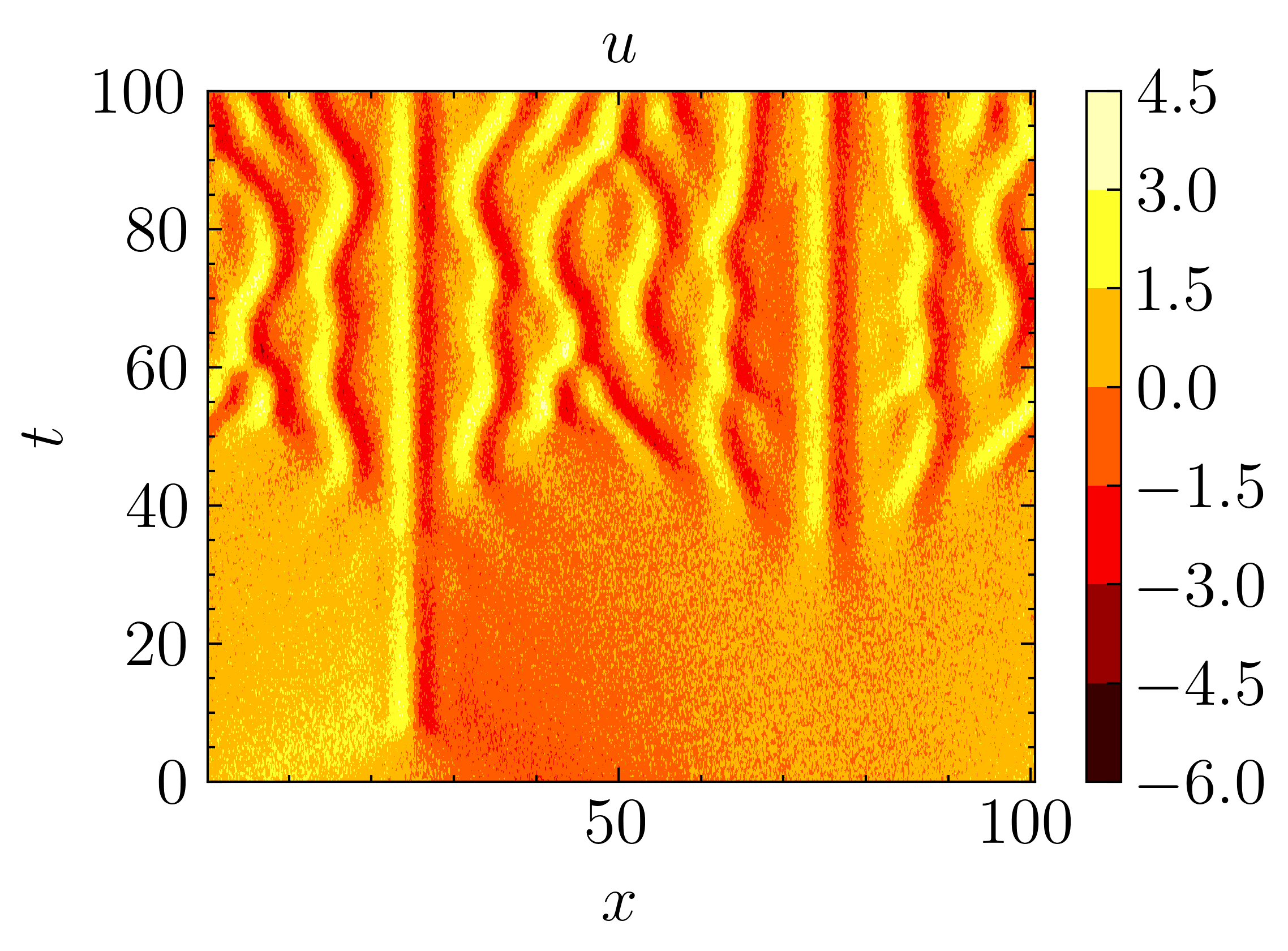}
        \caption{Kuramoto--Sivashinsky}
        \label{fig:ks}
    \end{subfigure}
    \hspace*{\fill}
    \begin{subfigure}[b]{0.39\textwidth}
        \centering
        \includegraphics[width=\textwidth]{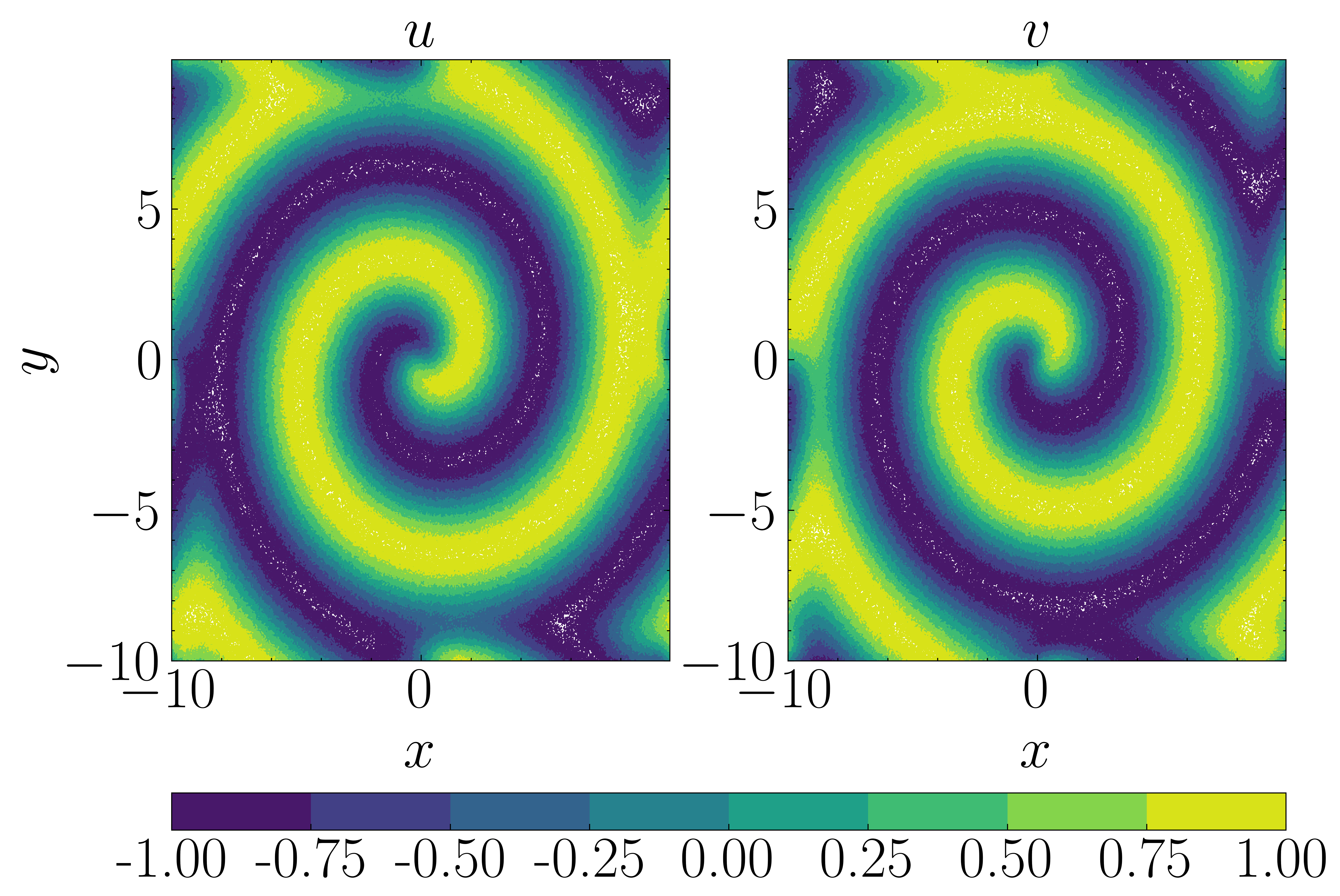}
        \caption{Reaction--Diffusion system}
        \label{fig:rd}
    \end{subfigure}
    \hfill
    \begin{subfigure}[b]{0.39\textwidth}
        \centering
        \includegraphics[width=\textwidth]{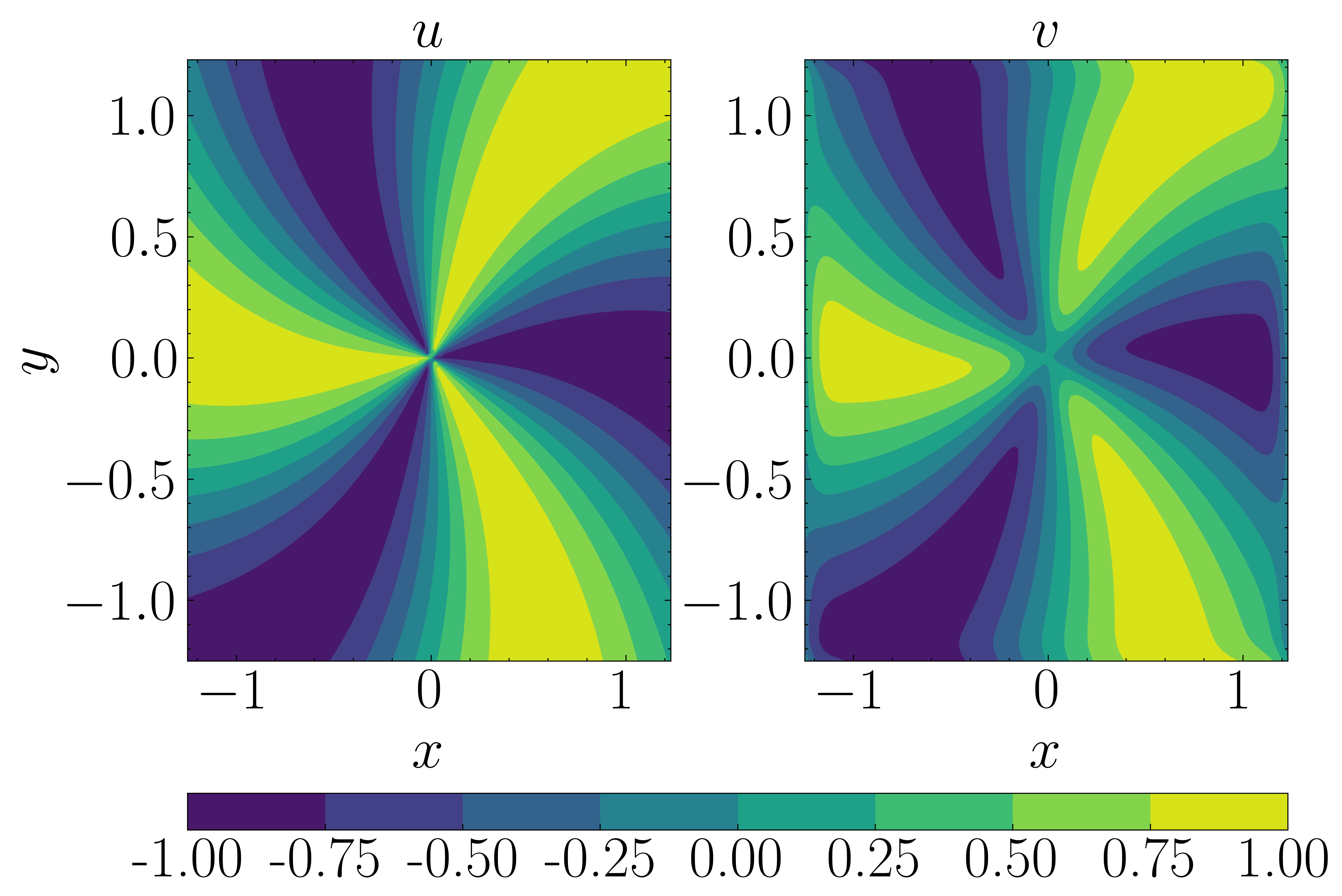}
        \caption{Gray--Scott system}
        \label{fig:dataset5}
    \end{subfigure}
    \hspace*{\fill}
    \caption{Two-dimensional visualisation of noisy state variables for each PDE dataset.}
    \label{fig:pde_datasets}
\end{figure}

We benchmarked the framework on five canonical PDEs spanning a range of dimensionalities and structural complexities: the 1D viscous Burgers' equation, the 1D Korteweg--de Vries (KdV) equation, the 1D chaotic Kuramoto--Sivashinsky (KS) equation, a 2D Reaction--Diffusion (RD) system (seven-term $u$- and $v$-equations with coupled polynomial reactions), and the 3D Gray--Scott (GS) system (five- and six-term coupled equations with cubic interactions). The dataset specifications are summarised in Table~\ref{tab:pde_datasets}. All experiments were run on an Apple MacBook Air M1 with 16~GB of memory.

To test the robustness of the KO-PDE-IDENT framework against noisy observation data, the numerical state variables were corrupted with additive Gaussian noise scaled proportionally to the noise level $\zeta$ and the standard deviation ($\mathrm{sd}$) of the clean data. Mathematically, $\textrm{U}_{ij} \gets \textrm{U}_{ij} + Z$, where $Z \sim \frac{\mathrm{sd}(\bs{\mathrm{U}})\zeta}{100}\mathcal{N}(0, 1)$. The noise levels are $50\%$ for Burgers, KdV and KS; $10\%$ for RD; and $0.1\%$ for GS. The noisy data were denoised using BM3D~\cite{dabov2007bm3d} for the 1D PDEs (whose space--time fields are image-like) and Savitzky--Golay filtering for the higher-dimensional RD and GS systems, prior to library construction. Overcomplete candidate libraries comprising polynomial terms up to degree $6$ and spatial derivatives up to order $6$ were assembled via the weak formulation of Reinbold et al.~\cite{reinbold2020using} as implemented in PySINDy~\cite{kaptanoglu2022pysindy}, with the number of domain centres set to $2000$ for Burgers, KdV and KS. For RD and GS, we considered up to second-order polynomial terms and spatial derivatives, with $10\,000$ domain centres. We employ the SHAP-based difference statistic as the knockoff feature importance, the \texttt{abess} solver as the default estimator, $100$ knockoff realisations for aggregation, and an initial FDR threshold $q_{0} = 0.50$ unless stated otherwise.

\subsection{Sensitivity to the FDR control level and knockoff feature statistics}\label{sec:exp_fdr}

Algorithm~\ref{alg:adaptive_fdr} requires an initial FDR level $q_{0}$ that governs the stringency of the knockoff selection threshold. We assessed the sensitivity of the resulting support set $\hat{\mathcal{S}}$ to the choice of $q_{0} \in \left\{\tfrac{1}{4}, \tfrac{1}{3}, \tfrac{1}{2}\right\}$, with corresponding upper bounds $q_{\max} \in \left\{\tfrac{1}{3}, \tfrac{1}{2}, 1\right\}$, while holding all other algorithmic hyperparameters fixed. Table~\ref{tab:fdr_sensitivity} reports the target FDR $q$ tuned by Algorithm~\ref{alg:adaptive_fdr}, together with the eFDR, ePOWER and support size. We observe that ePOWER equals $1.00$ in every configuration: the knockoff filter never discards a true governing term, regardless of $q_{0}$. Moreover, for four of the five benchmarks (KdV, KS, RD and GS), the selected support set is entirely invariant to $q_{0}$---the same candidates are retained at the same eFDR across all three settings. Only Burgers exhibits any sensitivity: a more generous initialisation ($q_{0} = \tfrac{1}{2}$) admits additional terms, increasing $\abs{\hat{\mathcal{S}}}$ from $4$ to $8$ (with eFDR rising from $0.5$ to $0.75$); the support size nevertheless remains small enough for Algorithms~\ref{alg:shap_rfe} and~\ref{alg:mcdm_pde} to eliminate the false discoveries and recover the sparse governing PDE. These results confirm that the $\ell_{0}$-constrained knockoff filter is largely insensitive to the initial FDR choice: a generous starting value $q_{0} = \tfrac{1}{2}$ safely maximises statistical power while producing a manageably small candidate set for subsequent refinement and model selection.

\begin{table}[htbp]
\centering
\caption{Sensitivity to the initial FDR control level $q_{0}$. We report the eFDR and ePOWER of the selected support set after applying the $\ell_{0}$-constrained multiple knockoff filtering, but before RFE is performed by Algorithm~\ref{alg:shap_rfe}.}\label{tab:fdr_sensitivity}
\small
\begin{tabular}{llcccc}
\hline
PDE & Initial FDRs: $q_{0}, q_{\max}$ & Target FDR $q$ & eFDR & ePOWER & $\abs{\hat{\mathcal{S}}}$ \\
\hline
\multirow{3}{*}{Burgers}
& 1/4, 1/3 & 0.27 & 0.50 & 1.00 & 4 \\
& 1/3, 1/2 & 0.33 & 0.50 & 1.00 & 4 \\
& 1/2, 1 & 0.50 & 0.75 & 1.00 & 8 \\
\hline
\multirow{3}{*}{KdV}
& 1/4, 1/3 & 0.25 & \multirow{3}{*}{0.71} & \multirow{3}{*}{1.00} & \multirow{3}{*}{7} \\
& 1/3, 1/2 & 0.33 &  &  &  \\
& 1/2, 1 & 0.50 &  &  &  \\
\hline
\multirow{3}{*}{KS}
& 1/4, 1/3 & 0.25 & \multirow{3}{*}{0.00} & \multirow{3}{*}{1.00} & \multirow{3}{*}{3} \\
& 1/3, 1/2 & 0.33 &  &  &  \\
& 1/2, 1 & 0.50 &  &  &  \\
\hline
\multirow{3}{*}{RD (u, v)}
& 1/4, 1/3 & (0.25, 0.25) & \multirow{3}{*}{(0.36, 0.36)} & \multirow{3}{*}{(1.00, 1.00)} & \multirow{3}{*}{(11, 11)} \\
& 1/3, 1/2 & (0.33, 0.33) &  &  &  \\
& 1/2, 1 & (0.50, 0.50) &  &  &  \\
\hline
\multirow{3}{*}{GS (u, v)}
& 1/4, 1/3 & (0.25, 0.25) & \multirow{3}{*}{(0.00, 0.00)} & \multirow{3}{*}{(1.00, 1.00)} & \multirow{3}{*}{(6, 5)} \\
& 1/3, 1/2 & (0.33, 0.33) &  &  &  \\
& 1/2, 1 & (0.50, 0.50) &  &  &  \\
\hline
\end{tabular}
\end{table}

The choice of feature-importance statistic $W_{j}$ in the knockoff filter determines how the original and knockoff features are contrasted and may affect eFDR and selection power. We therefore compared our SHAP-based difference statistic (SHAP-DS), which measures the marginal SHAP contribution of each original feature minus its knockoff counterpart from the \texttt{abess} fit, with the swap and swap-integral statistics of Gimenez et al.~\cite{gimenez2019knockoffs}, which are also generic to regression estimators. Recall that we use the $\ell_{0}$-penalised \texttt{abess} as the backbone sparse estimator in Algorithm~\ref{alg:adaptive_fdr}. All three statistics yield identical eFDR, power and support size $\abs{\hat{\mathcal{S}}}$ (see Table~\ref{tab:fdr_sensitivity} for the values at $q = 0.50$) on every PDE system. The choice among them does not affect the FDR guarantee---which holds for any valid flip-sign statistic. Empirically, we observe that they differ only in computational cost, as reported in Table~\ref{tab:stat_comparison}. SHAP-DS is therefore the preferred default: it achieves the same structural recovery at the fastest wall time, with favourable scaling in the number of candidate terms; this advantage becomes particularly pronounced for the higher-dimensional RD and GS systems. A computational complexity analysis of these feature statistics is provided in Appendix~\ref{app:complexity}.

\begin{table}[htbp]
	\centering
	\caption{Wall time (seconds) of three knockoff feature statistics---SHAP-DS, SWAP, and SWAP-INT---at $q = 0.50$. All three statistics yield identical eFDR, power, and $\abs{\hat{\mathcal{S}}}$ on every PDE (see Table~\ref{tab:fdr_sensitivity}). SHAP-DS is the only feature statistic that can be executed in under one minute on the specified computer machine.}\label{tab:stat_comparison}
	\small
	\begin{tabular}{lccc}
		\hline
		PDE & SHAP-DS & SWAP & SWAP-INT \\
		\hline
		Burgers    & 6  & 7  & 67  \\
		KdV        & 6  & 8  & 77  \\
		KS         & 6  & 7  & 74  \\
		RD (u, v)  & 43 & 74 & 452 \\
		GS (u, v)  & 20 & 54 & 380 \\
		\hline
	\end{tabular}
\end{table}

\subsection{Knockoff generators and aggregation methods}\label{sec:exp_gen}

Since the diagonal matrix $\mathrm{diag}(\bs{d})$ determines the joint distribution between the original variables and their knockoff copies, it consequently affects the discriminative power of the resulting feature statistics. We compared four algorithms listed below for constructing $\mathrm{diag}(\bs{d})$.

\begin{itemize}
	\item \textit{Equi} (Equicorrelated): minimises the mean absolute covariance (MAC) between $\bs{\Phi}$ and $\bs{\tilde{\Phi}}$ under the constraint that the correlation between each feature and its knockoff is uniform across all $p$ candidates~\cite{candes2018modelx}.
	\item \textit{SDP}: minimises the MAC via semidefinite programming, without the uniform-correlation constraint~\cite{candes2018modelx}.
	\item \textit{CI} (conditional independence): guarantees that each feature $\phi_{j}$ and its knockoff $\tilde{\phi}_{j}$ are conditionally independent given all the other features $\bs{\phi}_{-j}$; however, CI knockoffs do not always exist, so a heuristic defined in Ke et al.~\cite{ke2020power} is used in general.
	\item \textit{MVR} (minimum variance-based reconstructability): minimises the reconstructability of each original feature from the remaining original features and knockoffs---equivalently, maximises the conditional variance $\mathrm{var}(\phi_{j} \mid \bs{\phi}_{-j}, \bs{\tilde{\phi}})$ of each feature given all the others and the knockoffs---thereby making the knockoffs harder to distinguish from the originals~\cite{spector2022powerful}.
\end{itemize}

Each knockoff generation method was paired with three aggregation procedures for combining the $K = 100$ knockoff realisations: the (parameter-free) e-BH procedure, and the Benjamini--Hochberg procedure with quantile aggregation (BHq) of Nguyen et al.~\cite{nguyen2020aggregation} at quantile levels $\gamma \in \{0.50, 0.75\}$. All other settings---the \texttt{abess} estimator, the SHAP-DS feature statistic, the two covariance estimators in $\mathcal{C}$ and Algorithm~\ref{alg:adaptive_fdr}'s adaptive FDR scheme with $q_{0} = 0.50$---were held fixed. Table~\ref{tab:gen_comparison} reports the eFDR, ePOWER and support size $\abs{\hat{\mathcal{S}}}$ for each combination, evaluated before the RFE stage (Algorithm~\ref{alg:shap_rfe}).

\begin{table}[htbp]
	\centering
	\caption{Comparison of knockoff sampling algorithms and aggregation methods. Four algorithms for constructing the knockoff S-matrix $\mathrm{diag}(\bs{d})$ are paired with three aggregation procedures: the e-BH procedure and Benjamini--Hochberg with quantile aggregation at $\gamma \in \{0.50, 0.75\}$. All results are obtained with $K = 100$ realisations and $q_{0} = 0.50$. ePOWER equals $1.00$ throughout except where marked with $\dagger$. For RD, the $u$- and $v$-equations yield identical results. For GS, the $u$-equation is reported; the Equi entries for GS ($u$) and GS ($v$) reflect the exact support recovered after the MIOSR truncation step in Algorithm~\ref{alg:adaptive_fdr}.}\label{tab:gen_comparison}
	\small
	\setlength{\tabcolsep}{3.5pt}
	\begin{tabular}{ll cc cc cc}
		\hline
		& & \multicolumn{2}{c}{e-BH} & \multicolumn{2}{c}{BHq ($\gamma{=}0.50$)} & \multicolumn{2}{c}{BHq ($\gamma{=}0.75$)} \\
		\cmidrule(lr){3-4}\cmidrule(lr){5-6}\cmidrule(lr){7-8}
		PDE & Generator & eFDR & $\abs{\hat{\mathcal{S}}}$ & eFDR & $\abs{\hat{\mathcal{S}}}$ & eFDR & $\abs{\hat{\mathcal{S}}}$ \\
		\hline
		\multirow{4}{*}{Burgers}
		& Equi & 0.75 & 8  & 0.60 & 5 & 0.50 & 4 \\
		& SDP  & 0.75 & 8  & 0.50 & 4 & 0.50 & 4 \\
		& CI   & 0.50 & 4  & 0.60 & 5 & 0.50 & 4 \\
		& MVR  & 0.50 & 4  & 0.50 & 4 & 0.50 & 4 \\
		\hline
		\multirow{4}{*}{KdV}
		& Equi & 0.71 & 7  & 0.71 & 7 & 0.71 & 7 \\
		& SDP  & 0.71 & 7  & 0.67 & 6 & 0.67 & 6 \\
		& CI   & 0.71 & 7  & 0.71 & 7 & 0.67 & 6 \\
		& MVR  & 0.71 & 7  & 0.71 & 7 & 0.67 & 6 \\
		\hline
		\multirow{4}{*}{KS}
		& Equi & 0.00 & 3  & 0.00 & 3 & 0.00 & 3 \\
		& SDP  & 0.00 & 3  & 0.25 & 4 & 0.00 & 3 \\
		& CI   & 0.25 & 4  & 0.25 & 4 & 0.25 & 4 \\
		& MVR  & 0.25 & 4  & 0.25 & 4 & 0.25 & 4 \\
		\hline
		\multirow{4}{*}{RD ($u, v$)}
		& Equi & 0.36 & 11 & 0.36 & 11 & 0.36 & 11 \\
		& SDP  & 0.30 & 10 & 0.30 & 10 & 0.30 & 10 \\
		& CI   & 0.30 & 10 & 0.30 & 10 & 0.30 & 10 \\
		& MVR  & 0.30 & 10 & 0.30 & 10 & 0.30 & 10 \\
		\hline
		\multirow{4}{*}{GS ($u$)}
		& Equi & 0.00 & 6  & 0.00 & 6  & 0.00 & 6  \\
		& SDP  & 0.45 & 11 & 0.40 & 10 & 0.40 & 10 \\
		& CI   & 0.45 & 11 & 0.45 & 11 & 0.40 & 10 \\
		& MVR  & 0.45 & 11 & 0.40 & 10 & 0.40 & 10 \\
		\hline
		\multirow{4}{*}{GS ($v$)}
		& Equi & 0.00 & 5    & 0.00 & 5    & 0.76$^{\dagger}$ & 17 \\
		& SDP  & 0.55 & 11   & 0.50 & 10   & 0.50 & 10 \\
		& CI   & 0.55 & 11   & 0.55 & 11   & 0.50 & 10 \\
		& MVR  & 0.55 & 11   & 0.50 & 10   & 0.50 & 10 \\
		\hline
		\multicolumn{8}{l}{\footnotesize $^{\dagger}$\,ePOWER $= 0.80$; all other entries achieve ePOWER $= 1.00$.}
	\end{tabular}
\end{table}

The e-BH procedure consistently achieves ePOWER $= 1.00$. BHq at $\gamma = 0.50$ also maintains full power throughout, but BHq at $\gamma = 0.75$ incurs the only power failure in the entire table: on GS ($v$) with the equicorrelated generator, it over-selects $17$ features while dropping ePOWER to $0.80$. Since the e-BH procedure additionally requires no tuning parameter---unlike BHq with quantile aggregation, which depends on a choice of $\gamma$---we adopt it as the default aggregation method.

Although no single knockoff generation algorithm clearly dominates, the equicorrelated construction is a reliable choice overall. On KS, it recovers the exact true support ($\abs{\hat{\mathcal{S}}} = 3$, eFDR $= 0.00$) under both e-BH and BHq ($\gamma = 0.75$), whereas CI and MVR admit one false positive. On GS, the combination of the equicorrelated construction and e-BH achieves eFDR $= 0.00$ on both the $u$- and $v$-equations, because the selection set is large enough to trigger the MIOSR truncation step in Algorithm~\ref{alg:adaptive_fdr}, which trims the support to exactly $s_{\max}$ terms while preserving all true terms. On the other hand, SDP, CI and MVR produce support sets of size $10$--$11$ with eFDR $\geqslant 0.40$. On Burgers under e-BH, the equicorrelated generator is the least selective ($\abs{\hat{\mathcal{S}}} = 8$) while CI and MVR produce tighter sets ($\abs{\hat{\mathcal{S}}} = 4$); the additional false positives are nevertheless harmless, since Algorithms~\ref{alg:shap_rfe} and~\ref{alg:mcdm_pde} are designed to eliminate them. On RD, the $u$- and $v$-equations yield identical support sets for every configuration, and the results are invariant to the aggregation method within each generator, differing only between the equicorrelated construction ($\abs{\hat{\mathcal{S}}} = 11$) and the remaining three ($\abs{\hat{\mathcal{S}}} = 10$).

We adopt the equicorrelated construction with e-BH aggregation as the default throughout the remaining experiments because it achieves full power on every PDE, ensuring that no true governing term is lost. The equicorrelated construction is also the simplest, requiring only a scalar solve for $\mathrm{diag}(\bs{d})$ that satisfies the positive-semidefinite constraint. Its tendency to over-select rather than under-select is acceptable: the resulting support set remains small enough that false positives can be efficiently removed in subsequent stages, whereas lost true terms may never be recovered.

\subsection{Recursive feature elimination}\label{sec:exp_rfe}

The knockoff filter (Algorithm~\ref{alg:adaptive_fdr}) is designed to produce a small, overcomplete support set $\hat{\mathcal{S}}$ that retains all true governing terms at the cost of admitting a small number of false positives. Algorithm~\ref{alg:shap_rfe} subsequently refines this set via SHAP selection and then subjects each remaining vulnerable term to a knockoff-perturbed Wilcoxon signed-rank test at significance level $\alpha = 0.10$. The knockoff-perturbation subprocedure substitutes the vulnerable term with its conditionally independent knockoff copy. We then test whether this substitution yields a statistically significant improvement, or no impairment, in the penalised residual score. If it does ($p$-value $< \alpha$), this provides evidence that the vulnerable term is statistically unnecessary, and we are advised to trigger its removal. Table~\ref{tab:rfe_results} reports the support set $\hat{\mathcal{S}}_{\mathrm{RFE}}$ after RFE on all five benchmarks. We observe ePOWER $= 1.00$ in every case: the RFE stage never discards a true governing term, and the knockoff-perturbed Wilcoxon test never rejects the null hypothesis for a genuine mechanism. On KS, RD and GS, RFE additionally achieves eFDR $= 0.00$, recovering the exact true support set.

\begin{table}[htbp]
	\centering
	\caption{Identified support $\hat{\mathcal{S}}_{\mathrm{RFE}}$ after RFE (Algorithm~\ref{alg:shap_rfe}). We list the true support size $c$ defined in~\eqref{eq:pde-formulation}. $\abs{\hat{\mathcal{S}}_{\mathrm{RFE}}}$ is the discovered support size after RFE.}\label{tab:rfe_results}
	\small
	\setlength{\tabcolsep}{4pt}
	\begin{tabular}{lccccl}
		\hline
		PDE & True support size $c$ & $\abs{\hat{\mathcal{S}}_{\mathrm{RFE}}}$ & eFDR & ePOWER & $\hat{\mathcal{S}}_{\mathrm{RFE}}$ \\
		\hline \noalign{\vskip 1mm}
		Burgers  & 2 & 5 & 0.60 & 1.00 & $\{uu_{x},\; u_{xx},\; u_{xxxx},\; u^{5}u_{x},\; u^{6}u_{x}\}$ \\
		KdV      & 2 & 4 & 0.50 & 1.00 & $\{uu_{x},\; u_{xxx},\; u_{x},\; u_{xxxxx}\}$ \\
		KS       & 3 & 3 & 0.00 & 1.00 & $\{u_{xx},\; u_{xxxx},\; uu_{x}\}$ \\
		RD ($u$) & 7 & 7 & 0.00 & 1.00 & $\{u,\; v^{3},\; u^{3},\; u^{2}v,\; uv^{2},\; u_{xx},\; u_{yy}\}$ \\
		RD ($v$) & 7 & 7 & 0.00 & 1.00 & $\{v,\; u^{3},\; v^{3},\; uv^{2},\; u^{2}v,\; v_{yy},\; v_{xx}\}$ \\
		GS ($u$) & 6 & 6 & 0.00 & 1.00 & $\{1,\; u,\; uv^{2},\; u_{xx},\; u_{yy},\; u_{zz}\}$ \\
		GS ($v$) & 5 & 5 & 0.00 & 1.00 & $\{v,\; uv^{2},\; v_{xx},\; v_{yy},\; v_{zz}\}$ \\
		\hline
	\end{tabular}
\end{table}

The Burgers and KdV examples illustrate the complementary roles of RFE and MCDM. For neither equation does RFE drive the eFDR to zero; using the knockoff-perturbed Wilcoxon test, we cannot conclude that the surviving false positives---higher-order derivatives ($u_{xxxx}$, $u_{xxxxx}$) and high-degree nonlinear interactions ($u^{5}u_{x}$, $u^{6}u_{x}$)---are unnecessary. Nevertheless, the $\ell_{0}$-constrained knockoff filter combined with RFE narrows the candidate library to just $4$--$5$ high-potential alternatives. The true degree of parsimony is then resolved at the MCDM stage (Section~\ref{sec:exp_mcdm}), which leverages multiple criteria (e.g., distributional predictive accuracy, complexity and coefficient uncertainty) to identify the governing equation among the alternatives generated by best-subset selection.

\paragraph{Remark (Effectiveness of knockoff-perturbed hypothesis testing).}Although knockoff perturbations are not primarily designed to aggressively remove insignificant features, they nevertheless contribute usefully to the elimination of false positives. To illustrate, increasing $s_{\max}$ from $8$ to $16$ on the Burgers dataset yields, after Algorithm~\ref{alg:shap_rfe}, an eFDR of $0.82$ ($\abs{\hat{\mathcal{S}}_{\mathrm{RFE}}} = 11$) without hypothesis testing and $0.80$ ($\abs{\hat{\mathcal{S}}_{\mathrm{RFE}}} = 10$, one false positive removed) with hypothesis testing, while ePOWER remains $1$ in both cases.

\subsection{Multi-criteria model selection and coefficient accuracy}\label{sec:exp_mcdm}

\subsubsection{Ranking PDE alternatives}

Once Algorithm~\ref{alg:shap_rfe} has produced the refined support set $\hat{\mathcal{S}}_{\mathrm{RFE}}$, best-subset selection enumerates all subsets of sizes $1, \ldots, \abs{\hat{\mathcal{S}}_{\mathrm{RFE}}}$ to generate competing PDE alternatives, each of which is evaluated on the decision-matrix criteria of Algorithm~\ref{alg:mcdm_pde}: model complexity (structural and informational), distributional predictive accuracy (mean relative pinball loss at the $5$th, $50$th and $95$th percentiles), PDE uncertainty (RIF-based coefficient dispersion), mean miscalibration (Murphy--Ehm decomposition) and the coverage-width criterion (CWC). Figures~\ref{fig:criteria_1d} and~\ref{fig:criteria_multi} visualise these criteria for all benchmarks; Figure~\ref{fig:mcdm_all} shows the corresponding normalised MCDM preference scores.

A consistent pattern emerges across all datasets. Taking the Burgers equation (Figure~\ref{fig:criteria_1d}a) as an illustrative example: the $1$-term alternative exhibits high pinball loss, miscalibration, PDE uncertainty and CWC, reflecting severe under-specification. Allowing the second term (the true $u_{xx}$ term added) produces a sharp improvement across all criteria. Beyond $\abs{\mathcal{S}} = 2$, predictive accuracy and miscalibration plateau, while both complexity measures and PDE uncertainty monotonically rise---the clear hallmark of over-parameterisation. KdV (Figure~\ref{fig:criteria_1d}b) exhibits the same elbow at $\abs{\mathcal{S}} = 2$, with PDE uncertainty showing a spike at $\abs{\mathcal{S}} = 3$. For KS (Figure~\ref{fig:criteria_1d}c), all quality criteria improve monotonically up to the full $3$-term support, consistent with RFE having already recovered the exact governing equation.

The multi-component systems present a richer picture. On RD (Figure~\ref{fig:criteria_multi}a,\,b), pinball loss and CWC decrease steadily up to $\abs{\mathcal{S}} = 7$ for both the $u$- and $v$-equations, while informational complexity rises sharply at the full support size. The predictive gains from allowing additional terms are genuine, reflecting the seven-term reaction--diffusion dynamics, and the two equations display nearly identical criterion profiles. Mean miscalibration drops sharply at $\abs{\mathcal{S}} = 2$ and remains near zero thereafter, indicating that even a minimal two-term model captures the conditional quantile structure well. PDE uncertainty, by contrast, dips at $\abs{\mathcal{S}} = 2$, rises through $\abs{\mathcal{S}} = 6$, and drops back considerably only at the full $\abs{\mathcal{S}} = 7$ support---suggesting that intermediate subsets exhibit unstable coefficient estimates that are resolved only when the complete reaction--diffusion structure is assembled. On GS (Figure~\ref{fig:criteria_multi}c,\,d), pinball loss and CWC fall steadily as support sizes are increased, reaching their minima at $\abs{\mathcal{S}} = 6$ for $u$ and $\abs{\mathcal{S}} = 5$ for $v$. Mean miscalibration and PDE uncertainty exhibit a spike at $\abs{\mathcal{S}} = 2$ before overall dropping by orders of magnitude and levelling off.

\begin{figure}[htbp]
	\centering
	\begin{tabular}{@{}c@{}}
		\includegraphics[width=\textwidth]{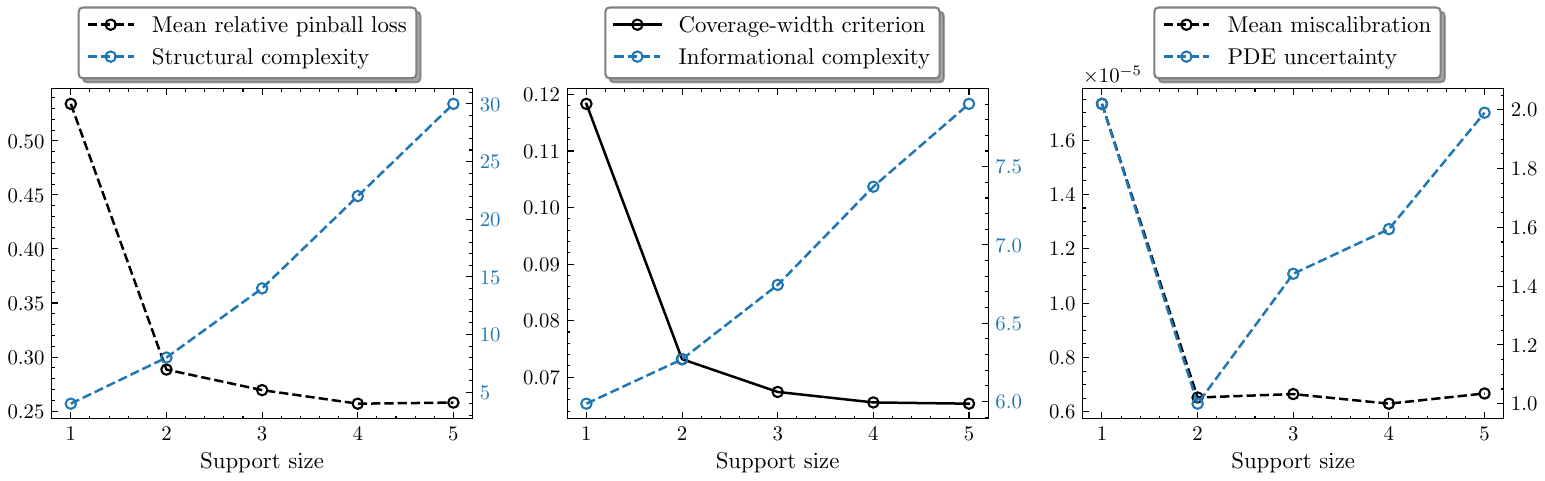} \\[-2pt]
		{\small (a) Burgers} \\[6pt]
		\includegraphics[width=\textwidth]{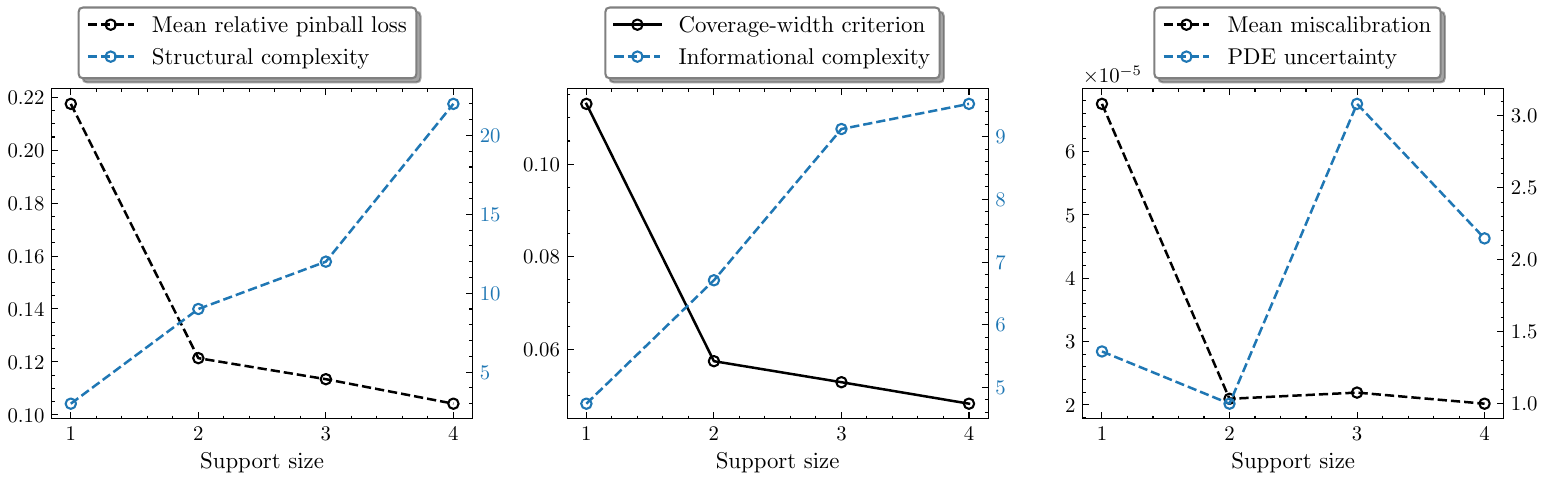} \\[-2pt]
		{\small (b) KdV} \\[6pt]
		\includegraphics[width=\textwidth]{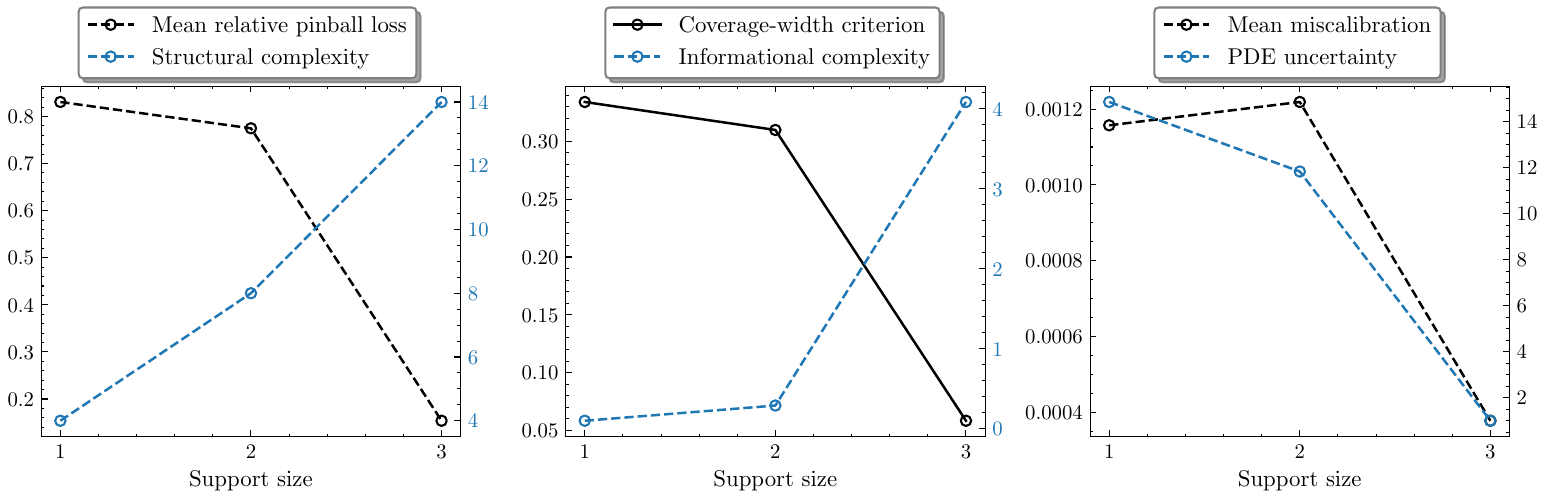} \\[-2pt]
		{\small (c) KS}
	\end{tabular}
	\caption{Decision-matrix criteria for the 1D benchmarks. Each row shows three panels: (left) mean relative pinball loss (black) and structural term complexity (blue); (centre) coverage-width criterion (black) and informational complexity via RICOMP-MM (blue); (right) mean miscalibration (black) and RIF-based PDE uncertainty (blue), all plotted against the support size $\abs{\mathcal{S}}$. The true support size occupies the optimal trade-off point in every case.}\label{fig:criteria_1d}
\end{figure}

\begin{figure}[htbp]
	\centering
	\begin{tabular}{@{}c@{}}
		\includegraphics[width=\textwidth]{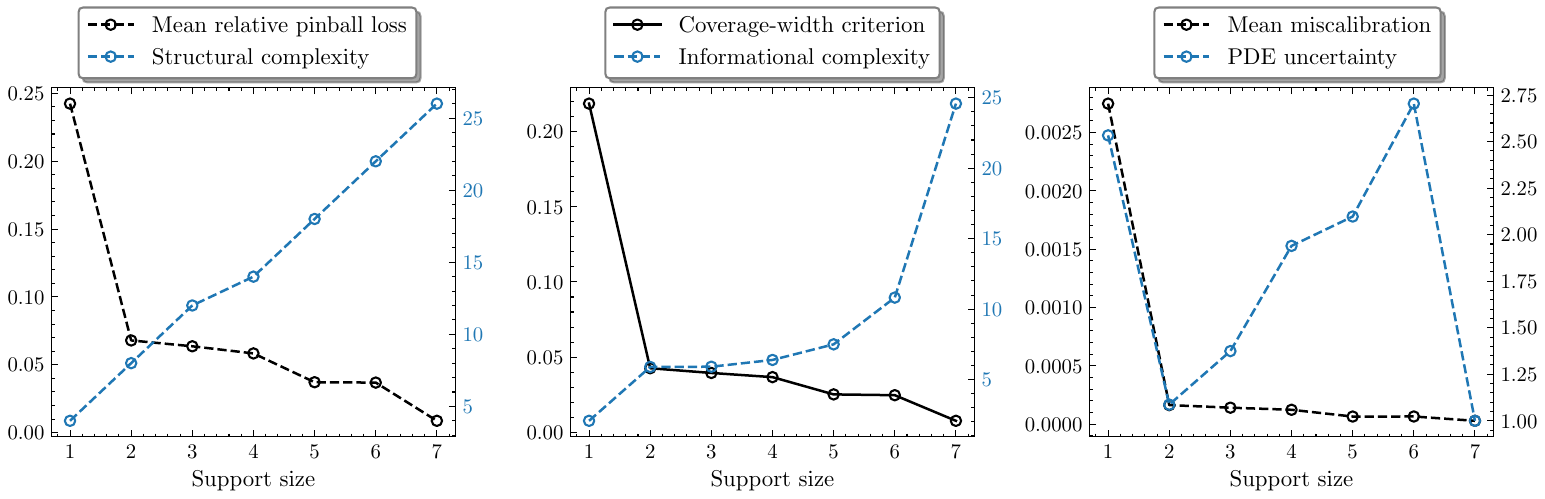} \\[-2pt]
		{\small (a) RD ($u$)} \\[6pt]
		\includegraphics[width=\textwidth]{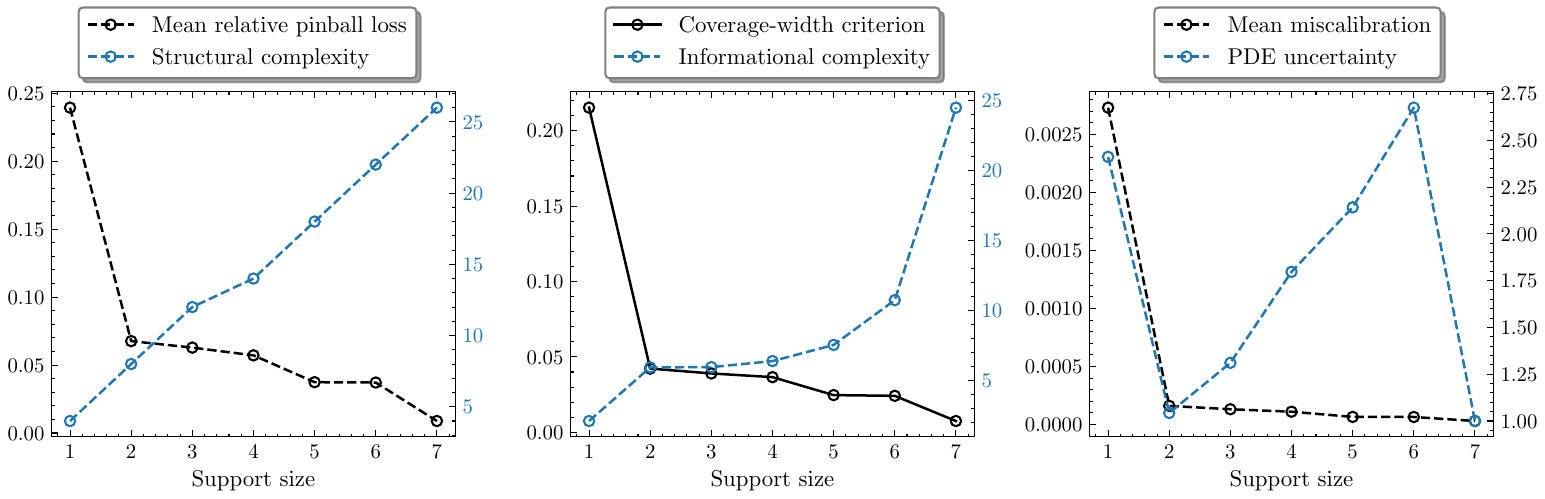} \\[-2pt]
		{\small (b) RD ($v$)} \\[6pt]
		\includegraphics[width=\textwidth]{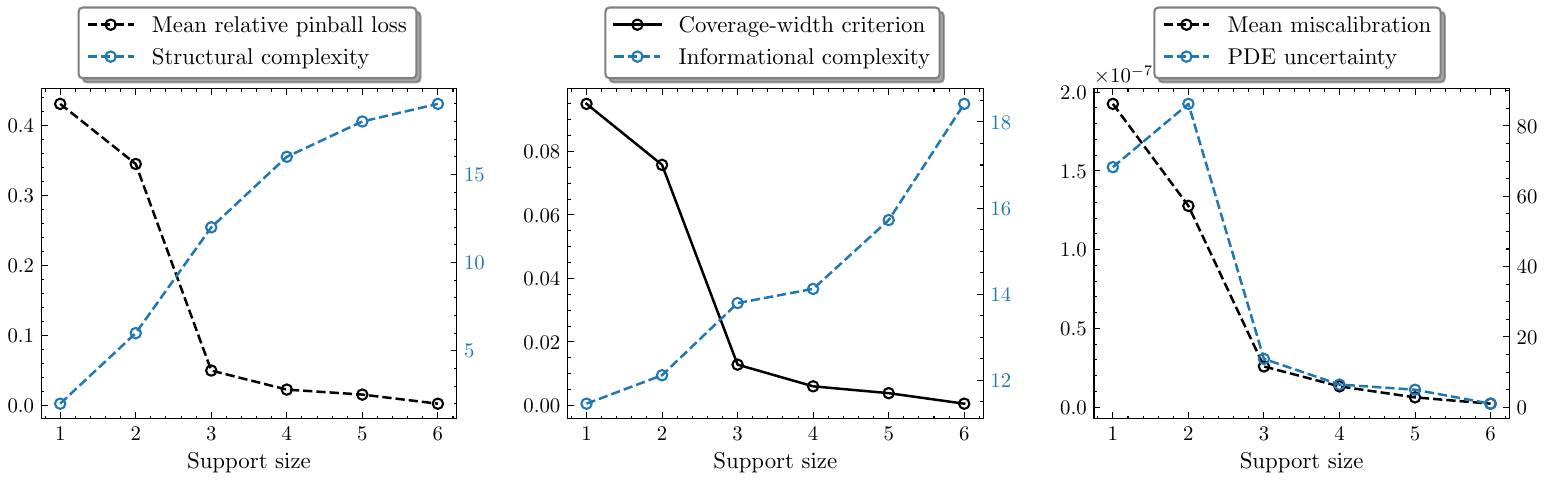} \\[-2pt]
		{\small (c) GS ($u$)} \\[6pt]
		\includegraphics[width=\textwidth]{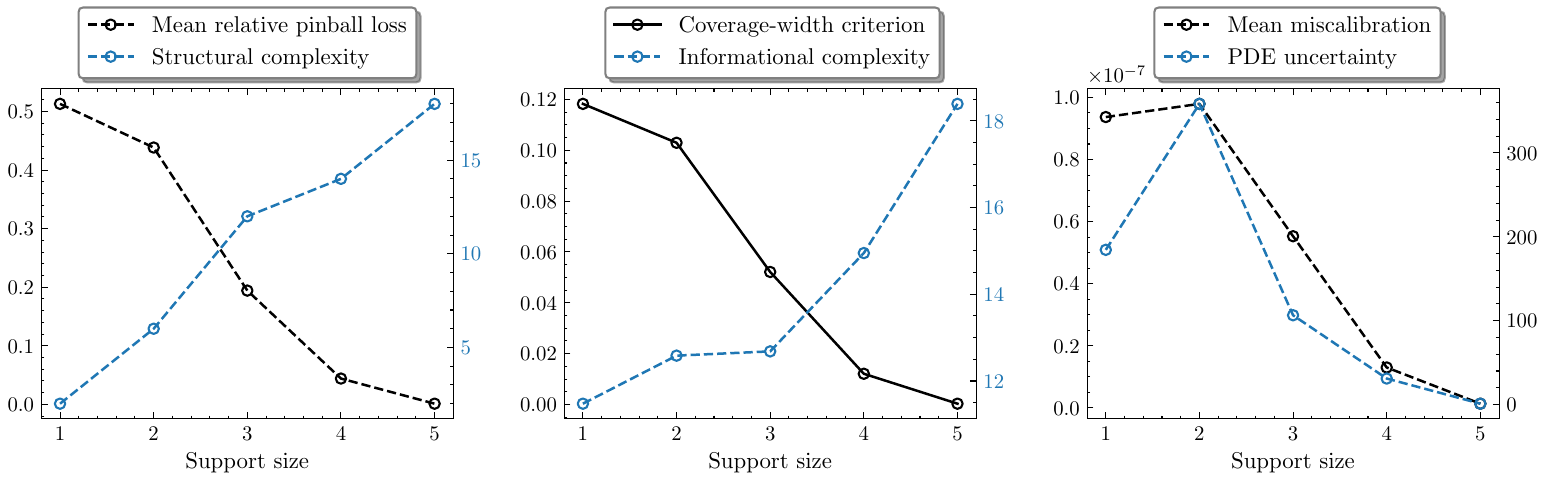} \\[-2pt]
		{\small (d) GS ($v$)}
	\end{tabular}
	\caption{Decision-matrix criteria for the multi-component benchmarks. The RD $u$- and $v$-equations display nearly identical criterion profiles. For GS, the true support size is $6$ for the $u$-equation and $5$ for the $v$-equation; both are correctly identified as the optimal trade-off points.}\label{fig:criteria_multi}
\end{figure}

The MCDM preference scores (Figure~\ref{fig:mcdm_all}) corroborate these observations. On Burgers (Figure~\ref{fig:mcdm_all}a), four of the five algorithms (TOPSIS, VIKOR, COMET, CoCoSo) assign peak preference to the 2-term alternative; PROMETHEE-II peaks at $\abs{\mathcal{S}} = 4$, but the consensus mean still ranks $\abs{\mathcal{S}} = 2$ first. On KdV (Figure~\ref{fig:mcdm_all}b), all five algorithms agree at $\abs{\mathcal{S}} = 2$, producing a sharper consensus. KS (Figure~\ref{fig:mcdm_all}c) is unanimous: the full 3-term support receives peak preference from every method. On RD (Figure~\ref{fig:mcdm_all}d,\,e), the individual methods are more dispersed---VIKOR appears to favour $\abs{\mathcal{S}} = 3$ initially---but the consensus correctly selects the full 7-term support for both the $u$- and $v$-equations. On GS (Figure~\ref{fig:mcdm_all}f,\,g), the consensus peaks at the true support size ($\abs{\mathcal{S}} = 6$ for $u$ and $\abs{\mathcal{S}} = 5$ for $v$); the $v$-equation additionally exhibits a transient dip at $\abs{\mathcal{S}} = 2$ before recovering.

\begin{figure}[htbp]
	\centering
	\begin{tabular}{@{}ccc@{}}
		\includegraphics[width=0.32\textwidth]{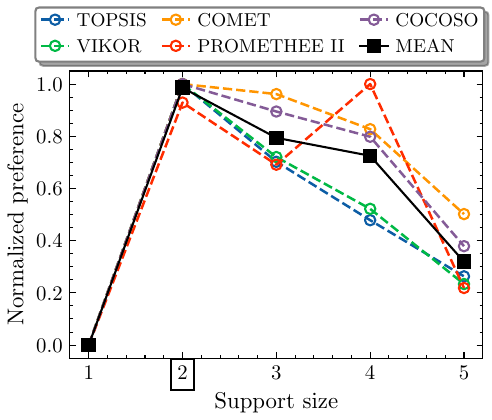} &
		\includegraphics[width=0.32\textwidth]{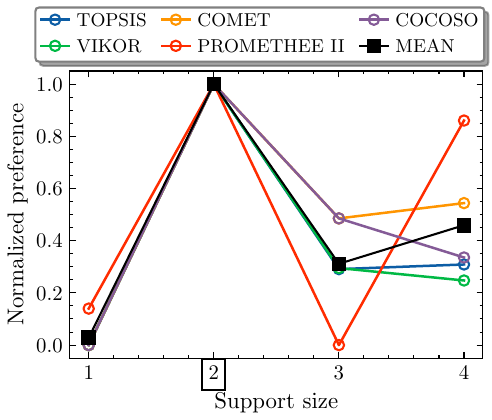} &
		\includegraphics[width=0.32\textwidth]{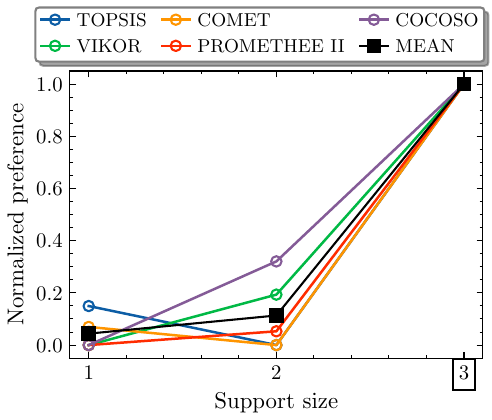} \\
		{\small (a) Burgers} & {\small (b) KdV} & {\small (c) KS} \\[6pt]
		\multicolumn{3}{c}{%
			\includegraphics[width=0.32\textwidth]{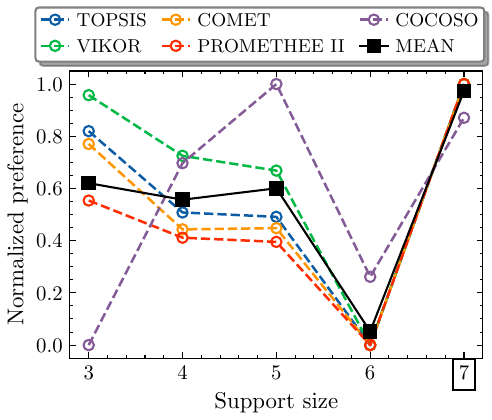}\hfill
			\includegraphics[width=0.32\textwidth]{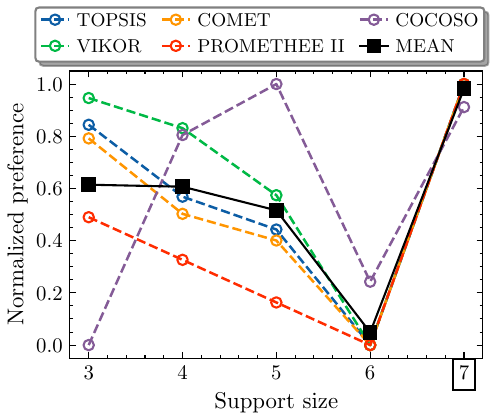}%
		} \\
		\multicolumn{3}{c}{%
			{\small (d) RD ($u$)}\hfill{\small (e) RD ($v$)}%
		} \\[6pt]
		\multicolumn{3}{c}{%
			\includegraphics[width=0.32\textwidth]{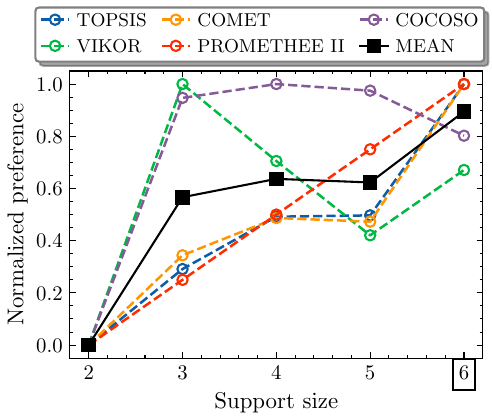}\hfill
			\includegraphics[width=0.32\textwidth]{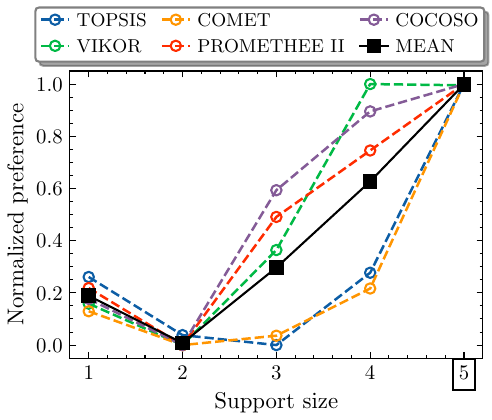}%
		} \\
		\multicolumn{3}{c}{%
			{\small (f) GS ($u$)}\hfill{\small (g) GS ($v$)}%
		}
	\end{tabular}
	\caption{Normalised MCDM preference scores for all benchmarks. Each panel shows the preference assigned by TOPSIS, VIKOR, COMET, PROMETHEE-II, and CoCoSo, together with the consensus mean (MEAN, solid black). The true support size (boxed on the horizontal axis) receives the highest consensus preference in every case.}\label{fig:mcdm_all}
\end{figure}

Across all five benchmarks, the MCDM ensemble consistently selected the correct governing equation as the top-ranked alternative. For Burgers and KdV---where RFE alone could not drive the eFDR to zero (Section~\ref{sec:exp_rfe})---the MCDM stage successfully discarded the remaining false positives ($u_{xxxx}$, $u^{5}u_{x}$, $u^{6}u_{x}$ for Burgers; $u_{x}$, $u_{xxxxx}$ for KdV), confirming that multi-criteria ranking can resolve the residual trade-offs across complexity, accuracy, uncertainty, calibration and coverage that single-criterion optimisation alone cannot. For KS, RD and GS---where RFE already achieved exact recovery---the MCDM consensus validated the RFE-selected support by ranking the full set first, providing an independent confirmation.

\subsubsection{Coefficient accuracy}

After structural recovery, we assessed the accuracy of the discovered equations using the percentage relative coefficient error ($\%$CE). For each truly relevant term from an exact recovery, we compute
\begin{equation*}
	\%\textrm{CE}_{j} = 100 \times \frac{\abs{\hat{\xi}^{\mathrm{OLS}}_{j} - \xi_{j}}}{\abs{\xi_{j}}},
\end{equation*}
where $\xi_{j}$ is the true coefficient and $\hat{\xi}^{\mathrm{OLS}}_{j}$ is obtained from an unbiased OLS re-estimation on the discovered support; we report the resulting mean and standard deviation (std) across the recovered terms. This re-estimation is unbiased because the MCDM stage selects the exact support set before it is performed on the original candidate library $\bs{\Phi}$.

Table~\ref{tab:coefficient_accuracy} reports the discovered equations and their coefficient errors. Under $50\%$ additive noise, the 1D equations achieve $\%$CE below $5\%$: Burgers and KS exhibit particularly low errors ($1.73\%$ and $2.57\%$, respectively), while KdV---whose governing terms involve higher-order derivatives that are more sensitive to noise amplification---reaches $4.51\%$. The multi-component RD system, observed at $10\%$ noise, achieves sub-$1\%$ coefficient errors on both equations. The GS system, observed at $0.1\%$ noise, attains errors below $0.03\%$, reflecting the near-exact governing equations. These results confirm that the weak-form integral formulation effectively attenuates measurement noise; the FDR-controlled knockoff screening produces a small set of candidate terms with finite-sample guarantees, and the RFE refinement and MCDM-based model selection that follow then preserve structural accuracy throughout the remainder of the data-driven discovery process.

\begin{table}[htbp]
	\centering
	\caption{Discovered governing equations and coefficient accuracy. The MCDM stage recovers the exact true support for every PDE. $\%$CE denotes the percentage relative coefficient error, computed from the full-precision coefficient estimates and reported as the mean and standard deviation across the recovered terms; the displayed coefficients are rounded to four decimal places, so $\%$CE values may not be exactly reproducible from them.}\label{tab:coefficient_accuracy}
	\small
	\resizebox{\textwidth}{!}{
		\begin{tabular}{lccc}
			\hline
			PDE & Discovered equation & $\%$CE (mean) & $\%$CE (std) \\
			\hline
			Burgers & $\partial_t u = 0.0991\partial_x^2 u - 0.9749u\partial_x u$ & 1.7302 & 0.7804 \\[3pt]
			KdV & $\partial_t u = -0.9352\partial_x^3 u - 5.8478u\partial_x u$ & 4.5107 & 1.9741 \\[3pt]
			KS & $\partial_t u = -0.9699\partial_x^2 u - 0.9749\partial_x^4 u - 0.9780u\partial_x u$ & 2.5740 & 0.3358 \\[3pt]
			\multirow{2}{*}{RD} & $\partial_t u = 0.9868u - 0.9858u^3 + 1.0004u^2v - 0.9855uv^2 + 0.9997v^3 + 0.0998\partial_y^2 u + 0.0996\partial_x^2 u$ & 0.7033 & 0.6130 \\
			& $\partial_t v = 0.9913v - 0.9999u^3 - 0.9902u^2v - 1.0001uv^2 - 0.9903v^3 + 0.1001\partial_y^2 v + 0.0999\partial_x^2 v$ & 0.4243 & 0.4456 \\[3pt]
			\multirow{2}{*}{GS} & $\partial_t u = 0.0140 - 0.0140u - 0.9999uv^2 + 0.0200\partial_z^2 u + 0.0200\partial_y^2 u + 0.0200\partial_x^2 u$ & 0.0281 & 0.0444 \\
			& $\partial_t v = -0.0670v + 1.0000uv^2 + 0.0100\partial_z^2 v + 0.0100\partial_y^2 v + 0.0100\partial_x^2 v$ & 0.0062 & 0.0042 \\
			\hline
		\end{tabular}
	}
\end{table}

\subsubsection{Comparison with sparse-regression baselines}\label{sec:exp_baselines}

A comparison with three PySINDy methods (STLSQ, SR3 and SSR) under a $3$-fold cross-validated tuning procedure confirms that no existing method offers formal FDR control across all benchmarks (Appendix~\ref{app:baselines}). On the high-noise 1D equations, all three baselines select $28$--$43$ terms against a true support of $2$, with eFDR exceeding $0.93$; KO-PDE-IDENT ultimately achieves eFDR $= 0$ and ePOWER $= 1$ (exact recovery) on every benchmark. Further empirical comparisons---against Lasso as the underlying sparse regressor, against selective inference, and against conventional information criteria (AIC and EBIC) for the final selection stage---are reported in Appendix~\ref{app:more_results}.

\section{Conclusion}

\paragraph{Summary.}We have introduced KO-PDE-IDENT, a framework that shifts data-driven PDE discovery from heuristic curve-fitting to the paradigm of statistical inference. The core idea is to first extract a statistically controlled set of potential candidate terms---with provable finite-sample FDR guarantees---using knockoff filtering before selecting the governing equation. This systematic approach prevents the search (the equation-mining process) from being lost in the vast space of candidate terms that plagues conventional loss-minimisation approaches, and thus makes the subsequent model selection more efficient and reliable.

The framework comprises three stages. First, weak-form integral representations attenuate observational noise, after which MX knockoff filters---equipped with an $\ell_{0}$-constrained \texttt{abess} solver and SHAP-based difference statistics---prune the overcomplete candidate library while controlling the FDR. An adaptive scheme evaluates both Graphical Lasso and Ledoit--Wolf covariance estimators, retaining the selection with the lower information-criterion score; the equicorrelated knockoff construction paired with the parameter-free e-BH aggregation (for derandomisation) is adopted as the default, achieving full power across all benchmarks. Second, an RFE procedure removes terms whose marginal contributions are dispensable, and subjects each vulnerable term to a knockoff-perturbed Wilcoxon signed-rank test to assess statistical necessity. Third, best-subset selection generates a discrete set of PDE alternatives, which are ranked by an ensemble of five MCDM algorithms aggregated through mean normalised preferences, using five distinct types of criteria (six in total): distributional predictive accuracy (via conformal prediction), model complexity (structural and informational), RIF-based coefficient uncertainty, miscalibration, and a coverage-width criterion.

We validated KO-PDE-IDENT on five canonical PDEs of varying dimensionality and noise severity: the 1D Burgers, KdV and KS equations at $50\%$ noise, a 2D Reaction--Diffusion system at $10\%$ noise, and a 3D Gray--Scott system at $0.1\%$ noise. The key findings are as follows.
\begin{itemize}
	\item \textit{Robustness to initial FDRs and feature statistics.} The combination of Algorithm~\ref{alg:adaptive_fdr} (knockoff filtering) and Algorithm~\ref{alg:shap_rfe} (refinement via RFE) decouples the initial FDR level $q_{0}$ from the final discovery outcome: structural recovery is identical for $q_{0} \in \{1/4, 1/3, 1/2\}$ on four of the five benchmarks (Section~\ref{sec:exp_fdr}). The SHAP-based difference statistic matches the structural recovery of the swap and swap-integral statistics at the lowest computational cost, with favourable asymptotic scaling in the library dimension (Appendix~\ref{app:complexity}).
	\item \textit{Reliable knockoff generation and aggregation.} The equicorrelated construction paired with e-BH achieves ePOWER $= 1.00$ across all $24$ generator--PDE combinations tested, and requires no tuning parameter (Section~\ref{sec:exp_gen}). Its tendency to over-select is conservative in the right direction: false positives are removed by subsequent stages, whereas lost true terms cannot be recovered.
	\item \textit{Exact structural recovery.} On KS, RD and GS, the knockoff filter combined with RFE recovers the exact true support even before the MCDM stage is applied. On Burgers and KdV, RFE narrows the candidate library to $4$--$5$ high-potential alternatives, and the MCDM consensus correctly selects the true two-term governing equation (Section~\ref{sec:exp_rfe}).
	\item \textit{Accurate coefficients and MCDM consensus.} The recursive MCDM procedure unanimously ranked the true governing equation first on every benchmark (Section~\ref{sec:exp_mcdm}). Unbiased OLS re-estimation on the discovered support yielded low mean coefficient errors at the noise levels considered (Table~\ref{tab:coefficient_accuracy}). A direct comparison with PySINDy's STLSQ, SR3 and SSR baselines (Section~\ref{sec:exp_baselines}) confirms that KO-PDE-IDENT is the only method achieving eFDR $= 0$ and ePOWER $= 1$ on every benchmark.
\end{itemize}

KO-PDE-IDENT establishes a new standard for constructing trustworthy data-driven PDE discovery pipelines: a rigorous screening stage with finite-sample FDR control delivers a small set of candidate terms, which empirical refinement and multi-criteria selection then distil into a parsimonious governing equation. KO-PDE-IDENT is designed to be flexible; for example, if the support set obtained from Algorithm~\ref{alg:adaptive_fdr} is already sufficiently small for exhaustive best-subset generation, one may consider ablating the RFE refinement in practice. By ensuring that the identified dynamics are both statistically screened under finite-sample FDR control at the initial stage and physically parsimonious after model selection, KO-PDE-IDENT lays the groundwork for the reliable data-driven discovery of governing physical laws in complex, noisy dynamical systems.

\paragraph{Future work.}Several promising directions remain for future work. First, the current framework treats each governing equation independently---FDR is controlled per response variable, with no joint guarantee across multiple targets. For systems of coupled PDEs (such as RD or GS), this leaves open the question of multi-task FDR control: extending the knockoff construction along the lines of Dai and Barber~\cite{dai2016knockoff} would allow simultaneous discovery of all governing equations under a single, joint FDR guarantee. Second, although the current framework assumes a single temporal derivative on the left-hand side, adapting the methodology to higher-order temporal operators or systems with implicit algebraic constraints would broaden the class of discoverable dynamics. Third, integrating the discovered equations with physics-informed neural-network solvers for forward validation and uncertainty propagation \cite{thanasutives2023nPIML} represents a natural downstream extension. Finally, applying KO-PDE-IDENT to real-world experimental data, where the governing equation is unknown, would test the framework's practicality and reveal its benefits beyond canonical benchmarks.

\section*{Acknowledgement}
Pongpisit Thanasutives is supported by the research fund from the Special Postdoctoral Researcher (SPDR) Program at RIKEN, Japan.

\bibliographystyle{unsrt}
\bibliography{reference}

@article{rudy2017data,
  title={{Data-driven discovery of partial differential equations}},
  author={Rudy, Samuel H and Brunton, Steven L and Proctor, Joshua L and Kutz, J Nathan},
  journal={Science Advances},
  volume={3},
  number={4},
  pages={e1602614},
  year={2017},
  publisher={American Association for the Advancement of Science}
}

@article{brunton2016sindy,
  author  = {Brunton, Steven L. and Proctor, Joshua L. and Kutz, J. Nathan},
  title   = {Discovering governing equations from data by sparse identification of nonlinear dynamical systems},
  journal = {Proceedings of the National Academy of Sciences},
  year    = {2016},
  volume  = {113},
  number  = {15},
  pages   = {3932--3937},
  doi     = {10.1073/pnas.1517384113}
}

@article{schaeffer2017learning,
  author  = {Schaeffer, Hayden},
  title   = {Learning partial differential equations via data discovery and sparse optimization},
  journal = {Proceedings of the Royal Society A: Mathematical, Physical and Engineering Sciences},
  year    = {2017},
  volume  = {473},
  number  = {2197},
  pages   = {20160446},
  doi     = {10.1098/rspa.2016.0446}
}

@article{champion2019datadriven,
  author  = {Champion, Kathleen and Lusch, Bethany and Kutz, J. Nathan and Brunton, Steven L.},
  title   = {Data-driven discovery of coordinates and governing equations},
  journal = {Proceedings of the National Academy of Sciences},
  year    = {2019},
  volume  = {116},
  number  = {45},
  pages   = {22445--22451},
  doi     = {10.1073/pnas.1906995116}
}

@article{kaheman2020sindypi,
  author  = {Kaheman, Kadierdan and Kutz, J. Nathan and Brunton, Steven L.},
  title   = {{SINDy-PI}: a robust algorithm for parallel implicit sparse identification of nonlinear dynamics},
  journal = {Proceedings of the Royal Society A},
  year    = {2020},
  volume  = {476},
  number  = {2242},
  pages   = {20200279},
  doi     = {10.1098/rspa.2020.0279}
}

@article{messenger2021weak,
  author  = {Messenger, Daniel A. and Bortz, David M.},
  title   = {Weak {SINDy} for partial differential equations},
  journal = {Journal of Computational Physics},
  year    = {2021},
  volume  = {443},
  pages   = {110525},
  doi     = {10.1016/j.jcp.2021.110525}
}

@article{fasel2022ensemble,
  author  = {Fasel, Urban and Kutz, J. Nathan and Brunton, Bingni W. and Brunton, Steven L.},
  title   = {Ensemble-{SINDy}: Robust sparse model discovery in the low-data, high-noise limit, with active learning and control},
  journal = {Proceedings of the Royal Society A},
  year    = {2022},
  volume  = {478},
  number  = {2260},
  pages   = {20210904},
  doi     = {10.1098/rspa.2021.0904}
}

@article{maddu2022stability,
  author  = {Maddu, Suryanarayana and Cheeseman, Bevan L. and Sbalzarini, Ivo F. and M{\"u}ller, Christian L.},
  title   = {Stability selection enables robust learning of differential equations from limited noisy data},
  journal = {Proceedings of the Royal Society A},
  year    = {2022},
  volume  = {478},
  number  = {2262},
  pages   = {20210916},
  doi     = {10.1098/rspa.2021.0916}
}

@article{deSilva2020pysindy,
  author  = {de Silva, Brian M. and Champion, Kathleen and Quade, Markus and Loiseau, Jean-Christophe and Kutz, J. Nathan and Brunton, Steven L.},
  title   = {{PySINDy}: A {Python} package for the sparse identification of nonlinear dynamical systems from data},
  journal = {Journal of Open Source Software},
  year    = {2020},
  volume  = {5},
  number  = {49},
  pages   = {2104},
  doi     = {10.21105/joss.02104}
}

@article{kaptanoglu2022pysindy,
  author  = {Kaptanoglu, Alan A. and de Silva, Brian M. and Fasel, Urban and Kaheman, Kadierdan and Goldschmidt, Andy J. and Callaham, Jared L. and Delahunt, Charles B. and Nicolaou, Zachary G. and Champion, Kathleen and Loiseau, Jean-Christophe and Kutz, J. Nathan and Brunton, Steven L.},
  title   = {{PySINDy}: A comprehensive {Python} package for robust sparse system identification},
  journal = {Journal of Open Source Software},
  year    = {2022},
  volume  = {7},
  number  = {69},
  pages   = {3994},
  doi     = {10.21105/joss.03994}
}

@article{zheng2019sr3,
  author  = {Zheng, Peng and Askham, Travis and Brunton, Steven L. and Kutz, J. Nathan and Aravkin, Aleksandr Y.},
  title   = {A unified framework for sparse relaxed regularized regression: {SR3}},
  journal = {IEEE Access},
  year    = {2019},
  volume  = {7},
  pages   = {1404--1423},
  doi     = {10.1109/ACCESS.2018.2886528}
}

@article{champion2020unified,
  author  = {Champion, Kathleen and Zheng, Peng and Aravkin, Aleksandr Y. and Brunton, Steven L. and Kutz, J. Nathan},
  title   = {A unified sparse optimization framework to learn parsimonious physics-informed models from data},
  journal = {IEEE Access},
  year    = {2020},
  volume  = {8},
  pages   = {169259--169271},
  doi     = {10.1109/ACCESS.2020.3023625}
}

@article{mangan2017model,
  author  = {Mangan, Niall M. and Kutz, J. Nathan and Brunton, Steven L. and Proctor, Joshua L.},
  title   = {Model selection for dynamical systems via sparse regression and information criteria},
  journal = {Proceedings of the Royal Society A},
  year    = {2017},
  volume  = {473},
  number  = {2204},
  pages   = {20170009},
  doi     = {10.1098/rspa.2017.0009}
}

@article{chen2008extended,
  author  = {Chen, Jiahua and Chen, Zehua},
  title   = {Extended {Bayesian} information criteria for model selection with large model spaces},
  journal = {Biometrika},
  year    = {2008},
  volume  = {95},
  number  = {3},
  pages   = {759--771},
  doi     = {10.1093/biomet/asn034}
}

@article{hirsh2022sparsifying,
  author  = {Hirsh, Seth M. and Barajas-Solano, David A. and Kutz, J. Nathan},
  title   = {Sparsifying priors for {Bayesian} uncertainty quantification in model discovery},
  journal = {Royal Society Open Science},
  year    = {2022},
  volume  = {9},
  number  = {2},
  pages   = {211823},
  doi     = {10.1098/rsos.211823}
}

@article{north2022bayesian,
  author  = {North, Joshua S. and Wikle, Christopher K. and Schliep, Erin M.},
  title   = {A {Bayesian} approach for data-driven dynamic equation discovery},
  journal = {Journal of Agricultural, Biological and Environmental Statistics},
  year    = {2022},
  volume  = {27},
  number  = {4},
  pages   = {728--747},
  doi     = {10.1007/s13253-022-00514-1}
}

@article{galioto2020bayesian,
  author  = {Galioto, Nicholas and Gorodetsky, Alex A.},
  title   = {{Bayesian} system {ID}: optimal management of parameter, model, and measurement uncertainty},
  journal = {Nonlinear Dynamics},
  year    = {2020},
  volume  = {102},
  number  = {1},
  pages   = {241--267},
  doi     = {10.1007/s11071-020-05925-8}
}

@inproceedings{long2018pdenet,
  author    = {Long, Zichao and Lu, Yiping and Ma, Xianzhong and Dong, Bin},
  title     = {{PDE-Net}: Learning {PDEs} from data},
  booktitle = {Proceedings of the 35th International Conference on Machine Learning (ICML)},
  series    = {Proceedings of Machine Learning Research},
  volume    = {80},
  pages     = {3208--3216},
  year      = {2018}
}

@article{long2019pdenet2,
  author  = {Long, Zichao and Lu, Yiping and Dong, Bin},
  title   = {{PDE-Net 2.0}: Learning {PDEs} from data with a numeric-symbolic hybrid deep network},
  journal = {Journal of Computational Physics},
  year    = {2019},
  volume  = {399},
  pages   = {108925},
  doi     = {10.1016/j.jcp.2019.108925}
}

@article{barber2015controlling,
  author  = {Barber, Rina Foygel and Cand{\`e}s, Emmanuel J.},
  title   = {Controlling the false discovery rate via knockoffs},
  journal = {The Annals of Statistics},
  year    = {2015},
  volume  = {43},
  number  = {5},
  pages   = {2055--2085},
  doi     = {10.1214/15-AOS1337}
}

@article{candes2018modelx,
  author  = {Cand{\`e}s, Emmanuel and Fan, Yingying and Janson, Lucas and Lv, Jinchi},
  title   = {Panning for gold: `model-{X}' knockoffs for high dimensional controlled variable selection},
  journal = {Journal of the Royal Statistical Society: Series B (Statistical Methodology)},
  year    = {2018},
  volume  = {80},
  number  = {3},
  pages   = {551--577},
  doi     = {10.1111/rssb.12265}
}

@article{barber2020robust,
  author  = {Barber, Rina Foygel and Cand{\`e}s, Emmanuel J. and Samworth, Richard J.},
  title   = {Robust inference with knockoffs},
  journal = {The Annals of Statistics},
  year    = {2020},
  volume  = {48},
  number  = {3},
  pages   = {1409--1431},
  doi     = {10.1214/19-AOS1852}
}

@article{ren2024derandomised,
  author  = {Ren, Zhimei and Barber, Rina Foygel},
  title   = {Derandomised knockoffs: leveraging e-values for false discovery rate control},
  journal = {Journal of the Royal Statistical Society Series B: Statistical Methodology},
  year    = {2024},
  volume  = {86},
  number  = {1},
  pages   = {122--154},
  doi     = {10.1093/jrsssb/qkad085}
}

@article{wang2022evalues,
  author  = {Wang, Ruodu and Ramdas, Aaditya},
  title   = {False discovery rate control with e-values},
  journal = {Journal of the Royal Statistical Society Series B: Statistical Methodology},
  year    = {2022},
  volume  = {84},
  number  = {3},
  pages   = {822--852},
  doi     = {10.1111/rssb.12489}
}

@article{benjamini1995controlling,
  author  = {Benjamini, Yoav and Hochberg, Yosef},
  title   = {Controlling the false discovery rate: a practical and powerful approach to multiple testing},
  journal = {Journal of the Royal Statistical Society: Series B (Methodological)},
  year    = {1995},
  volume  = {57},
  number  = {1},
  pages   = {289--300},
  doi     = {10.1111/j.2517-6161.1995.tb02031.x}
}

@article{meinshausen2010stability,
  author  = {Meinshausen, Nicolai and B{\"u}hlmann, Peter},
  title   = {Stability selection},
  journal = {Journal of the Royal Statistical Society: Series B (Statistical Methodology)},
  year    = {2010},
  volume  = {72},
  number  = {4},
  pages   = {417--473},
  doi     = {10.1111/j.1467-9868.2010.00740.x}
}

@article{lee2016exact,
  author  = {Lee, Jason D. and Sun, Dennis L. and Sun, Yuekai and Taylor, Jonathan E.},
  title   = {Exact post-selection inference, with application to the lasso},
  journal = {The Annals of Statistics},
  year    = {2016},
  volume  = {44},
  number  = {3},
  pages   = {907--927},
  doi     = {10.1214/15-AOS1371}
}

@article{zhu2020polynomial,
  author  = {Zhu, Junxian and Wen, Canhong and Zhu, Jin and Zhang, Heping and Wang, Xueqin},
  title   = {A polynomial algorithm for best-subset selection problem},
  journal = {Proceedings of the National Academy of Sciences},
  year    = {2020},
  volume  = {117},
  number  = {52},
  pages   = {33117--33123},
  doi     = {10.1073/pnas.2014241117}
}

@article{zhu2022abess,
  author  = {Zhu, Jin and Wang, Xueqin and Hu, Liyuan and Huang, Junhao and Jiang, Kangkang and Zhang, Yanhang and Lin, Shiyun and Zhu, Junxian},
  title   = {abess: A fast best-subset selection library in {Python} and {R}},
  journal = {Journal of Machine Learning Research},
  year    = {2022},
  volume  = {23},
  number  = {202},
  pages   = {1--7}
}

@inproceedings{lundberg2017unified,
  author    = {Lundberg, Scott M. and Lee, Su-In},
  title     = {A unified approach to interpreting model predictions},
  booktitle = {Advances in Neural Information Processing Systems 30 (NeurIPS 2017)},
  pages     = {4765--4774},
  year      = {2017}
}

@article{lundberg2020trees,
  author  = {Lundberg, Scott M. and Erion, Gabriel and Chen, Hugh and DeGrave, Alex and Prutkin, Jordan M. and Nair, Bala and Katz, Ronit and Himmelfarb, Jonathan and Bansal, Nisha and Lee, Su-In},
  title   = {From local explanations to global understanding with explainable {AI} for trees},
  journal = {Nature Machine Intelligence},
  year    = {2020},
  volume  = {2},
  number  = {1},
  pages   = {56--67},
  doi     = {10.1038/s42256-019-0138-9}
}

@book{vovk2005algorithmic,
  author    = {Vovk, Vladimir and Gammerman, Alexander and Shafer, Glenn},
  title     = {Algorithmic Learning in a Random World},
  publisher = {Springer},
  address   = {New York},
  year      = {2005},
  doi       = {10.1007/b106715}
}

@article{lei2018distribution,
  author  = {Lei, Jing and G'Sell, Max and Rinaldo, Alessandro and Tibshirani, Ryan J. and Wasserman, Larry},
  title   = {Distribution-free predictive inference for regression},
  journal = {Journal of the American Statistical Association},
  year    = {2018},
  volume  = {113},
  number  = {523},
  pages   = {1094--1111},
  doi     = {10.1080/01621459.2017.1307116}
}

@inproceedings{romano2019conformalized,
  author    = {Romano, Yaniv and Patterson, Evan and Cand{\`e}s, Emmanuel J.},
  title     = {Conformalized quantile regression},
  booktitle = {Advances in Neural Information Processing Systems 32 (NeurIPS 2019)},
  pages     = {3538--3548},
  year      = {2019}
}

@article{barber2021predictive,
  author = {Rina Foygel Barber and Emmanuel J. Cand{\`e}s and Aaditya Ramdas and Ryan J. Tibshirani},
  title = {{Predictive inference with the jackknife+}},
  volume = {49},
  journal = {The Annals of Statistics},
  number = {1},
  publisher = {Institute of Mathematical Statistics},
  pages = {486--507},
  year = {2021},
  doi = {10.1214/20-AOS1965},
  URL = {https://doi.org/10.1214/20-AOS1965}
}

@article{firpo2009unconditional,
  author  = {Firpo, Sergio and Fortin, Nicole M. and Lemieux, Thomas},
  title   = {Unconditional quantile regressions},
  journal = {Econometrica},
  year    = {2009},
  volume  = {77},
  number  = {3},
  pages   = {953--973},
  doi     = {10.3982/ECTA6822}
}

@article{bozdogan2000akaike,
  author  = {Bozdogan, Hamparsum},
  title   = {{Akaike's} information criterion and recent developments in information complexity},
  journal = {Journal of Mathematical Psychology},
  year    = {2000},
  volume  = {44},
  number  = {1},
  pages   = {62--91},
  doi     = {10.1006/jmps.1999.1277}
}

@article{murphy1973vector,
  author  = {Murphy, Allan H.},
  title   = {A new vector partition of the probability score},
  journal = {Journal of Applied Meteorology},
  year    = {1973},
  volume  = {12},
  number  = {4},
  pages   = {595--600},
  doi     = {10.1175/1520-0450(1973)012<0595:ANVPOT>2.0.CO;2}
}

@article{ehm2016quantiles,
  author  = {Ehm, Werner and Gneiting, Tilmann and Jordan, Alexander and Kr{\"u}ger, Fabian},
  title   = {Of quantiles and expectiles: consistent scoring functions, {Choquet} representations and forecast rankings},
  journal = {Journal of the Royal Statistical Society: Series B (Statistical Methodology)},
  year    = {2016},
  volume  = {78},
  number  = {3},
  pages   = {505--562},
  doi     = {10.1111/rssb.12154}
}

@article{gneiting2007strictly,
  author  = {Gneiting, Tilmann and Raftery, Adrian E.},
  title   = {Strictly proper scoring rules, prediction, and estimation},
  journal = {Journal of the American Statistical Association},
  year    = {2007},
  volume  = {102},
  number  = {477},
  pages   = {359--378},
  doi     = {10.1198/016214506000001437}
}

@book{hwang1981multiple,
  author    = {Hwang, Ching-Lai and Yoon, Kwangsun},
  title     = {Multiple Attribute Decision Making: Methods and Applications -- A State-of-the-Art Survey},
  series    = {Lecture Notes in Economics and Mathematical Systems},
  volume    = {186},
  publisher = {Springer-Verlag},
  address   = {Berlin},
  year      = {1981},
  doi       = {10.1007/978-3-642-48318-9}
}

@article{opricovic2004compromise,
  author  = {Opricovic, Serafim and Tzeng, Gwo-Hshiung},
  title   = {Compromise solution by {MCDM} methods: A comparative analysis of {VIKOR} and {TOPSIS}},
  journal = {European Journal of Operational Research},
  year    = {2004},
  volume  = {156},
  number  = {2},
  pages   = {445--455},
  doi     = {10.1016/S0377-2217(03)00020-1}
}

@article{brans1985preference,
  author  = {Brans, Jean-Pierre and Vincke, Philippe},
  title   = {A preference ranking organisation method (the {PROMETHEE} method for multiple criteria decision-making)},
  journal = {Management Science},
  year    = {1985},
  volume  = {31},
  number  = {6},
  pages   = {647--656},
  doi     = {10.1287/mnsc.31.6.647}
}

@article{yazdani2019cocoso,
  author  = {Yazdani, Morteza and Zarate, Pascale and Zavadskas, Edmundas Kazimieras and Turskis, Zenonas},
  title   = {A combined compromise solution ({CoCoSo}) method for multi-criteria decision-making problems},
  journal = {Management Decision},
  year    = {2019},
  volume  = {57},
  number  = {9},
  pages   = {2501--2519},
  doi     = {10.1108/MD-05-2017-0458}
}

@article{salabun2015characteristic,
  author  = {Sa{\l}abun, Wojciech},
  title   = {The characteristic objects method: A new distance-based approach to multicriteria decision-making problems},
  journal = {Journal of Multi-Criteria Decision Analysis},
  year    = {2015},
  volume  = {22},
  number  = {1--2},
  pages   = {37--50},
  doi     = {10.1002/mcda.1525}
}

@article{kizielewicz2023pymcdm,
  author  = {Kizielewicz, Bart{\l}omiej and Shekhovtsov, Andrii and Sa{\l}abun, Wojciech},
  title   = {{pymcdm} -- The universal library for solving multi-criteria decision-making problems},
  journal = {SoftwareX},
  year    = {2023},
  volume  = {22},
  pages   = {101368},
  doi     = {10.1016/j.softx.2023.101368}
}

@article{barber2019knockoff,
  author  = {Barber, Rina Foygel and Cand{\`e}s, Emmanuel J.},
  title   = {A knockoff filter for high-dimensional selective inference},
  journal = {Annals of Statistics},
  year    = {2019},
  volume  = {47},
  number  = {5},
  pages   = {2504--2537},
  doi     = {10.1214/18-AOS1755}
}

@inproceedings{dai2016knockoff,
  author    = {Dai, Ran and Barber, Rina Foygel},
  title     = {The knockoff filter for {FDR} control in group-sparse and multitask regression},
  booktitle = {Proceedings of the 33rd International Conference on Machine Learning},
  series    = {Proceedings of Machine Learning Research},
  volume    = {48},
  pages     = {1851--1859},
  year      = {2016},
  publisher = {PMLR}
}

@inproceedings{gimenez2019knockoffs,
  author    = {Gimenez, Jaime Roquero and Zou, James},
  title     = {Improving the stability of the knockoff procedure: Multiple simultaneous knockoffs and entropy maximization},
  booktitle = {Proceedings of the 22nd International Conference on Artificial Intelligence and Statistics},
  series    = {Proceedings of Machine Learning Research},
  volume    = {89},
  pages     = {2184--2192},
  year      = {2019},
  publisher = {PMLR}
}

@article{khosravi2010cwc,
  author  = {Khosravi, Abbas and Nahavandi, Saeid and Creighton, Doug and Atiya, Amir F.},
  title   = {Lower upper bound estimation method for construction of neural network-based prediction intervals},
  journal = {IEEE Transactions on Neural Networks},
  year    = {2011},
  volume  = {22},
  number  = {3},
  pages   = {337--346},
  doi     = {10.1109/TNN.2010.2096824}
}

@inproceedings{marcilio2020shap,
  author    = {Marc{\'\i}lio, Wilson E. and Eler, Danilo M.},
  title     = {From explanations to feature selection: assessing {SHAP} values as feature selection mechanism},
  booktitle = {2020 33rd SIBGRAPI Conference on Graphics, Patterns and Images (SIBGRAPI)},
  pages     = {340--347},
  year      = {2020},
  publisher = {IEEE},
  doi       = {10.1109/SIBGRAPI51738.2020.00053}
}

@inproceedings{nguyen2020aggregation,
  author    = {Nguyen, Tuan-Binh and Chevalier, Jerome-Alexis and Thirion, Bertrand and Arlot, Sylvain},
  title     = {Aggregation of multiple knockoffs},
  booktitle = {Proceedings of the 37th International Conference on Machine Learning},
  series    = {Proceedings of Machine Learning Research},
  volume    = {119},
  pages     = {7283--7293},
  year      = {2020},
  publisher = {PMLR}
}

@article{reinbold2020using,
  author  = {Reinbold, Patrick A. K. and Gurevich, Daniel R. and Grigoriev, Roman O.},
  title   = {Using noisy or incomplete data to discover models of spatiotemporal dynamics},
  journal = {Physical Review E},
  year    = {2020},
  volume  = {101},
  number  = {1},
  pages   = {010203},
  doi     = {10.1103/PhysRevE.101.010203}
}

@article{spector2022powerful,
  author  = {Spector, Asher and Janson, Lucas},
  title   = {Powerful knockoffs via minimizing reconstructability},
  journal = {Annals of Statistics},
  year    = {2022},
  volume  = {50},
  number  = {1},
  pages   = {252--276},
  doi     = {10.1214/21-AOS2104}
}

@article{thanasutives2024ubic,
  author  = {Thanasutives, Pongpisit and Morita, Takashi and Numao, Masayuki and Fukui, Ken-ichi},
  title   = {Adaptive uncertainty-penalized model selection for data-driven {PDE} discovery},
  journal = {IEEE Access},
  year    = {2024},
  volume  = {12},
  pages   = {13165--13182},
  doi     = {10.1109/ACCESS.2024.3354819}
}

@article{thanasutives2025ubic,
  title = {Bayesian model selection for variable-coefficient partial differential equation discovery},
  journal = {Results in Engineering},
  volume = {27},
  pages = {106930},
  year = {2025},
  issn = {2590-1230},
  doi = {10.1016/j.rineng.2025.106930},
  author = {Pongpisit Thanasutives and Yoshinobu Kawahara and Ken{-}ichi Fukui}
}

@article{dabov2007bm3d,
  author  = {Dabov, Kostadin and Foi, Alessandro and Katkovnik, Vladimir and Egiazarian, Karen},
  title   = {Image denoising by sparse {3-D} transform-domain collaborative filtering},
  journal = {IEEE Transactions on Image Processing},
  year    = {2007},
  volume  = {16},
  number  = {8},
  pages   = {2080--2095},
  doi     = {10.1109/TIP.2007.901238}
}

@article{ke2020power,
  author  = {Ke, Zheng Tracy and Liu, Jun S. and Ma, Yucong},
  title   = {Power of knockoff: the impact of ranking algorithm, augmented design, and symmetric statistic},
  journal = {Journal of Machine Learning Research},
  year    = {2024},
  volume  = {25},
  number  = {3},
  pages   = {1--67}
}

@article{akaike1974new,
  author  = {Akaike, Hirotugu},
  title   = {A new look at the statistical model identification},
  journal = {IEEE Transactions on Automatic Control},
  year    = {1974},
  volume  = {19},
  number  = {6},
  pages   = {716--723},
  doi     = {10.1109/TAC.1974.1100705}
}

@article{schwarz1978bic,
  author = {Schwarz, Gideon},
  title = {{Estimating the Dimension of a Model}},
  volume = {6},
  journal = {The Annals of Statistics},
  number = {2},
  publisher = {Institute of Mathematical Statistics},
  pages = {461 -- 464},
  year = {1978},
  doi = {10.1214/aos/1176344136}
}

@article{gohain2023ebicr,
  author  = {Gohain, Prakash Borpatra and Jansson, Magnus},
  title   = {Robust information criterion for model selection in sparse high-dimensional linear regression models},
  journal = {IEEE Transactions on Signal Processing},
  year    = {2023},
  volume  = {71},
  pages   = {2251--2266},
  doi     = {10.1109/TSP.2023.3284365}
}

@article{tibshirani1996regression,
  author  = {Tibshirani, Robert},
  title   = {Regression shrinkage and selection via the lasso},
  journal = {Journal of the Royal Statistical Society: Series B (Methodological)},
  year    = {1996},
  volume  = {58},
  number  = {1},
  pages   = {267--288},
  doi     = {10.1111/j.2517-6161.1996.tb02080.x}
}

@article{tibshirani2016exact,
  author  = {Tibshirani, Ryan J. and Taylor, Jonathan and Lockhart, Richard and Tibshirani, Robert},
  title   = {Exact post-selection inference for sequential regression procedures},
  journal = {Journal of the American Statistical Association},
  year    = {2016},
  volume  = {111},
  number  = {514},
  pages   = {600--620},
  doi     = {10.1080/01621459.2015.1108848}
}

@article{taylor2015statistical,
  author  = {Taylor, Jonathan and Tibshirani, Robert J.},
  title   = {Statistical learning and selective inference},
  journal = {Proceedings of the National Academy of Sciences},
  year    = {2015},
  volume  = {112},
  number  = {25},
  pages   = {7629--7634},
  doi     = {10.1073/pnas.1507583112}
}

@article{boninsegna2018sparse,
  author  = {Boninsegna, Lorenzo and N{\"u}ske, Feliks and Clementi, Cecilia},
  title   = {Sparse learning of stochastic dynamical equations},
  journal = {The Journal of Chemical Physics},
  year    = {2018},
  volume  = {148},
  number  = {24},
  pages   = {241723},
  doi     = {10.1063/1.5018409}
}

@article{thanasutives2023nPIML,
  doi = {10.1088/2632-2153/acb1f0},
  url = {https://doi.org/10.1088/2632-2153/acb1f0},
  year = {2023},
  month = {feb},
  publisher = {IOP Publishing},
  volume = {4},
  number = {1},
  pages = {015009},
  author = {Thanasutives, Pongpisit and Morita, Takashi and Numao, Masayuki and Fukui, Ken-ichi},
  title = {{Noise-aware physics-informed machine learning for robust PDE discovery}},
  journal = {Machine Learning: Science and Technology}
}

@inproceedings{kim2020predictive,
  author    = {Kim, Byol and Xu, Chen and Barber, Rina Foygel},
  title     = {Predictive inference is free with the jackknife+-after-bootstrap},
  booktitle = {Advances in Neural Information Processing Systems 33 (NeurIPS 2020)},
  pages     = {4138--4149},
  year      = {2020}
}

@article{blanchard2008two,
  author  = {Blanchard, Gilles and Roquain, {\'E}tienne},
  title   = {Two simple sufficient conditions for {FDR} control},
  journal = {Electronic Journal of Statistics},
  year    = {2008},
  volume  = {2},
  pages   = {963--992},
  doi     = {10.1214/08-EJS180}
}

@article{guney2021robust,
  author  = {G{\"u}ney, Y{\i}lmaz and Bozdogan, Hamparsum and Arslan, Olcay},
  title   = {Robust model selection in linear regression models using information complexity},
  journal = {Journal of Computational and Applied Mathematics},
  year    = {2021},
  volume  = {398},
  pages   = {113679},
  doi     = {10.1016/j.cam.2021.113679}
}

\appendix

\section{Proof of FDR control for the knockoff+ filter}\label{app:ko_proof}

Let $\hat{\tau}$ be the data-dependent threshold of the knockoff+ filter, as defined in~\eqref{eq:ko_threshold} with $\textrm{offset} = 1$. We write $\mathcal{H}_{0} = \{j \in [p] : \xi_{j} = 0\}$ for the set of null candidates. For any threshold $\tau > 0$, define

\begin{equation*}
	M^{+}(\tau) = \#\{j \in \mathcal{H}_{0} : W_{j} \geqslant \tau\}, \qquad M^{-}(\tau) = \#\{j \in \mathcal{H}_{0} : W_{j} \leqslant -\tau\}.
\end{equation*}

\noindent The FDP at the threshold $\hat{\tau}$ is

\begin{equation*}
	\textrm{FDP}(\hat{\tau}) = \frac{\#\{j \in \mathcal{H}_{0} : W_{j} \geqslant \hat{\tau}\}}{\#\{j \in [p] : W_{j} \geqslant \hat{\tau}\}} = \frac{M^{+}(\hat{\tau})}{\abs{\hat{\mathcal{S}}}},
\end{equation*}

\noindent where $\hat{\mathcal{S}} = \{j \in [p] : W_{j} \geqslant \hat{\tau}\}$. This is bounded from above by

\begin{align*}
	\textrm{FDP}(\hat{\tau})
	&\leqslant \textrm{FDP}(\hat{\tau}) \cdot \frac{1 + \#\{j \in [p] : W_{j} \leqslant -\hat{\tau}\}}{1 + M^{-}(\hat{\tau})} \\[4pt]
	&= \frac{M^{+}(\hat{\tau})}{\abs{\hat{\mathcal{S}}}} \cdot \frac{1 + \#\{j \in [p] : W_{j} \leqslant -\hat{\tau}\}}{1 + M^{-}(\hat{\tau})} \\[4pt]
	&\leqslant q \cdot \frac{M^{+}(\hat{\tau})}{1 + M^{-}(\hat{\tau})},
\end{align*}

\noindent where the final inequality uses the definition of $\hat{\tau}$ in~\eqref{eq:ko_threshold}:
\begin{equation*}
	\frac{1 + \#\{j \in [p] : W_{j} \leqslant -\hat{\tau}\}}{\#\{j \in [p] : W_{j} \geqslant \hat{\tau}\}} \leqslant q.
\end{equation*}

\noindent The coin-flip property of the knockoff statistics for null candidates---namely that, conditional on the unordered values $\{\abs{W_{j}}\}_{j \in [p]}$, the signs $\mathrm{sign}(W_{j})$ for $j \in \mathcal{H}_{0}$ are i.i.d.\ uniform on $\{-1, +1\}$~\cite{candes2018modelx}---implies that $M^{+}(\tau) / (1 + M^{-}(\tau))$ forms a supermartingale in $\tau$ (run in reverse time as $\tau$ decreases). Order the null indices so that $\abs{W_{(1)}} \geqslant \abs{W_{(2)}} \geqslant \cdots$ and let $B_{j} = \mathbbm{1}\{W_{(j)} < 0\}$; then $\hat{\tau}$ is a stopping time in reverse time with respect to the filtration $\{\mathcal{F}_{j}\}$ defined by $\mathcal{F}_{j} = \{B_{1} + \cdots + B_{j},\, B_{j+1}, \dots, B_{\abs{\mathcal{H}_{0}}}\}$~\cite{barber2019knockoff}, and the optional stopping theorem gives $\mathbb{E}\!\left[M^{+}(\hat{\tau}) / (1 + M^{-}(\hat{\tau}))\right] \leqslant 1$. Taking expectations of the bound above yields

\begin{equation*}
	\textrm{FDR}(\hat{\mathcal{S}}) = \mathbb{E}\!\left[\textrm{FDP}(\hat{\tau})\right] \leqslant q \cdot \mathbb{E}\!\left[\frac{M^{+}(\hat{\tau})}{1 + M^{-}(\hat{\tau})}\right] \leqslant q.
\end{equation*}

\noindent This establishes that the knockoff+ filter controls the FDR at the target level, $\textrm{FDR}(\hat{\mathcal{S}}) \leqslant q$. For a comprehensive treatment, we refer the reader to~\cite{candes2018modelx,ren2024derandomised}.


\section{Size-constrained selection preserves FDR control}\label{app:smax_proof}

Algorithm~\ref{alg:adaptive_fdr} (the $\ell_{0}$-constrained knockoff filter with adaptive FDR) imposes an explicit upper bound $\abs{\hat{\mathcal{S}}} \leqslant s_{\max}$ on the size of the support set returned by the e-BH procedure (the rejection rule reads $\hat{\jmath}^{*} = \max\{j : j \leqslant s_{\max}\,\text{and}\,e^{\textrm{avg}}_{(j)} \geqslant p/(q_{\textrm{ebh}}j)\}$). Bounding $\abs{\hat{\mathcal{S}}}$ at $s_{\max}$ is essential in our pipeline because it ensures that the candidate set entering the downstream RFE and MCDM stages remains tractable even in high-dimensional settings where typically $p \gg n$. The question that this appendix answers is: \emph{does adding the cardinality constraint $\abs{\hat{\mathcal{S}}} \leqslant s_{\max}$ to the e-BH procedure preserve the finite-sample FDR guarantee at the target level $q_{\textrm{ebh}}$?}

We prove this in three steps, organised by procedure: (i) the size-constrained e-BH procedure on averaged e-values (Theorem~\ref{thm:smax_ebh}), which is the principal procedure deployed by Algorithm~\ref{alg:adaptive_fdr}; (ii) the size-constrained knockoff+ filter (Theorem~\ref{thm:smax_knockoff}), a corollary that covers the single-run ($K = 1$) case; (iii) the size-constrained BH procedure on aggregated p-values (Theorem~\ref{thm:smax_bh}), which establishes the same guarantee for the BHq aggregation method of Nguyen et al.~\cite{nguyen2020aggregation} reported in Table~\ref{tab:gen_comparison}. The three theorems are subsumed by a single unified statement (Theorem~\ref{thm:smax_unified}). The order matches the order in which the procedures appear in the main text: the e-BH proof is the most general; the knockoff+ and BH proofs follow as corollaries that reuse the same Markov-inequality / super-uniformity argument.

\subsection*{Notation}

We adopt the notation established in Section~2. The overcomplete candidate library $\bs{\Phi} \in \mathbb{R}^{n \times p}$ has columns $\bs{\phi}_{1}, \dots, \bs{\phi}_{p}$, with response $\bs{y}$. The set of null candidates is $\mathcal{H}_{0} = \{j \in [p] : \xi_{j} = 0\}$, and the FDR of any support set $\hat{\mathcal{S}} \subseteq [p]$ is defined in~\eqref{eq:fdr}. For $K$ knockoff realisations, the per-run e-value of the $j$th candidate in the $k$th realisation is

\begin{equation*}
	e^{k}_{j} = p \cdot \frac{\mathbbm{1}\{W^{k}_{j} \geqslant \hat{\tau}^{k}\}}{1 + \#\{j^{\prime} \in [p] : W^{k}_{j^{\prime}} \leqslant -\hat{\tau}^{k}\}},
\end{equation*}

\noindent where $\mathbbm{1}\{\cdot\}$ is the indicator function defined in Section~\ref{sec:ko}, and the averaged e-value is $e^{\textrm{avg}}_{j} = \frac{1}{K}\sum_{k=1}^{K} e^{k}_{j}$. The unconstrained e-BH support set is

\begin{equation*}
	\hat{\mathcal{S}}_{\textrm{ebh}} = \left\{j \in [p] : e^{\textrm{avg}}_{j} \geqslant \frac{p}{q_{\textrm{ebh}}\,\hat{\jmath}}\right\}, \quad \hat{\jmath} = \max\left\{j \in [p] : e^{\textrm{avg}}_{(j)} \geqslant \frac{p}{q_{\textrm{ebh}}\,j}\right\},
\end{equation*}

\noindent where $e^{\textrm{avg}}_{(1)} \geqslant \cdots \geqslant e^{\textrm{avg}}_{(p)}$ are the order statistics and $q_{\textrm{ebh}}$ is the target FDR level. The critical condition underpinning FDR control is

\begin{equation}
	\sum_{j \in \mathcal{H}_{0}} \mathbb{E}\!\left[e^{\textrm{avg}}_{j}\right] \leqslant p.
	\label{eq:eval_condition}
\end{equation}

\subsection*{Size-constrained e-BH procedure}

\begin{theorem}[Size-constrained e-BH controls FDR]\label{thm:smax_ebh}
	Let the averaged e-values $e^{\textrm{avg}}_{1}, \dots, e^{\textrm{avg}}_{p}$ satisfy~\eqref{eq:eval_condition}. For any fixed $s_{\max} \in [p]$, define the size-constrained e-BH procedure through the rejection count
	\begin{equation*}
		\hat{\jmath}^{*} = \max\left\{j \in [p] : j \leqslant s_{\max} \;\text{and}\; e^{\textrm{avg}}_{(j)} \geqslant \frac{p}{q_{\textrm{ebh}}\,j}\right\},
	\end{equation*}
	with $\hat{\jmath}^{*} = 0$ if this set is empty, in which case $\hat{\mathcal{S}}^{*}_{\textrm{ebh}} = \emptyset$ and $\textrm{FDP} := 0$ by convention. The procedure rejects the $\hat{\jmath}^{*}$ candidates attaining the largest averaged e-values, so that
	\begin{equation*}
		\hat{\mathcal{S}}^{*}_{\textrm{ebh}} = \left\{j \in [p] : e^{\textrm{avg}}_{j} \text{ is among the } \hat{\jmath}^{*} \text{ largest}\right\}, \qquad \abs{\hat{\mathcal{S}}^{*}_{\textrm{ebh}}} = \hat{\jmath}^{*}.
	\end{equation*}
	Ties in the averaged e-value are broken by the same deterministic rule throughout; in the practical implementation we rank ties by the averaged knockoff statistic $\bar{W}_{j} = \frac{1}{K}\sum_{k} W^{k}_{j}$ (descending) and then by index, though the proof below does not depend on the particular rule. This is the standard e-BH rejection set, defined through the rejection \emph{count} $\hat{\jmath}^{*}$ rather than through the threshold $p/(q_{\textrm{ebh}}\hat{\jmath}^{*})$; the distinction matters only in the boundary event of ties or when the cap $s_{\max}$ binds, where it ensures $\abs{\hat{\mathcal{S}}^{*}_{\textrm{ebh}}} = \hat{\jmath}^{*}$ exactly. Then $\abs{\hat{\mathcal{S}}^{*}_{\textrm{ebh}}} \leqslant s_{\max}$ and $\textrm{FDR}(\hat{\mathcal{S}}^{*}_{\textrm{ebh}}) \leqslant q_{\textrm{ebh}}$.
\end{theorem}

\begin{proof}
	\textit{Step~1: $\hat{\mathcal{S}}^{*}_{\textrm{ebh}} \subseteq \hat{\mathcal{S}}_{\textrm{ebh}}$.}
	The search space for $\hat{\jmath}^{*}$ is a subset of that for $\hat{\jmath}$ (the additional constraint $j \leqslant s_{\max}$ can only restrict), so $\hat{\jmath}^{*} \leqslant \hat{\jmath}$. The constrained procedure therefore rejects the top $\hat{\jmath}^{*}$ averaged e-values, a subset of the top $\hat{\jmath}$ rejected by the unconstrained procedure (under the same tie-breaking), giving $\hat{\mathcal{S}}^{*}_{\textrm{ebh}} \subseteq \hat{\mathcal{S}}_{\textrm{ebh}}$. This inclusion is recorded for later reference; it is not used in establishing FDR control below.

	\textit{Step~2: $\abs{\hat{\mathcal{S}}^{*}_{\textrm{ebh}}} = \hat{\jmath}^{*} \leqslant s_{\max}$ and self-consistency.}
	By construction the procedure rejects exactly the $\hat{\jmath}^{*}$ candidates with the largest averaged e-values, so $\abs{\hat{\mathcal{S}}^{*}_{\textrm{ebh}}} = \hat{\jmath}^{*}$, and the constraint $\hat{\jmath}^{*} \leqslant s_{\max}$ yields the size bound. Moreover, every rejected candidate satisfies
	\begin{equation*}
		e^{\textrm{avg}}_{j} \;\geqslant\; e^{\textrm{avg}}_{(\hat{\jmath}^{*})} \;\geqslant\; \frac{p}{q_{\textrm{ebh}}\,\hat{\jmath}^{*}},
	\end{equation*}
	the second inequality being the e-BH self-consistency condition that holds at the realised rejection count $\hat{\jmath}^{*}$ by its very definition. This is the only property of $\hat{\jmath}^{*}$ required below; in particular, the proof does not require $\hat{\jmath}^{*}$ to be the unrestricted maximum.
	
	\textit{Step~3: FDR control.}
	If $\hat{\jmath}^{*}=0$, then $\hat{\mathcal{S}}^{*}_{\textrm{ebh}}=\emptyset$ and the FDP is zero by convention. Hence suppose $\hat{\jmath}^{*}>0$. Using $\abs{\hat{\mathcal{S}}^{*}_{\textrm{ebh}}} = \hat{\jmath}^{*}$ from Step~2, for each $j \in \mathcal{H}_{0}$ we apply Markov's inequality in the form $\mathbbm{1}\{a \geqslant b\} \leqslant a/b$ (valid for any $a \geqslant 0$ and $b > 0$) with $a=e^{\textrm{avg}}_{j}$ and $b = p/(q_{\textrm{ebh}}\,\abs{\hat{\mathcal{S}}^{*}_{\textrm{ebh}}})$:
	\begin{equation*}
		\mathbbm{1}\bigl\{e^{\textrm{avg}}_{j} \geqslant \tfrac{p}{q_{\textrm{ebh}}\,\abs{\hat{\mathcal{S}}^{*}_{\textrm{ebh}}}}\bigr\} \leqslant \frac{q_{\textrm{ebh}}\,\abs{\hat{\mathcal{S}}^{*}_{\textrm{ebh}}}}{p}\, e^{\textrm{avg}}_{j}.
	\end{equation*}
	Every rejected candidate satisfies $e^{\textrm{avg}}_{j} \geqslant p/(q_{\textrm{ebh}}\,\abs{\hat{\mathcal{S}}^{*}_{\textrm{ebh}}})$ by the self-consistency established in Step~2, so the rejected nulls form a subset of the nulls exceeding the threshold, $\{j \in \mathcal{H}_{0} : j \in \hat{\mathcal{S}}^{*}_{\textrm{ebh}}\} \subseteq \{j \in \mathcal{H}_{0} : e^{\textrm{avg}}_{j} \geqslant p/(q_{\textrm{ebh}}\,\abs{\hat{\mathcal{S}}^{*}_{\textrm{ebh}}})\}$; the inclusion may be strict in the boundary event of ties, where a null attaining the threshold value is excluded by the tie-breaking rule. Summing over $j \in \mathcal{H}_{0}$ and taking expectations therefore yields
	\begin{align*}
		\textrm{FDR}(\hat{\mathcal{S}}^{*}_{\textrm{ebh}}) &= \mathbb{E}\!\left[\frac{\#\{j \in \hat{\mathcal{S}}^{*}_{\textrm{ebh}} :\, \xi_{j} = 0\}}{\abs{\hat{\mathcal{S}}^{*}_{\textrm{ebh}}}}\right] \leqslant \mathbb{E}\!\left[\sum_{j \in \mathcal{H}_{0}} \frac{\mathbbm{1}\bigl\{e^{\textrm{avg}}_{j} \geqslant \tfrac{p}{q_{\textrm{ebh}}\,\abs{\hat{\mathcal{S}}^{*}_{\textrm{ebh}}}}\bigr\}}{\abs{\hat{\mathcal{S}}^{*}_{\textrm{ebh}}}}\right]\\[4pt]
		&\leqslant \frac{q_{\textrm{ebh}}}{p}\sum_{j \in \mathcal{H}_{0}}\mathbb{E}\!\left[e^{\textrm{avg}}_{j}\right] \;\leqslant\; \frac{q_{\textrm{ebh}}}{p} \cdot p \;=\; q_{\textrm{ebh}},
	\end{align*}
	where the first inequality uses the set inclusion just established, the second applies Markov's inequality, and the final one applies condition~\eqref{eq:eval_condition}.
\end{proof}

\paragraph{Remark (Why the size constraint preserves FDR control).}
It is tempting to argue that, because $\hat{\mathcal{S}}^{*}_{\textrm{ebh}} \subseteq \hat{\mathcal{S}}_{\textrm{ebh}}$, the constrained procedure is automatically more conservative. This reasoning is unsound: an arbitrary subset of an FDR-controlled rejection set need not itself be FDR-controlled, since discarding true positives shrinks the denominator $\abs{\hat{\mathcal{S}}}$ and can raise the realised FDP. The set inclusion in Step~1 is a by-product, not the justification. What actually preserves control is that the constrained procedure is a \emph{bona fide} e-BH procedure in its own right: capping the search at $j \leqslant s_{\max}$ replaces the unrestricted maximiser $\hat{\jmath}$ by the largest self-consistent index $\hat{\jmath}^{*} \leqslant s_{\max}$, and the self-consistency $e^{\textrm{avg}}_{(\hat{\jmath}^{*})} \geqslant p/(q_{\textrm{ebh}}\,\hat{\jmath}^{*})$ at the realised count---the only property used in Step~3---continues to hold at the (weakly higher) effective threshold. The realised eFDR of the smaller set may be larger or smaller than that of the unconstrained set; what the theorem guarantees is that its \emph{expectation} remains bounded by $q_{\textrm{ebh}}$.

\subsection*{Extension to the single-run knockoff+ filter}

When $K = 1$ and the e-values are constructed at the same level used for selection (i.e.\ $q_{\textrm{ebh}}^{\textrm{base}} = q_{\textrm{ebh}} = q$), the e-BH procedure reduces to the standard knockoff+ filter. The support set is $\hat{\mathcal{S}} = \{j \in [p] : W_{j} \geqslant \hat{\tau}\}$, where $\hat{\tau}$ is the data-driven threshold defined in~\eqref{eq:ko_threshold} with offset~$= 1$. (Under the decoupled construction of the next subsection, with $q_{\textrm{ebh}}^{\textrm{base}} \neq q_{\textrm{ebh}}$, the single-run procedure is no longer literally the knockoff+ filter at level $q_{\textrm{ebh}}$, but it remains a valid e-BH procedure on the knockoff-derived e-values, and the statement below applies verbatim.)

\begin{theorem}[Size-constrained knockoff+ filter controls FDR]\label{thm:smax_knockoff}
	When $K = 1$, the knockoff e-values take the common value $e_{j} = p/(1 + m)$ on $\hat{\mathcal{S}} = \{j \in [p] : W_{j} \geqslant \hat{\tau}\}$ and $0$ elsewhere, where $m = \#\{j \in [p] : W_{j} \leqslant -\hat{\tau}\}$. Let $\hat{\mathcal{S}}^{*}$ be the size-constrained e-BH selection of Theorem~\ref{thm:smax_ebh} applied to these e-values, with $q_{\textrm{ebh}} = q$. Then $\abs{\hat{\mathcal{S}}^{*}} \leqslant s_{\max}$ and $\textrm{FDR}(\hat{\mathcal{S}}^{*}) \leqslant q$. Explicitly,
	\begin{equation*}
		\hat{\mathcal{S}}^{*} =
		\begin{cases}
			\text{the $\min(\abs{\hat{\mathcal{S}}},\, s_{\max})$ candidates of $\hat{\mathcal{S}}$ with the largest $W_{j}$}, & \text{if } \min(\abs{\hat{\mathcal{S}}},\, s_{\max}) \geqslant (1 + m)/q,\\[4pt]
			\emptyset, & \text{otherwise.}
		\end{cases}
	\end{equation*}
\end{theorem}

\begin{proof}
	The single-run knockoff+ filter is equivalent to the e-BH procedure applied to the knockoff e-values $e_{j}$~\cite{ren2024derandomised}, so $\hat{\mathcal{S}}$ coincides with the unconstrained e-BH set and the claim is the $K = 1$ instance of Theorem~\ref{thm:smax_ebh}. The explicit form follows because the nonzero e-values are all equal to $p/(1+m)$, so the self-consistency requirement $e_{(j)} \geqslant p/(q\,j)$ at count $j$ reduces to $j \geqslant (1+m)/q$; the largest admissible count not exceeding $s_{\max}$ is thus $\min(\abs{\hat{\mathcal{S}}}, s_{\max})$ when this value is at least $(1+m)/q$, and there is no self-consistent count $\leqslant s_{\max}$ otherwise.
\end{proof}

\noindent The second branch is the crucial point of precision: when the cap $s_{\max}$ is smaller than $(1+m)/q$, no rejection set of size at most $s_{\max}$ is e-BH self-consistent, and the size-constrained filter selects \emph{nothing} rather than returning a top-$s_{\max}$ set that would violate the guarantee. Naively truncating $\hat{\mathcal{S}}$ to its $s_{\max}$ largest $W_{j}$ in this regime would retain candidates whose e-value $p/(1+m)$ falls below the level $p/(q\,s_{\max})$ required for a size-$s_{\max}$ rejection, breaking self-consistency; see the remark on the practical implementation below.

\subsection*{Extension to BH on aggregated p-values}

When p-values $\hat{p}_{1}, \dots, \hat{p}_{p}$ are obtained via quantile aggregation of empirical knockoff p-values~\cite{nguyen2020aggregation}, the Benjamini--Hochberg procedure selects
\begin{equation*}
	\hat{\mathcal{S}}_{\textrm{bh}} = \left\{j \in [p] : \hat{p}_{j} \leqslant \hat{t}\right\}, \quad \hat{t} = \max\left\{\frac{q\,j}{p} : \hat{p}_{(j)} \leqslant \frac{q\,j}{p}\right\}.
\end{equation*}

\begin{theorem}[Size-constrained BH controls FDR]\label{thm:smax_bh}
	Define the constrained BH procedure through the rejection count
	\begin{equation*}
		\hat{\jmath}^{*}_{\textrm{bh}} = \max\left\{j \in [p] : j \leqslant s_{\max} \;\text{and}\; \hat{p}_{(j)} \leqslant \frac{q\,j}{p}\right\},
	\end{equation*}
	with $\hat{\jmath}^{*}_{\textrm{bh}} = 0$ if this set is empty, rejecting exactly the $\hat{\jmath}^{*}_{\textrm{bh}}$ smallest p-values (ties broken by a fixed rule). Assume the null p-values are super-uniform and satisfy a dependence condition under which the usual BH step-up procedure controls FDR at level $q$ (for example, independence or PRDS). Then $\abs{\hat{\mathcal{S}}^{*}_{\textrm{bh}}} \leqslant s_{\max}$ and $\textrm{FDR}(\hat{\mathcal{S}}^{*}_{\textrm{bh}}) \leqslant q$.
\end{theorem}

\begin{proof}
	The size bound is immediate from $\hat{\jmath}^{*}_{\textrm{bh}} \leqslant s_{\max}$. If $\hat{\jmath}^{*}_{\textrm{bh}}=0$, the FDP is zero by convention. Otherwise, at the realised count $\hat{\jmath}^{*}_{\textrm{bh}}$ the BH self-consistency condition $\hat{p}_{(\hat{\jmath}^{*}_{\textrm{bh}})} \leqslant q\,\hat{\jmath}^{*}_{\textrm{bh}}/p$ holds by definition, so every rejected candidate satisfies $\hat{p}_{j} \leqslant q\,\abs{\hat{\mathcal{S}}^{*}_{\textrm{bh}}}/p$. Thus the constrained rule is a BH-type self-consistent step-up rule whose admissible rejection counts are restricted to $j \leqslant s_{\max}$. Under the stated BH-valid dependence assumptions, the self-consistency argument for step-up procedures applies to this restricted rule and gives $\textrm{FDR}(\hat{\mathcal{S}}^{*}_{\textrm{bh}}) \leqslant q$: under independence this is the textbook BH self-consistency proof, while under PRDS it follows from the self-consistency conditions of Blanchard and Roquain~\cite{blanchard2008two}, which cover step-up procedures more conservative than the unconstrained BH rule. Restricting the search to $j \leqslant s_{\max}$ can only make the p-value cutoff more stringent while preserving self-consistency at the realised count.
\end{proof}

\subsection*{Unified theorem on size-constrained FDR control}

\begin{theorem}[Size-constrained FDR control]\label{thm:smax_unified}
	Let $\hat{\mathcal{S}}$ be the support set produced by any of the following procedures under its corresponding validity assumptions: the knockoff+ filter~\eqref{eq:ko_threshold}, the e-BH procedure on averaged e-values (Algorithm~\ref{alg:adaptive_fdr}), or the BH procedure on aggregated p-values. Define the size-constrained counterpart $\hat{\mathcal{S}}^{*}$ by restricting the step-up search to admissible counts $j \leqslant s_{\max}$, so that it rejects the $\hat{\jmath}^{*} \leqslant s_{\max}$ candidates carrying the strongest evidence against the null at the largest self-consistent count. This theorem applies to this count-constrained step-up rule, not to arbitrary post-hoc truncation of an already selected set. Then:
	\begin{enumerate}
		\item $\hat{\mathcal{S}}^{*} \subseteq \hat{\mathcal{S}}$,
		\item $\abs{\hat{\mathcal{S}}^{*}} \leqslant s_{\max}$,
		\item $\textrm{FDR}(\hat{\mathcal{S}}^{*}) \leqslant q$.
	\end{enumerate}
\end{theorem}

\noindent The mechanism is uniform across all three procedures and rests on two ingredients, not on the set inclusion $\hat{\mathcal{S}}^{*} \subseteq \hat{\mathcal{S}}$. First, a per-candidate bound (Markov's inequality for e-values, or the super-uniformity property for p-values) controls the expected contribution of each null and holds regardless of the cardinality of the support set. Second, the realised rejection count remains \emph{self-consistent}: the candidate at the constrained cutoff still meets the step-up condition at that count, which is the only structural property the FDR argument invokes. Capping the search at $s_{\max}$ replaces the unrestricted maximiser by the largest self-consistent count not exceeding $s_{\max}$, making the rejection cutoff more stringent while preserving self-consistency; the guarantee $\textrm{FDR}(\hat{\mathcal{S}}^{*}) \leqslant q$ then follows under the corresponding validity assumptions above. We stress that this does not follow from set inclusion alone: an arbitrary subset of an FDR-controlled set need not be FDR-controlled, because removing true positives can inflate the realised proportion. This justifies the $s_{\max}$ constraint employed in Algorithm~\ref{alg:adaptive_fdr}.

\paragraph{Remark (Practical implementation: the size cap and post-hoc truncation).}
The theorems above analyse the cap as a \emph{restriction of the rejection rule}: the search for the rejection count is confined to $j \leqslant s_{\max}$, so that the realised count $\hat{\jmath}^{*}$ is the largest self-consistent index not exceeding $s_{\max}$. Our practical implementation realises the cap slightly differently: it first computes the unconstrained e-BH selection $\hat{\mathcal{S}}_{\textrm{ebh}}$ at level $q_{\textrm{ebh}}$ and, only if $\abs{\hat{\mathcal{S}}_{\textrm{ebh}}} > s_{\max}$, truncates it to the $s_{\max}$ candidates with the largest averaged e-values (ties broken by $\bar{W}_{j}$, then by index). The two coincide whenever the cap does not bind, i.e.\ whenever $\abs{\hat{\mathcal{S}}_{\textrm{ebh}}} \leqslant s_{\max}$, in which case no truncation occurs and both return $\hat{\mathcal{S}}_{\textrm{ebh}}$. In the regime in which the pipeline operates---a generous initialisation $q_{0}$ together with $s_{\max}$ set above the anticipated support size, so that the adaptive search of Algorithm~\ref{alg:adaptive_fdr} keeps $s_{\min} \leqslant \abs{\hat{\mathcal{S}}_{\textrm{ebh}}} \leqslant s_{\max}$---the cap is slack, the truncation does not bind, and the deployed per-estimator selection is exactly the FDR-controlled set of Theorem~\ref{thm:smax_ebh}. The two procedures can differ only when the truncation binds and the retained cutoff fails self-consistency, that is when $e^{\textrm{avg}}_{(s_{\max})} < p/(q_{\textrm{ebh}}\,s_{\max})$; the truncated set of size $s_{\max}$ is then \emph{not} a self-consistent e-BH rejection and is not covered by Theorem~\ref{thm:smax_ebh}. To retain the finite-sample guarantee unconditionally---including in this boundary regime---it suffices to apply the cap within the rejection rule as analysed here, equivalently to discard from any truncated set those candidates whose averaged e-value falls below $p/(q_{\textrm{ebh}}\,\abs{\hat{\mathcal{S}}^{*}})$ evaluated at the final count, applied iteratively until the count stabilises; this reduces to the count-based rule of Theorem~\ref{thm:smax_ebh} and never weakens the guarantee. Our implementation exposes this count-constrained variant directly, so that the guarantee of Theorem~\ref{thm:smax_ebh} holds unconditionally when it is used.

\subsection*{Decoupling e-value construction from the selection FDR}\label{app:ebh_fdr}

In Algorithm~\ref{alg:adaptive_fdr}, the per-run knockoff threshold $\hat{\tau}^{k}$ is computed at the current adaptive FDR level $q_{\textrm{ebh}}$, and the same $q_{\textrm{ebh}}$ is subsequently used in the e-BH procedure. Because the e-values
\begin{equation*}
	e^{k}_{j} = p \cdot \frac{\mathbbm{1}\{W^{k}_{j} \geqslant \hat{\tau}^{k}(q_{\textrm{ebh}})\}}{1 + \#\{j^{\prime} \in [p] : W^{k}_{j^{\prime}} \leqslant -\hat{\tau}^{k}(q_{\textrm{ebh}})\}}
\end{equation*}
depend on $q_{\textrm{ebh}}$ through $\hat{\tau}^{k}$, varying $q_{\textrm{ebh}}$ changes both the e-values and the selection threshold simultaneously. This coupling can break the monotonicity of the rejection set: a stricter (smaller) $q_{\textrm{ebh}}$ raises $\hat{\tau}^{k}$, which reduces the denominator $1 + \#\{j^{\prime}: W^{k}_{j^{\prime}} \leqslant -\hat{\tau}^{k}\}$, thereby inflating the e-values of surviving candidates. The inflated e-values may then exceed the e-BH selection threshold for candidates that were not selected at a more liberal level, producing the counter-intuitive outcome $\hat{\mathcal{S}}_{\textrm{ebh}}(q') \not\subseteq \hat{\mathcal{S}}_{\textrm{ebh}}(q)$ for $q' < q$.

To restore monotonicity, we decouple the two roles by introducing a base FDR level $q_{\textrm{ebh}}^{\textrm{base}}$ that governs e-value construction, while the selection FDR $q_{\textrm{ebh}}$ controls only the step-up threshold. Concretely, we fix
\begin{equation}
	\hat{\tau}^{k} \;=\; \hat{\tau}^{k}\!\left(q_{\textrm{ebh}}^{\textrm{base}}\right), \qquad e^{k}_{j}\!\left(q_{\textrm{ebh}}^{\textrm{base}}\right) = p \cdot \frac{\mathbbm{1}\{W^{k}_{j} \geqslant \hat{\tau}^{k}(q_{\textrm{ebh}}^{\textrm{base}})\}}{1 + \#\{j^{\prime} \in [p] : W^{k}_{j^{\prime}} \leqslant -\hat{\tau}^{k}(q_{\textrm{ebh}}^{\textrm{base}})\}},
	\label{eq:decoupled_eval}
\end{equation}
so that $e^{k}_{j}$ no longer depends on the selection level, and apply the e-BH procedure at target level $q_{\textrm{ebh}}$ to the averaged e-values $e^{\textrm{avg}}_{j}(q_{\textrm{ebh}}^{\textrm{base}}) = \frac{1}{K}\sum_{k} e^{k}_{j}(q_{\textrm{ebh}}^{\textrm{base}})$.

\begin{theorem}[Decoupled e-BH preserves FDR control and yields monotone rejections]\label{thm:decoupled_ebh}
	Let $q_{\textrm{ebh}}^{\textrm{base}} \in (0,1]$ be fixed and let the e-values be constructed as in~\eqref{eq:decoupled_eval}. Then:
	\begin{enumerate}
		\item \textbf{Validity.} The condition $\sum_{j \in \mathcal{H}_{0}} \mathbb{E}\!\left[e^{\textrm{avg}}_{j}(q_{\textrm{ebh}}^{\textrm{base}})\right] \leqslant p$ holds for every $q_{\textrm{ebh}}^{\textrm{base}} \in (0,1]$, so the e-BH procedure at any fixed selection level $q_{\textrm{ebh}} \in (0,1]$ satisfies $\textrm{FDR}\!\left(\hat{\mathcal{S}}_{\textrm{ebh}}\right) \leqslant q_{\textrm{ebh}}$.
		\item \textbf{Monotonicity.} For $q' < q$, the decoupled procedure satisfies $\hat{\mathcal{S}}_{\textrm{ebh}}(q') \subseteq \hat{\mathcal{S}}_{\textrm{ebh}}(q)$; that is, a stricter selection FDR can only shrink the support set.
	\end{enumerate}
\end{theorem}

\begin{proof}
	\textit{Part~1 (Validity).}
	The coin-flip property of null knockoff statistics ensures that, for any threshold $\hat{\tau}^{k}$ that is a valid stopping time with respect to the ordered null-sign filtration, the ratio $\#\{j \in \mathcal{H}_{0}: W^{k}_{j} \geqslant \hat{\tau}^{k}\}\big/\bigl(1 + \#\{j \in \mathcal{H}_{0}: W^{k}_{j} \leqslant -\hat{\tau}^{k}\}\bigr)$ has expectation at most~$1$. Since $\hat{\tau}^{k}(q_{\textrm{ebh}}^{\textrm{base}})$ is such a stopping time for any fixed $q_{\textrm{ebh}}^{\textrm{base}}$, the derivation in Section~2 carries through unchanged:
	\begin{equation*}
		\sum_{j \in \mathcal{H}_{0}} \mathbb{E}\!\left[e^{k}_{j}(q_{\textrm{ebh}}^{\textrm{base}})\right] \;\leqslant\; p \quad \text{for each } k,
	\end{equation*}
	and averaging over $K$ runs preserves the bound. The FDR guarantee then follows from the standard e-BH proof (Theorem~\ref{thm:smax_ebh} with $s_{\max} = p$).
	
	\textit{Part~2 (Monotonicity).}
	Because the e-values $e^{\textrm{avg}}_{j}(q_{\textrm{ebh}}^{\textrm{base}})$ are now fixed with respect to $q_{\textrm{ebh}}$, the e-BH selection threshold $p/(q_{\textrm{ebh}}\,\hat{\jmath})$ is monotonically non-increasing in $q_{\textrm{ebh}}$: decreasing $q_{\textrm{ebh}}$ raises the threshold $p/(q_{\textrm{ebh}}\,j)$ for every $j$, which can only reduce $\hat{\jmath}$ (further raising the threshold) and hence shrink $\hat{\mathcal{S}}_{\textrm{ebh}}$. Formally, for $q' < q$,
	\begin{equation*}
		\frac{p}{q'\,j} \;>\; \frac{p}{q\,j} \quad \text{for all } j \geqslant 1,
	\end{equation*}
	so $\hat{\jmath}(q') \leqslant \hat{\jmath}(q)$, and consequently $\hat{\mathcal{S}}_{\textrm{ebh}}(q') \subseteq \hat{\mathcal{S}}_{\textrm{ebh}}(q)$.
\end{proof}

\paragraph{Remark (Practical implementation in Algorithm~\ref{alg:adaptive_fdr}).}
In our implementation, we set $q_{\textrm{ebh}}^{\textrm{base}} = q_{0}$ (the initial FDR level supplied to Algorithm~\ref{alg:adaptive_fdr}), so that the e-values are computed once and held fixed throughout the adaptive search. As Algorithm~\ref{alg:adaptive_fdr} increments $q_{\textrm{ebh}}$ from $q_{0}$ upward (towards $q_{\max}$), only the e-BH selection threshold changes. This decoupling guarantees that, for any two target levels $q' < q$ evaluated during the search, the support set at $q'$ is a subset of that at $q$, thereby ensuring the intuitive property that a more stringent target FDR yields a smaller or equal support set. FDR control is preserved at every fixed $q_{\textrm{ebh}}$ by Part~1 of Theorem~\ref{thm:decoupled_ebh}. If the final level is chosen adaptively from the data during the search, the realised $\hat{q}_{\textrm{ebh}} \leqslant q_{\max}$ is data-dependent and the clean finite-sample bound is stated at the pre-specified ceiling $q_{\max}$, because the selected set is self-consistent at some $\hat{q}_{\textrm{ebh}} \leqslant q_{\max}$ and hence also satisfies the e-BH self-consistency inequality at the fixed level $q_{\max}$. The size constraint $\abs{\hat{\mathcal{S}}} \leqslant s_{\max}$ from Theorem~\ref{thm:smax_ebh} applies unchanged, since the e-values still satisfy condition~\eqref{eq:eval_condition} irrespective of $q_{\textrm{ebh}}^{\textrm{base}}$.

\paragraph{Remark (Non-monotonicity of the coupled procedure).}
When $q_{\textrm{ebh}}^{\textrm{base}} = q_{\textrm{ebh}}$ (the original coupled formulation), the knockoff threshold $\hat{\tau}^{k}$ varies with $q_{\textrm{ebh}}$, and a stricter $q_{\textrm{ebh}}$ raises $\hat{\tau}^{k}$. This simultaneously (i)~zeros out e-values of borderline candidates and (ii)~reduces the denominator $1 + \#\{j': W^{k}_{j'} \leqslant -\hat{\tau}^{k}\}$, inflating the e-values of surviving candidates. The net effect is that a stricter $q_{\textrm{ebh}}$ can select more candidates than a liberal one. This does not violate FDR control when $q_{\textrm{ebh}}$ is fixed in advance---condition~\eqref{eq:eval_condition} holds for the corresponding valid stopping-time threshold $\hat{\tau}^{k}$---but it can be undesirable in practice when one expects monotonically decreasing support sizes. The decoupled formulation eliminates this artefact while retaining the same finite-sample FDR guarantee.


\section{Computational complexity of knockoff feature statistics}\label{app:complexity}

Section~\ref{sec:exp_fdr} compares three knockoff feature statistics---SHAP-DS, SWAP, and SWAP-INT---empirically. Here we derive their per-realisation computational costs and verify the analysis against the wall times in Table~\ref{tab:stat_comparison}.

\subsection*{Definitions}

Recall that the augmented matrix is $\bs{X} = \begin{pmatrix} \bs{\Phi} & \bs{\tilde{\Phi}} \end{pmatrix} \in \mathbb{R}^{n \times 2p}$ and the abess estimator produces coefficients $\hat{\bs{\beta}} \in \mathbb{R}^{2p}$ with $\norm{\hat{\bs{\beta}}}_{0} \leqslant s_{\max}$. All three statistics share the same knockoff generation cost $T_{\mathrm{ko}}$ and estimator fitting cost $T_{\mathrm{abess}}$; only the statistic-computation cost $T_{\mathrm{stat}}$ differs.

\paragraph{SHAP-DS (this work).}
Recall from Section~\ref{sec:ko} that, for the linear \texttt{abess} model, the interventional SHAP value of feature~$j$ on sample~$i$ is $\varphi_{j}(i) = \hat{\beta}_{j}(X_{ij} - \bar{X}_{j})$, where $\bar{X}_{j} = \frac{1}{n}\sum_{i} X_{ij}$. The same arithmetic is performed by the SHAP library's \texttt{LinearExplainer}; for fair, reproducible parallel benchmarking against the SWAP and SWAP-INT statistics evaluated below, our default implementation reproduces it directly via three vectorised NumPy operations on the $(n \times 2p)$ augmented matrix,
\begin{equation*}
	\bar{\bs{X}} \gets \mathrm{mean}(\bs{X}, \mathrm{axis}{=}0),\qquad \bs{\Psi} \gets \hat{\bs{\beta}} \odot (\bs{X} - \bar{\bs{X}}),\qquad Z_{j} \gets \mathrm{mean}(\abs{\bs{\Psi}_{:,j}}),
\end{equation*}
yielding numerically identical $Z_{j}$ values to the library implementation while bypassing the per-call object-construction overhead (background-data summarisation, masker initialisation and input validation) that would otherwise dominate the wall-time measurement when $K = 100$ realisations are evaluated in parallel.

\paragraph{SWAP~\cite{gimenez2019knockoffs}.}
Each feature $j \in [2p]$ is evaluated by replacing its column in $\bs{X}$ with the corresponding knockoff partner and measuring the increase in prediction loss. Concretely, for each $j$, the implementation constructs $\bs{X}^{(j)}$ by setting $\bs{X}^{(j)}_{:,j} \gets \bs{X}_{:,\sigma(j)}$ (where $\sigma$ maps $j$ to its knockoff partner) and computes $Z^{\,\mathrm{SWAP}}_{j} = \mathrm{MSE}(g(\bs{X}^{(j)}), \bs{y})$, where $g$ is the fitted linear predictor. Each per-feature evaluation requires (i) a copy of the $n \times 2p$ matrix, (ii) a single column overwrite, and (iii) one prediction $g(\bs{X}^{(j)}) = \bs{X}^{(j)}\hat{\bs{\beta}}$ followed by a residual norm---each step costs $O(np)$.

\paragraph{SWAP-INT~\cite{gimenez2019knockoffs}.}
The swap integral generalises SWAP by replacing the binary substitution with a path of interpolated columns and integrating the resulting loss via trapezoidal quadrature. For each $\lambda$ in a fixed grid $\{\lambda_{1}, \dots, \lambda_{L}\}$ with step size $\Delta\lambda$, the implementation forms $\bs{X}^{(j,\lambda)}$ via the column update $\bs{X}^{(j,\lambda)}_{:,j} \gets \bs{X}_{:,j} + \lambda\bigl(\bs{X}_{:,\sigma(j)} - \bs{X}_{:,j}\bigr)$, evaluates the loss $\ell_{j,\lambda} = \mathrm{MSE}(g(\bs{X}^{(j,\lambda)}), \bs{y})$, and accumulates $Z^{\,\mathrm{SWAP\text{-}INT}}_{j} = \sum_{\ell=1}^{L} \tfrac{\Delta\lambda}{2}(\ell_{j,\lambda_{\ell-1}} + \ell_{j,\lambda_{\ell}})$. We use $L = 10$ evaluation points $\lambda \in \{0.5, 1.0, \dots, 5.0\}$ with $\Delta\lambda = 0.5$, paired with a baseline evaluation at $\lambda = 0$ to seed the recursion. The implementation nests the same $2p$-feature loop inside an outer loop over the $L$ quadrature points.

\subsection*{Per-realisation costs}

\paragraph{SHAP-DS.}
The three vectorised operations each visit the full $(n \times 2p)$ matrix once, giving
\begin{equation*}
	T_{\mathrm{stat}}^{\,\mathrm{SHAP\text{-}DS}} = O(np).
\end{equation*}

\paragraph{SWAP.}
The Python-level loop iterates $2p$ times. Each iteration performs (i) a matrix copy at cost $O(np)$, (ii) a single column overwrite at cost $O(n)$, and (iii) a prediction $\bs{X}^{(j)}\hat{\bs{\beta}}$ plus residual computation at cost $O(np)$. The loop does not skip features with $\hat{\beta}_{j} = 0$. The total is
\begin{equation*}
	T_{\mathrm{stat}}^{\,\mathrm{SWAP}} = 2p \cdot O(np) = O(np^{2}).
\end{equation*}

\paragraph{SWAP-INT.}
The outer loop over $L$ quadrature points multiplies the per-feature work by $L$, giving
\begin{equation*}
	T_{\mathrm{stat}}^{\,\mathrm{SWAP\text{-}INT}} = L \cdot 2p \cdot O(np) = O(L \cdot np^{2}).
\end{equation*}

\paragraph{Estimator fitting cost $T_{\mathrm{abess}}$.}
By Theorem 2 of Zhu et al.~\cite{zhu2020polynomial}, the abess splicing algorithm at a fixed support size has per-iteration cost $O(n p\,s_{\max})$, and the algorithm terminates in $O(s_{\max} \log p \log\log n)$ outer iterations (Theorem 3 ibid.). The dominant cost is therefore
\begin{equation*}
	T_{\mathrm{abess}} = O\!\bigl(n p\,s_{\max}^{2}\,\log p\,\log\log n\bigr).
\end{equation*}
Since $s_{\max} \ll p$ for the overcomplete libraries considered in this work (i.e., $s_{\max} = O(\log p)$ per Theorem 4 ibid.), $T_{\mathrm{abess}}$ scales near-linearly in $p$ for fixed $n$. Consequently, the abess fit and the SHAP-DS statistic together cost $O(np\,s_{\max}^{2}\,\log p\,\log\log n)$ per realisation, asymptotically smaller than the $O(np^{2})$ statistic-computation cost of SWAP.

\subsection*{Parallelisation}

Algorithm~\ref{alg:adaptive_fdr} aggregates $K$ knockoff realisations per covariance estimator using $J$ parallel workers (\texttt{n\_jobs}). All three statistics are parallelised identically: each worker processes one realisation (estimator fit plus statistic computation). The total wall time is therefore
\begin{equation*}
	T_{\mathrm{total}} \;\approx\; \frac{K}{J}\,(T_{\mathrm{ko}} + T_{\mathrm{abess}} + T_{\mathrm{stat}}).
\end{equation*}

Table~\ref{tab:complexity_summary} collects the per-realisation costs.

\begin{table}[htbp]
	\centering
	\caption{Per-realisation computational cost of each knockoff feature statistic. $T_{\mathrm{ko}}$ and $T_{\mathrm{abess}}$ are shared across all three. $L$: number of quadrature points ($L = 10$ by default).}\label{tab:complexity_summary}
	\small
	\begin{tabular}{lcc}
		\hline
		& $T_{\mathrm{stat}}$ (per realisation) & Total wall time \\
		\hline
		SHAP-DS  & $O(np)$       & $\tfrac{K}{J}\,(T_{\mathrm{ko}} + T_{\mathrm{abess}} + O(np))$ \\[3pt]
		SWAP     & $O(np^{2})$   & $\tfrac{K}{J}\,(T_{\mathrm{ko}} + T_{\mathrm{abess}} + O(np^{2}))$ \\[3pt]
		SWAP-INT & $O(Lnp^{2})$  & $\tfrac{K}{J}\,(T_{\mathrm{ko}} + T_{\mathrm{abess}} + O(Lnp^{2}))$ \\
		\hline
	\end{tabular}
\end{table}


\section{Comparison with sparse regression baselines}\label{app:baselines}

We compare KO-PDE-IDENT against three sparse regression methods available in PySINDy~\cite{deSilva2020pysindy, kaptanoglu2022pysindy}: STLSQ, SR3, and SSR. Every method receives the identical denoised weak-form candidate library $\bs{\Phi}$ and response $\bs{y}$ described in Section~\ref{sec:exp_setup}; the only difference is the sparse regression procedure used to recover the governing terms.

\subsection*{Baseline methods}

The three methods and their hyperparameter grids are as follows.
\begin{enumerate}
	\item \textit{STLSQ} (Sequential Thresholded Least Squares)~\cite{brunton2016sindy}: the default SINDy optimiser, alternating between least-squares fitting and hard thresholding. Threshold $\tau \in \{0.01, 0.05, 0.1, 0.5, 1.0\}$; ridge penalty $\lambda_{\textrm{ridge}} \in \{0, 10^{-5}, 10^{-3}\}$.
	\item \textit{SR3} (Sparse Relaxed Regularised Regression)~\cite{zheng2019sr3}: decouples thresholding from regression through a relaxation variable, with $\ell_{0}$ regularisation. Regularisation weight $\lambda \in \{0.001, 0.005, 0.01, 0.05, 0.1, 0.5, 1.0\}$; relaxation coefficient $\nu \in \{1, 10\}$.
	\item \textit{SSR} (Stepwise Sparse Regression)~\cite{boninsegna2018sparse}: backward elimination removing the least important term at each step. Ridge penalty $\lambda_{\textrm{ridge}} \in \{0, 10^{-5}\}$; elimination criterion: coefficient value or model residual.
\end{enumerate}
Coefficients are debiased via OLS on the selected support for all three methods.

\subsection*{Hyperparameter selection}

Hyperparameters are selected without access to the ground-truth governing equation. Every method--hyperparameter combination is evaluated via $3$-fold cross-validated mean squared error (CV-MSE) on the full weak-form library. The configuration with the lowest CV-MSE is selected per method, and the method is then refit on the full dataset. The eFDR, ePOWER and support size are computed against the known ground truth. Coefficient errors ($\%$CE) are omitted from the table because every method can achieve the same $\%$CE through OLS re-estimation on the recovered support; the key distinguishing factor is therefore the structural recovery captured by eFDR and ePOWER.

\subsection*{Empirical results}

Table~\ref{tab:baseline_denoised} presents the results. KO-PDE-IDENT is the only method that achieves eFDR $= 0$ and ePOWER $= 1$ on every benchmark. No baseline attains this jointly across all seven equations.

\begin{table}[htbp]
	\centering
	\caption{Baseline comparison. The candidate library used for each baseline method is computed via the weak formulation after denoising the noisy observations, making it identical to the library used for KO-PDE-IDENT. Hyperparameters were selected via $3$-fold CV-MSE without ground-truth access.}\label{tab:baseline_denoised}
	\small
	\setlength{\tabcolsep}{4pt}
	\begin{tabular}{llccc}
		\hline
		PDE & Method & $\abs{\hat{\mathcal{S}}}$ & eFDR & ePOWER \\
		\hline
		\multirow{4}{*}{Burgers}
		& STLSQ     & 28 & 0.964 & 0.500 \\
		& SR3       & 28 & 0.964 & 0.500 \\
		& SSR       & 33 & 0.939 & 1.000 \\
		& KO-PDE-IDENT & 2  & 0.000 & 1.000 \\
		\hline
		\multirow{4}{*}{KdV}
		& STLSQ     & 42 & 0.952 & 1.000 \\
		& SR3       & 43 & 0.954 & 1.000 \\
		& SSR       & 31 & 0.935 & 1.000 \\
		& KO-PDE-IDENT & 2  & 0.000 & 1.000 \\
		\hline
		\multirow{4}{*}{KS}
		& STLSQ     & 3  & 0.000 & 1.000 \\
		& SR3       & 3  & 0.000 & 1.000 \\
		& SSR       & 33 & 0.909 & 1.000 \\
		& KO-PDE-IDENT & 3  & 0.000 & 1.000 \\
		\hline
		\multirow{4}{*}{RD ($u$)}
		& STLSQ     & 8  & 0.125 & 1.000 \\
		& SR3       & 7  & 0.000 & 1.000 \\
		& SSR       & 7  & 0.000 & 1.000 \\
		& KO-PDE-IDENT & 7  & 0.000 & 1.000 \\
		\hline
		\multirow{4}{*}{RD ($v$)}
		& STLSQ     & 8  & 0.125 & 1.000 \\
		& SR3       & 7  & 0.000 & 1.000 \\
		& SSR       & 7  & 0.000 & 1.000 \\
		& KO-PDE-IDENT & 7  & 0.000 & 1.000 \\
		\hline
		\multirow{4}{*}{GS ($u$)}
		& STLSQ     & 6  & 0.000 & 1.000 \\
		& SR3       & 1  & 0.000 & 0.167 \\
		& SSR       & 15 & 0.600 & 1.000 \\
		& KO-PDE-IDENT & 6  & 0.000 & 1.000 \\
		\hline
		\multirow{4}{*}{GS ($v$)}
		& STLSQ     & 2  & 0.000 & 0.400 \\
		& SR3       & 2  & 0.000 & 0.400 \\
		& SSR       & 9  & 0.444 & 1.000 \\
		& KO-PDE-IDENT & 5  & 0.000 & 1.000 \\
		\hline
	\end{tabular}
\end{table}

The severity of multicollinearity determines the difficulty of the discovery task. On the 1D equations under $50\%$ noise, the candidate libraries are large ($p = 49$ for Burgers and KdV), and CV-MSE systematically selects permissive hyperparameters that admit large numbers of spurious terms: all three baselines select $28$--$43$ terms against a true support of $2$, with eFDR exceeding $0.93$. The KS equation is an exception---STLSQ and SR3 achieve exact recovery---because its three-term governing equation is well-separated from the spurious features in the thresholded coefficient path. On the multi-component systems, the results are mixed. SR3 and SSR achieve exact recovery on both RD equations, reflecting the moderate library dimension ($p = 19$) and lower noise ($10\%$). On GS, however, the baselines diverge sharply: STLSQ recovers the $u$-equation exactly but retains only $2$ of $5$ true terms for $v$ (ePOWER $= 0.40$); SR3 under-selects on both equations (ePOWER $\leqslant 0.40$); and SSR over-selects on both (eFDR $\geqslant 0.44$). The small diffusion coefficients in GS ($D_{u} = 0.02$, $D_{v} = 0.01$) produce weak signals that are easily eliminated by hard thresholding or overwhelmed by correlated spurious terms during backward elimination. KO-PDE-IDENT avoids both pitfalls: the knockoff filter retains all true terms with FDR control, and the MCDM stage resolves the remaining ambiguity.

\section{Supplementary empirical results}\label{app:more_results}

To complement the main empirical study, we report three sets of supplementary experiments that probe the design choices underlying \textsc{KO-PDE-IDENT}. Section~\ref{app:lasso_vs_abess} replaces the best-subset solver with Lasso to examine how the choice of sparse regression estimator influences eFDR and ePOWER. Section~\ref{app:selective_inference} contrasts \textsc{KO-PDE-IDENT} with selective inference, an alternative FDR-controlled framework. Section~\ref{app:conventional_ic} demonstrates that conventional information criteria (AIC and EBIC) tend to select overly complex PDE alternatives, motivating the multi-criteria formulation of Section~\ref{sec:mcdm}.

\subsection{Comparison between \texttt{abess} and Lasso for KO-PDE-IDENT}\label{app:lasso_vs_abess}

To assess the influence of the underlying sparse-regression estimator on the knockoff-filtering stage, we replaced the $\ell_{0}$-constrained \texttt{abess} solver in Algorithm~\ref{alg:adaptive_fdr} with Lasso~\cite{tibshirani1996regression} and repeated the FDR-controlled selection under the same equicorrelated knockoff construction, the same SHAP-DS feature statistic, and the same e-BH aggregation across $K = 100$ realisations. All other settings (covariance estimators, weak-form library, sample size, noise level, target FDR) were held fixed. Table~\ref{tab:lasso_vs_abess} reports the resulting support size $\abs{\hat{\mathcal{S}}}$, eFDR and ePOWER for each PDE.

\begin{table}[ht]
	\centering
	\caption{KO-PDE-IDENT screening results with Lasso replacing \texttt{abess} as the sparse-regression estimator (target FDR $q = 0.50$). The corresponding \texttt{abess} results are reproduced from the equicorrelated--e-BH row of Table~\ref{tab:gen_comparison} for direct comparison.}
	\label{tab:lasso_vs_abess}
	\begin{tabular}{lccc|ccc}
		\toprule
		& \multicolumn{3}{c|}{Lasso} & \multicolumn{3}{c}{\texttt{abess}} \\
		PDE & $\abs{\hat{\mathcal{S}}}$ & eFDR & ePOWER & $\abs{\hat{\mathcal{S}}}$ & eFDR & ePOWER \\
		\midrule
		Burgers     & 7  & 0.71 & 1.00 & 8  & 0.75 & 1.00 \\
		KdV         & 8  & 0.75 & 1.00 & 7  & 0.71 & 1.00 \\
		KS          & 3  & 0.00 & 1.00 & 3  & 0.00 & 1.00 \\
		RD ($u$)    & 10 & 0.30 & 1.00 & 11 & 0.36 & 1.00 \\
		RD ($v$)    & 10 & 0.30 & 1.00 & 11 & 0.36 & 1.00 \\
		GS ($u$)    & 10 & 0.40 & 1.00 & 6  & 0.00 & 1.00 \\
		GS ($v$)    & 10 & 0.50 & 1.00 & 5  & 0.00 & 1.00 \\
		\bottomrule
	\end{tabular}
\end{table}

Both estimators maintain ePOWER $= 1$ across every PDE, indicating that overcompleteness (Definition~\ref{def:overcomplete}) holds under either choice of sparse regressor. The eFDR values, however, differ markedly. On the high-noise 1D Burgers and KdV equations and the moderate-noise RD system, the two estimators behave comparably: support sizes differ by at most one term and eFDRs by at most $0.06$, with Lasso producing a marginally smaller support on Burgers and RD ($\abs{\hat{\mathcal{S}}} = 7$ versus $8$ for Burgers, $10$ versus $11$ for each RD equation) and a marginally larger one on KdV. On KS, both estimators agree on the exact support ($\abs{\hat{\mathcal{S}}} = 3$, eFDR $= 0$). The clearest contrast emerges on the multi-component GS system: Lasso selects $10$ terms in both equations with eFDR $\geqslant 0.40$, whereas \texttt{abess} achieves exact recovery (eFDR $= 0$, $\abs{\hat{\mathcal{S}}} = 6$ and $5$ for the $u$- and $v$-equations). This gap reflects the role of the $\ell_{0}$ constraint: best-subset selection enforces a hard sparsity budget that prevents the inclusion of small but spurious coefficients, whereas $\ell_{1}$ shrinkage retains many such coefficients in the candidate library. The overall pattern is consistent with our methodological choice: Lasso preserves the favourable power properties associated with overcompleteness, and---in cases where the two diverge---\texttt{abess} produces an equal or stricter eFDR, yielding a more parsimonious candidate set entering the downstream RFE and MCDM stages.

\subsection{Comparison with selective inference}\label{app:selective_inference}

Selective inference~\cite{lee2016exact,taylor2015statistical,tibshirani2016exact} is the principal alternative to knockoff filtering for FDR-controlled variable selection: it computes valid post-selection $p$-values and confidence intervals (CIs) for terms entering a forward-stepwise path, conditional on the selection event. For each benchmark PDE, the candidate library $\bs{\Phi}$ and response vector $\bs{y}$ are constructed under the same setup as in Section~\ref{sec:exp_setup}. Columns of $\bs{\Phi}$ are scaled to unit norm and $\bs{y}$ is centred. Forward-stepwise regression is then run on $(\bs{\Phi}, \bs{y})$, and exact post-selection inference is computed at significance level $\alpha = 0.10$ (the target miscoverage rate, split as $\alpha/2$ in each tail of the constructed CIs)~\cite{tibshirani2016exact}. We compared two truncation rules for the forward-stepwise path, with results reported in Table~\ref{tab:selective_inference}.
\begin{itemize}
	\item \textit{Active-set rule}: The path is run to completion---i.e., until every candidate has entered---and the support $\hat{\mathcal{S}}$ is set to the active set of selected terms whose post-selection $p$-values fall below $\alpha$.
	\item \textit{BIC stopping rule}: The path is truncated at the step $k^{*}$ that minimises the Bayesian information criterion (BIC) $\textrm{BIC}_{k} = n\ln\hat{\sigma}^{2}_{k} + k\ln n$~\cite{schwarz1978bic}, where $\hat{\sigma}^{2}_{k}$ is the residual variance after $k$ forward steps~\cite{taylor2015statistical}. The support is the corresponding active set at step $k^{*}$.
\end{itemize}

\begin{table}[ht]
	\centering
	\caption{Selective inference results for the seven benchmark PDEs at $\alpha = 0.10$, under the two forward-stepwise truncation rules of Section~\ref{app:selective_inference} (the active-set rule and the BIC stopping rule). eFDR and ePOWER are the realised false-discovery proportion and true-positive rate on a single representative run per PDE.}
	\label{tab:selective_inference}
	\begin{tabular}{lccc|ccc}
		\toprule
		& \multicolumn{3}{c|}{Active-set rule} & \multicolumn{3}{c}{BIC stopping rule} \\
		PDE & $\abs{\hat{\mathcal{S}}}$ & eFDR & ePOWER & $\abs{\hat{\mathcal{S}}}$ & eFDR & ePOWER \\
		\midrule
		Burgers     & 12 & 0.83 & 1.00 & 25 & 0.92 & 1.00 \\
		KdV         & 10 & 0.90 & 0.50 & 14 & 0.86 & 1.00 \\
		KS          & 7  & 0.57 & 1.00 & 8  & 0.63 & 1.00 \\
		RD ($u$)    & 12 & 0.42 & 1.00 & 11 & 0.36 & 1.00 \\
		RD ($v$)    & 11 & 0.36 & 1.00 & 10 & 0.30 & 1.00 \\
		GS ($u$)    & 9  & 0.33 & 1.00 & 7  & 0.14 & 1.00 \\
		GS ($v$)    & 6  & 0.17 & 1.00 & 5  & 0.00 & 1.00 \\
		\bottomrule
	\end{tabular}
\end{table}

Selective inference fails to achieve exact recovery on every benchmark except GS ($v$) under the BIC stopping rule, whereas \textsc{KO-PDE-IDENT} ultimately recovers the exact support of all seven equations (Table~\ref{tab:coefficient_accuracy}); on most benchmarks its eFDR substantially exceeds even that of our screening stage alone (compare with Table~\ref{tab:gen_comparison}). On KdV under the active-set rule, selective inference additionally fails to retain all true governing terms (ePOWER $< 1$). The persistent over-selection reflects the well-documented limitation of post-selection inference when the candidate library is highly correlated: the conditional $p$-values become weakly informative, and the stepwise path admits many spurious terms before the truncation criterion intervenes. By contrast, the MX knockoff filter exploits the joint distribution of the candidate features (estimated through Graphical Lasso or Ledoit--Wolf) to construct null counterparts directly, avoiding the conditioning event that weakens selective inference under multicollinearity. The systematic advantage of \textsc{KO-PDE-IDENT} over selective inference across the benchmark suite confirms the appropriateness of our methodological choice for PDE discovery.

\subsection{Model selection using conventional information criteria}\label{app:conventional_ic}

As argued in the introduction, conventional information criteria are insufficient for selecting the parsimonious governing equation in data-driven PDE discovery~\cite{thanasutives2024ubic}. Here, we provide empirical evidence: when applied to the same PDE alternatives enumerated by best-subset selection over $\hat{\mathcal{S}}_{\mathrm{RFE}}$, both the Akaike information criterion (AIC) and EBIC~\cite{chen2008extended} systematically favour overly complex models over the true parsimonious one, despite EBIC's design for high-dimensional consistency and the existence of robust variants for sparse high-dimensional regression~\cite{gohain2023ebicr}.

\subsection*{Definitions}
For a candidate model with support $\mathcal{S} \subseteq [\abs{\hat{\mathcal{S}}_{\mathrm{RFE}}}]$ and cardinality $s = \abs{\mathcal{S}}$, the maximum-likelihood estimate of the noise variance under the Gaussian linear model is $\hat{\sigma}^{2}_{\mathcal{S}} = \norm{\bs{\Pi}^{\perp}_{\mathcal{S}}\bs{y}}_{2}^{2}/n$, where $\bs{\Pi}^{\perp}_{\mathcal{S}} = \bs{I}_{n} - \bs{\Phi}_{\mathcal{S}}(\bs{\Phi}_{\mathcal{S}}^{\intercal}\bs{\Phi}_{\mathcal{S}})^{-1}\bs{\Phi}_{\mathcal{S}}^{\intercal}$ is the orthogonal projector onto the orthogonal complement of $\mathrm{span}(\bs{\Phi}_{\mathcal{S}})$. AIC~\cite{akaike1974new} and EBIC~\cite{chen2008extended} for this support are
\begin{equation*}
	\mathrm{AIC}(\mathcal{S}) = n \ln \hat{\sigma}^{2}_{\mathcal{S}} + 2s, \qquad
	\mathrm{EBIC}(\mathcal{S}) = n \ln \hat{\sigma}^{2}_{\mathcal{S}} + s \ln n + 2\gamma \ln\!\binom{\abs{\hat{\mathcal{S}}_{\mathrm{RFE}}}}{s},
\end{equation*}
with $\gamma \in [0, 1]$ controlling the additional model-complexity penalty (we set $\gamma = 1$ for high-dimensional consistency). Each PDE alternative is scored, and the model minimising the criterion is selected.

\subsection*{Empirical results}
Figure~\ref{fig:aic_ebic} reports AIC and EBIC scores against support size for the three 1D benchmarks. For Burgers and KdV (both with true support size $c = 2$), AIC and EBIC decrease monotonically with support size and select the largest alternative considered ($\abs{\hat{\mathcal{S}}} = 5$ for Burgers, $\abs{\hat{\mathcal{S}}} = 4$ for KdV), neglecting the parsimonious true models. For KS (true $c = 3$), both criteria correctly identify the exact support, since adding further terms does not yield a sufficient reduction in residual variance to offset the complexity penalty.

\begin{figure}[ht]
	\centering
	\begin{subfigure}[b]{0.32\textwidth}
		\centering
		\includegraphics[width=\textwidth]{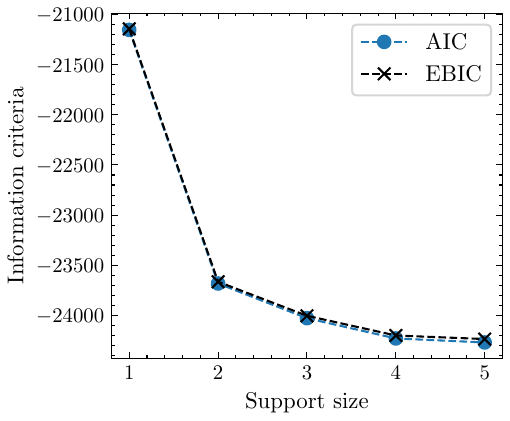}
		\caption{Burgers (true $c = 2$)}
	\end{subfigure}\hfill
	\begin{subfigure}[b]{0.32\textwidth}
		\centering
		\includegraphics[width=\textwidth]{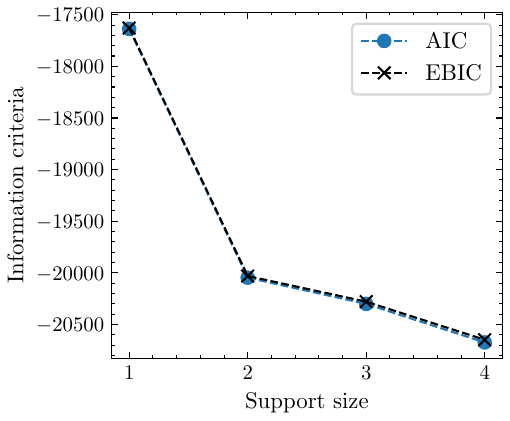}
		\caption{KdV (true $c = 2$)}
	\end{subfigure}\hfill
	\begin{subfigure}[b]{0.32\textwidth}
		\centering
		\includegraphics[width=\textwidth]{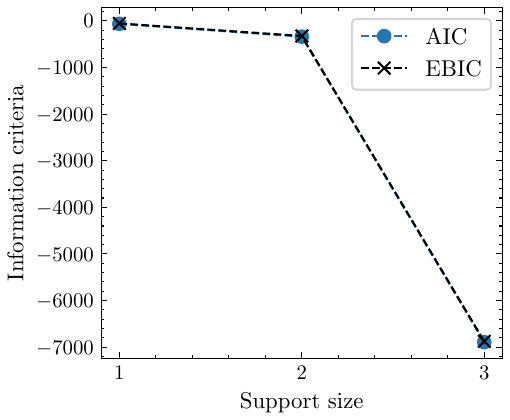}
		\caption{KS (true $c = 3$)}
	\end{subfigure}
	\caption{AIC and EBIC scores against support size on the three 1D benchmarks. Lower is better. On Burgers and KdV, both criteria decrease monotonically and select an overly complex alternative; on KS, they correctly select the parsimonious true support.}
	\label{fig:aic_ebic}
\end{figure}

This behaviour illustrates a known limitation of likelihood-based information criteria under heavy noise and overcomplete candidate libraries: their penalty terms are calibrated for ``true model'' selection under correctly specified linear models, but they treat the gain in fit linearly in the log-likelihood, which can be readily inflated by spurious terms that are mildly correlated with the response under finite samples. The MCDM-based framework introduced in Section~\ref{sec:mcdm} mitigates this pitfall by weighing predictive accuracy against multiple complexity and uncertainty criteria, rather than relying on any single information-theoretic score.


\section{Notation summary}\label{app:notation}

Table~\ref{tab:notation} summarises the principal notation used throughout the paper, grouped by the stage of the pipeline in which each symbol is introduced.

\begin{table}[!ht]
	\centering
	\caption{Summary of the principal notation.}\label{tab:notation}
	\small
	\begin{tabular}{@{}l p{0.60\textwidth}@{}}
		\toprule
		Symbol & Description \\
		\midrule
		\multicolumn{2}{l}{\textbf{Data and regression model}} \\
		$u$, $\bs{\mathrm{U}}$ & state variable and its discretised (noisy) observations \\
		$\bs{\Phi} \in \mathbb{R}^{n \times p}$, $\bs{\phi}_{j}$ & weak-form candidate library; its $j$-th feature \\
		$\bs{\tilde{\Phi}}$, $\tilde{\phi}_{j}$ & knockoff copy of the library; of feature $j$ \\
		$\bs{X} = \begin{pmatrix} \bs{\Phi} & \bs{\tilde{\Phi}} \end{pmatrix}$ & augmented feature matrix \\
		$\bs{y}$ & response (weak-form integral of $u_{t}$) \\
		$n$, $p$, $[p]$ & number of subdomain samples; number of candidates; index set $\{1, \dots, p\}$ \\
		$\bs{\xi}$, $\bs{\eta}$ & true coefficient vector; observation noise \\
		$\mathcal{H}_{0}$, $\mathcal{H}^{*}$ & null set $\{j: \xi_{j} = 0\}$; true support $[p] \setminus \mathcal{H}_{0}$ \\
		$c$ & true support size, $c = \abs{\mathcal{H}^{*}}$ \\
		\midrule
		\multicolumn{2}{l}{\textbf{Knockoff filters}} \\
		$\bs{\Sigma}$, $\mathrm{diag}(\bs{d})$, $\bs{G}$ & feature covariance; knockoff diagonal matrix; joint covariance \\
		$Z_{j}$, $\tilde{Z}_{j}$ & SHAP importance of feature $j$; of its knockoff \\
		$W_{j}$ & knockoff feature statistic (SHAP-DS: $Z_{j} - \tilde{Z}_{j}$) \\
		$\hat{\tau}$, offset & data-driven threshold~\eqref{eq:ko_threshold}; $0$ for knockoff, $1$ for knockoff+ \\
		$q$ & target FDR level \\
		\midrule
		\multicolumn{2}{l}{\textbf{Aggregation via e-values}} \\
		$K$ & number of knockoff realisations \\
		$e^{k}_{j}$, $e^{\textrm{avg}}_{j}$ & per-run e-value; average over $K$ runs \\
		$\hat{\jmath}$, $\hat{\jmath}^{*}$ & e-BH rejection count; size-constrained count \\
		$q_{\textrm{ebh}}$, $q_{\textrm{ebh}}^{\textrm{base}}$, $\hat{q}_{\textrm{ebh}}$ & e-BH selection level; base level for e-value construction; realised tuned level \\
		$q_{0}$, $\Delta q$, $q_{\max}$ & initial level, increment and ceiling of the adaptive search \\
		$s_{\min}$, $s_{\max}$ & minimum and maximum support size \\
		$\mathcal{C}$ & collection of covariance estimators \\
		\midrule
		\multicolumn{2}{l}{\textbf{Support sets}} \\
		(generic) $\hat{\mathcal{S}}$ & selected support set (of the stage under discussion) \\
		$\hat{\mathcal{S}}_{\textrm{ebh}}$, $\hat{\mathcal{S}}^{*}_{\textrm{ebh}}$ & unconstrained and size-constrained e-BH selections \\
		$\hat{\mathcal{S}}_{\textrm{RFE}}$, $\mathcal{S}^{*}$ & support after the RFE stage; final MCDM-selected equation \\
		\midrule
		\multicolumn{2}{l}{\textbf{Evaluation metrics and criteria}} \\
		FDR, FDP, eFDR & expected and realised false discovery proportion; realised value on one instance \\
		POWER, ePOWER & expected and realised true-positive rate \\
		mFDR & modified FDR (for knockoff filters with offset $= 0$) \\
		SC, RICOMP-MM & structural complexity; robust information complexity criterion \\
		CWC, PICP, NMPIL & coverage-width criterion; interval coverage probability; normalised mean interval length \\
		$\%$CE & percentage relative coefficient error \\
		$\alpha$ & significance level (knockoff-perturbed Wilcoxon test; selective inference) \\
		\bottomrule
	\end{tabular}
\end{table}

\end{document}